\documentclass[a4paper, oneside, 11pt]{book} 
\pdfoutput=1
\usepackage[T1]{fontenc}
\usepackage[utf8]{inputenc}
\usepackage[english]{babel}
\usepackage{comment}
\usepackage{minitoc}
\usepackage{mathtools}
\usepackage{amsmath}
\usepackage{booktabs}
\usepackage{amssymb}
\usepackage[perpage]{footmisc}
\usepackage{bbold}
\usepackage{siunitx}
\usepackage{graphicx}
\usepackage{hyperref}
\usepackage{enumitem}
\usepackage{subfigure}
\usepackage{amsthm}
\usepackage{csquotes}
\usepackage[titletoc]{appendix}
\newtheorem{theorem}{Theorem}
\usepackage{hepnames}

\usepackage{xpatch}

\makeatletter
\xpatchcmd\@HepConStyle
 {\edef\@upcode{\updefault}}
 {\ifdefined\shapedefault\edef\@upcode{\shapedefault}\else\edef\@upcode{\updefault}\fi}
 {}{}
\makeatother

\newcommand{\pnue}{\nu_{\mbox{\tiny e}}}
\newcommand{\apnue}{\overline{\nu}_{\!\;\!\mbox{\tiny e}}}
\newcommand{\pnum}{\nu_{\mbox{\tiny $\mu$}}}
\newcommand{\apnum}{\overline{\nu}_{\!\;\!\mbox{\tiny $\mu$}}}
\newcommand{\pnut}{\nu_{\mbox{\tiny $\tau$}}}

\newcommand{\ee}{\mathrm{e}}

\newcommand{\negphantom}[1]{\settowidth{\dimen0}{#1}\hspace*{-\dimen0}}
\newcommand{\rndot}{\dot{|n\hphantom{|}}\negphantom{$|$}\rangle}
\newcommand{\lndot}{\dot{\langle n\hphantom{\rangle}}\negphantom{$\rangle$}|}
\newcommand{\rxdot}[1]{\dot{|#1\hphantom{|}}\negphantom{$|$}\rangle}

\newcommand{\sectionname}{Section}

\usepackage[backend=biber,sorting=none]{biblatex} 

\hypersetup
{
    pdfauthor = {G.\ Fantini, A.\ Gallo Rosso,
    F.\ Vissani, V.\ Zema},
    pdfsubject={Didactic notes},
    pdftitle={The formalism of neutrino oscillations: an introduction},
    pdfcreator={\LaTeX},
    pdfkeywords={Neutrino, oscillation, mixing,
    Quantum Field Theory},
    hidelinks,
    bookmarks=true
}

\title{The formalism of neutrino
oscillations: an introduction}
\author{G.\ Fantini, A.\ Gallo Rosso,
F.\ Vissani, V.\ Zema}
\date{\today}
\allowdisplaybreaks											
\makeindex

\begin{document}
	\dominitoc
	\frontmatter
	\thispagestyle{empty}
\begin{center}
    \centering
    \null\vspace{2cm}
    {\uppercase{\huge The formalism of neutrino\\ oscillations: an introduction\par}}
    \vspace{4cm}
    \includegraphics[width=0.2\linewidth,bb=0 0 993 993]{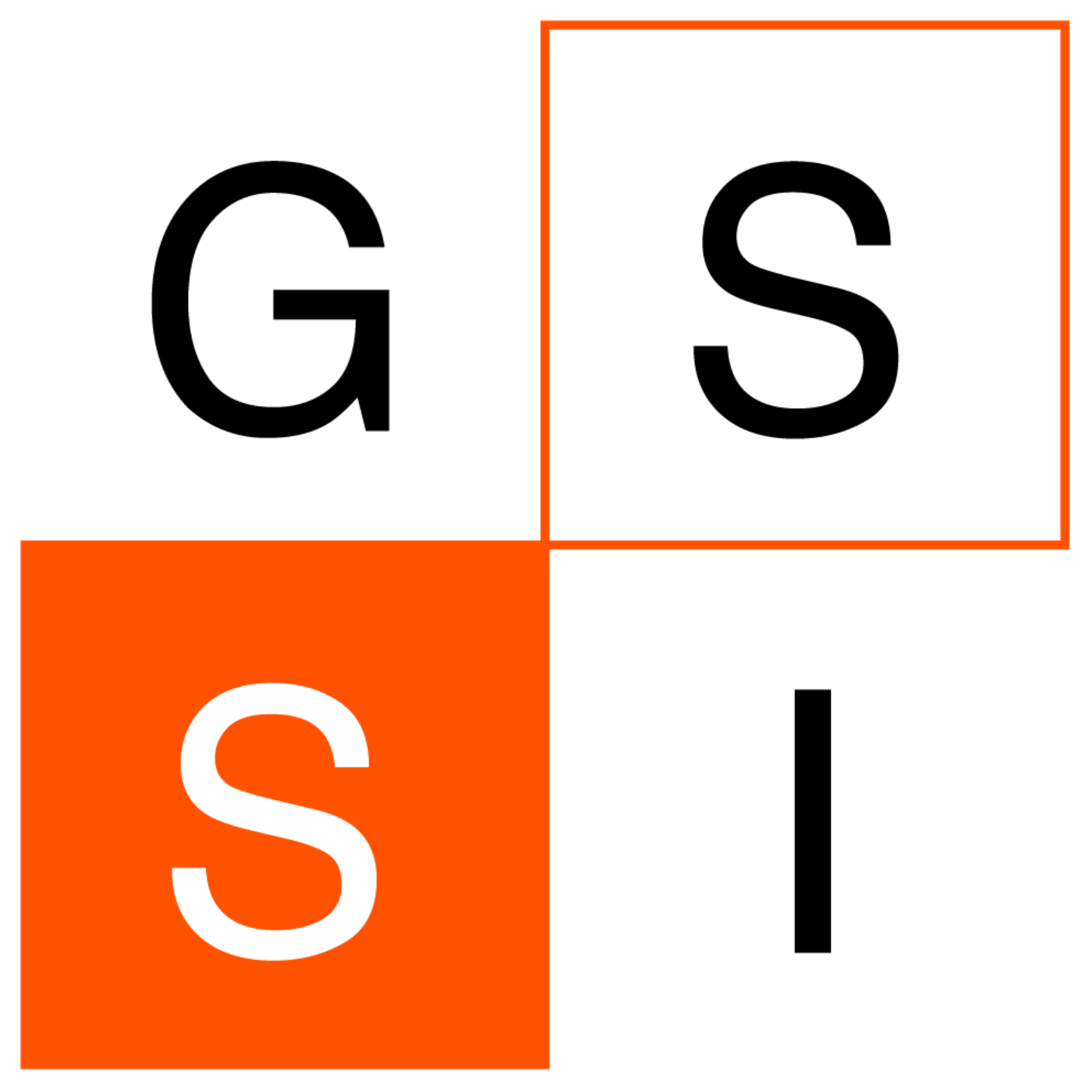}
    \par
    \vspace{4cm}
    {
    \LARGE G.\ Fantini, A.\ Gallo Rosso,
    F.\ Vissani, V.\ Zema
    }\\
    \vspace*{\fill}
    {\Large Gran Sasso Science Institute\\[0.2cm]
    \today}
\end{center}
\clearpage
\thispagestyle{empty}
\hfill
\vfill
\begin{minipage}[b]{0.605\textwidth}
{
	\small\noindent
	{\def\arraystretch{1}\tabcolsep=0em	
	\begin{tabular}{p{3.2cm}r}
		G.\ Fantini\textsuperscript{\dag} &
		{\href{mailto:guido.fantini@gssi.it}
		{guido.fantini@gssi.it}}\\
		{A.\ Gallo Rosso\textsuperscript{\dag\ddag}} &
		{\href{mailto:andrea.gallorosso@gssi.it}
		{andrea.gallorosso@gssi.it}}\\
		{F.\ Vissani\textsuperscript{\dag\ddag}} &
		{\href{francesco.vissani@lngs.infn.it}
		{francesco.vissani@lngs.infn.it}}\\
		{V.\ Zema\textsuperscript{\dag\S}} &
		{\href{vanessa.zema@gssi.it}
		{vanessa.zema@gssi.it}}\\
	\end{tabular}
	}\\[0.5em]
	{\def\arraystretch{1.2}\tabcolsep=0.2em	
	\begin{tabular}{lp{7.2cm}}
		\toprule
		\dag & {Gran Sasso Science Institute\newline
		67100 L'Aquila, Italy.}\\
		\ddag & {INFN -- Laboratori Nazionali del Gran Sasso\newline
		Via G.\ Acitelli, 22 67100 Assergi (AQ), Italy.}\\
		\S & {Chalmers University of Technology,\newline
		Physics Dep., SE-412 96 G\"{o}teborg, Sweden.}\\
		\bottomrule
	\end{tabular}
	}\\[0.6em]
	Preprint of a chapter from:
	Ereditato, A.\ (Ed.) (2018)
	\textit{The State of the
	Art of Neutrino Physics}.
	World Scientific Publishing Company.
}
\end{minipage}

%
\pdfbookmark[0]{Abstract}{Abstract}
\chapter*{Abstract}
\label{sec:abstract}

The recent wide recognition of the existence of neutrino oscillations
concludes the pioneer stage of these studies and poses the problem of
how to communicate effectively the basic aspects of this branch of
science. In fact, the phenomenon of 
neutrino oscillations has peculiar features and requires to master some
specific idea and some amount of formalism. The main aim of these
introductory notes is exactly to cover these aspects, in order
to allow the interested students to appreciate the modern developments
and possibly to begin to do research in neutrino oscillations.

	\chapter*{Preface}
\label{sec:struttura}

The structure of these notes is the following. In the first section, we describe the context of the discussion. Then we will introduce the concept of neutrino mixing and analyze its implications. Next, we will examine the basic formalism of neutrino oscillations, recalling a few interesting applications. Subsequently, we discuss the modifications to  neutrino oscillations that occur when these particles propagate in the matter. Finally, we offer a brief summary of the results and outline the perspectives. Several appendices supplement the discussion and collect various technical details. 

We strive to describe all relevant details of the calculations, in order to allow the Reader to understand  thoroughly and to appreciate the physics of neutrino oscillation. Instead, we do not aim to achieve completeness  and/or to collect the most recent results. We limit the  reference list to a minimum: We cite the seminal papers of this field in the next section,  mention some few books and review papers in the last section, and occasionally make reference to certain works that are needed to learn more or on which we relied to some large extent for an aspect or another. 
These choices are dictated not only by the existence of a huge amount of research work on neutrinos, but also and most simply in view of the introductory character of these notes.

We assume that the Reader knows special relativity and quantum mechanics, and 
some basic aspects of particle physics. 
As a rule we will adopt the system of ``natural units'' of particle physics, defined by the choices
\begin{equation*}
	\hbar=c=1
\end{equation*}
In the equations, the repeated indices are summed, whenever this is not reason of confusion.
Our metric is defined by 
$$
x p=x_\mu p^\mu=x_0\, p_0-\vec{x}\cdot\vec{p}
$$ where 
$x=(x_0,\vec{x}\,)$ and  $p=(p_0,\vec{p}\,)$ are two quadrivectors. Unless stated otherwise, we will use the Dirac (or non-relativistic) representation of the Dirac matrices; see the appendices for technical details.


	\tableofcontents
	
	\mainmatter	
	
\chapter{Introduction}
\minitoc

In this section, the main aspects and 
features of the neutrinos are recalled (Section \ref{sec:ewp}) 
and an introduction to the concept of neutrino oscillations
is offered (Section \ref{sec:no}).
In this manner, the interested Reader can review
the basic concepts and can diagnose or retrieve,
when necessary, the missing information.
This material, along with the appendices, is aimed
to introduce to the discussion of the main content
of this work, exposed in the subsequent three sections.

In view of the introductory character of the present
section, we do not list most of the works of historical
interest. However, there 
are a few papers that should be read by whoever is 
really interested in understanding the roots of the formalism. 
These include  the seminal papers on neutrino oscillations by
Pontecorvo \cite{pt1, pt2}, the one on neutrino mixing by 
Maki, Nakagawa and Sakata \cite{mns}
the papers of Wolfenstein \cite{www} and 
of Mikheyev and Smirnov \cite{ms} on the matter effect.

\section{Overview of neutrinos \label{sec:ewp}} 
We begin with a brief historical outline in 
\sectionname~\ref{sec:ach}, focussed on the basic properties of 
neutrinos and on their characteristic 
interactions, called (charged current) weak interactions or,
formerly, $\beta$ interaction.\footnote{Recall that the term
$\beta$-ray was introduced by Rutherford to describe a type of
nuclear radiation, that we know to be just high energy
electrons or positrons.} 
Then, we offer in \sectionname~\ref{sec:npr}
a slightly more formal overview of some important aspects, 
introducing the hypothesis of non-zero neutrino mass
and showing that neutrino masses 
play a rather peculiar role.  
Finally, we discuss in \sectionname~\ref{sec:rel}
the reasons why neutrinos require us to master the relativistic
formalism and in particular, require a full description of
relativistic spin 1/2 particles --- i.e., Dirac equation.
See the appendix for a reminder of the  main formal aspects
of the Dirac equation, 
and note incidentally that Pauli `invented' the neutrino
just after Dirac's relativistic theory of the electron
(1928) was proposed and before it was fully accepted.

\subsection{A brief history of the major achievements
\label{sec:ach}}
The main aim of this section is just to introduce some concepts and  terms  that are essential for the subsequent discussion; in other words, we use this historical excursion mostly as a convenient  excuse. 
For accurate historical accounts with references, the Reader is invited to consult 
the tables of Ref.~\cite{ponty},
chapter 1 of Ref.~\cite{bily} and Ref.~\cite{rior}.

\subsubsection{Existence of the neutrino}
The first idea of the existence of neutrinos was 
conceived by Pauli in 1930, who imagined them as components of the
nucleus.\footnote{\label{paulos}Before 1930, the prevailing theory
of the nucleus was that it is formed by protons and electrons
tightly bound, e.g., $\mathrm{D}=(2 \Pproton + \Pe)$.

The spins of certain nuclei, as  
$\text{\textsuperscript{6}Li}=(6 \Pproton+3 \Pe)$ or $\text{\textsuperscript{14}N}=(14 \Pproton+7 \Pe)$, 
were predicted to be wrong. Also the $\beta$ decaying nuclei
were predicted to have monochromatic decay spectra, which is,
once again, wrong. 

Pauli improved this model assuming, e.g.,
that $\text{\textsuperscript{6}Li}=(6 \Pproton+3 \Pe
+ 3\nu)$. This assumption was proposed before knowing the existence
of the neutron (funnily enough, Pauli called `neutron' the light
particle $\nu$ that we call today `neutrino').}
The modern theory is due to Fermi (1933), in which 
(anti)neutrinos are \emph{created} in association with $\beta$
rays in certain nuclear decays. From this theory, Wick (1934)
predicted the existence of electron capture; the nuclear recoil
observed by Allen (1942) with $\text{\textsuperscript{7}Be}+
\Pe\to\nu+\text{\textsuperscript{7}Li}$ 
provided evidence of the neutrino.

The first attempt to detect the final states produced by neutrino
interactions was by Davis (1955) following a method outlined by
Pontecorvo (1948). The first successful measurement was by 
Reines and Cowan (1956), using a reaction discussed by Bethe and
Peierls (1934). For this reason, Reines received the Nobel prize
(1995).

\subsubsection{The three families (=copies) of neutrinos}
Pontecorvo argued that the \Pe and \Pmu capture rates are
the same (1947). Then Puppi (1948) suggested the existence of a new
neutrino corresponding to the muon; see also Klein (1948); Tiomno
and Wheeler (1949); Lee, Rosenbluth, Yang (1949).
The fact that the \Pnum is different from the \Pnue
was demonstrated by Lederman, Schwartz, Steinberger in 1962
(Nobel prize in 1988). Evidences of the \Ptau lepton,
a third type of lepton after \Pe and \Pmu, 
were collected since 1974: these are the reasons of the Nobel
awarded to Perl (1995). The corresponding 
tau neutrino \Pnut was first seen by DONUT experiment
(2000), but the number of neutrinos undergoing weak interactions, 
$N_\nu=3$ was known since 1990, thanks to
LEP measurements of the \PZ width.

\subsubsection{Nature of weak interaction and of neutrinos}   
A turning point in the understanding of weak interactions is the
hypothesis that they violate parity, due to 
Lee and Yang (1956) a fact confirmed by the experiment of Wu (1957)
and recognized by the Nobel committee in 1957. This  
was the key to understand the structure of weak interactions 
and it allowed Landau, Lee \& Yang and Salam to conclude that,
for neutrinos, the spin and the momentum have opposite directions 
while, for antineutrinos, the direction is the same one.
One talks also of negative helicity of neutrinos and positive
helicity of antineutrinos.
The final proof of this picture was obtained by the 
impressive experiment of Goldhaber et al.\ (1958). 
Eventually, the theoretical picture was completed 
arguing for an universal vector-minus-axial (V--A) nature 
of the charged-current weak interactions  
(Sudarshan and Marshak, 1958; Feynman and Gell-Mann, 1958).

\subsubsection{Neutrino mixing and oscillations}   
The first idea of neutrino oscillations was introduced 
by Pontecorvo (1957). The limitations of the first proposal
were overcome by the same author, who developed the modern theory
of neutrino transformation in vacuum (1967). 
The new ingredient is the mixing of different families of neutrinos,
introduced by Katayama, Matumoto, Tanaka, Yamada and independently
and more generally by Maki, Nakagawa, Sakata in 1962. The connection
of neutrino mixing with neutrino mass was outlined by 
Nakagawa, Okonogi, Sakata, Toyoda (1963).
Wolfenstein (1978) pointed out a new
effect that concerns neutrinos  propagating in ordinary matter, nowadays
called matter effect; its physical meaning and relevance 
was clarified by Mikheyev and Smirnov (1986). The evidence of
oscillations accumulated from the observation of solar and
atmospheric neutrinos over many years.  
The decisive role of the results of SNO
(Sudbury Neutrino Observatory) and 
Super-Kamiokande as a proof of oscillations was recognized by the 
Nobel committee (2015); however, the number of experiments that
have contributed significantly to this discovery is quite large.

For the above reasons,  
a couple of acronyms are currently used in the physics of neutrinos and in particular in neutrino oscillations:
\begin{enumerate}[itemsep=-0.7ex,partopsep=1ex,parsep=1ex]
	\item PMNS mixing,  after Pontecorvo, Maki, Nagakawa, Sakata to indicate the {\em neutrino (or leptonic) mixing} discussed in the 
next section;
\item MSW effect, after Mikheyev, Smirnov, Wolfenstein to indicate 
the {\em matter effect} described later.
\end{enumerate}

\subsection{Neutrino properties\label{sec:npr}}
	In this section, we offer an introductory discussion of some important neutrino properties. In particular we will discuss the difference between neutrinos and antineutrinos and introduce the  masses of the neutrinos. Although we use the formalism of quantum field theory, we illustrate the results with a pair of pictures that we hope will make the access to the concepts easier.  See also \sectionname~\ref{sec:pera0} and \sectionname~\ref{app:ccm} for more 
	formal details.
	
	\subsubsection{Neutrinos, antineutrinos, their interactions, lepton number}
	When considering an electrically charged particle, say an electron, the difference between this particle and its antiparticle is evident: one has charge $-e$, the other $+e$. What happens when the particle is neutral?  There is no general answer e.g., the photon or the \Ppizero coincide with the their own antiparticle, whereas  the neutron or the neutral kaon \PKzero do not. 
	
	The case in which we are interested is the one of neutrinos.  The charged current weak interactions allow us to  tag  
neutrinos and antineutrinos, thanks to 
the associated charged lepton. In fact, the 
relativistic quantum field theory predicts the existence of 
several processes with the same amplitude; this feature is called {\em crossing symmetry}. A rather important case concerns the six processes listed in \tablename~\ref{tab:processi}.

\begin{table}[h]
	\centering
	{\begin{tabular}{@{}lrclr@{}}
	\toprule
	$\beta^-$ decay 
	& $\Pneutron\to \Pproton \Pelectron \bar\nu$
	&\hphantom{MM}&
	$\beta^+$ decay
	& $\Pproton\to \Pneutron \Ppositron \nu	$ \\
	$\beta^+$ capture
	& $\Pneutron \Ppositron\to \Pproton \bar\nu$ & & 
	$\beta^-$ capture
	& $\Pproton \Pelectron \to \Pneutron  \nu$	\\
	IBD on \Pneutron
	& $\Pneutron \nu \to \Pproton \Pelectron$
	& &
	IBD	& $\Pproton \bar{\nu} \to \Pneutron \Ppositron$	\\
	\bottomrule
\end{tabular}}
\caption{Most common charged-current weak interaction
	processes, characterized by the same amplitude:
	$\beta^\pm$ decay, electron (positron) capture,
	Inverse Beta Decay (IBD).}
\label{tab:processi}
\end{table}

The kinematics of these reactions, however, is not the same; moreover, in some of these cases, the nucleon should be inside a nucleus to trigger the decay and/or the initial lepton should have enough kinetic energy to trigger the reaction. 
The fact that neutrinos and antineutrinos are different means, e.g.,  that the basic neutrinos from the Sun, from $\Pproton\Pproton\to \mathrm{D} \Ppositron \nu$, will never trigger the Inverse Beta Decay (IBD) reaction $\Pproton\bar\nu \to \Pneutron \Ppositron$.
	
	It is easy to see that the above set of reactions is compatible with the conservation of the {\em lepton number}; namely, the net number of leptons (charged or neutral) in the initial and in the final states does not change.

\begin{figure}[t]
\centering
	\includegraphics[width=0.8\textwidth,bb = 0 -1 392 168]{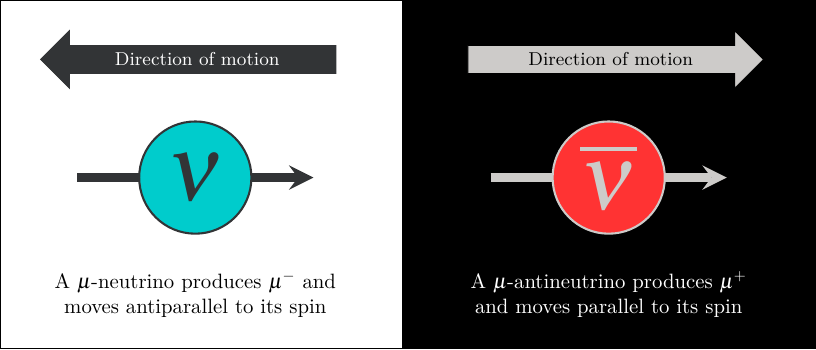}
	\caption{The chiral (or V--A) nature of the charged-current 
		interactions allows us to define what is a 
		neutrino and what is an antineutrino 
		\emph{in the ultrarelativistic limit}, 
		when chirality coincides with helicity and the value of the mass plays only a minor role.}
		\label{fig:hel}
\end{figure}

	In the Fermi theory (or generally in quantum field theory) the leptonic charged current describes the transition from one neutral lepton to a charged one and the other reactions connected by the crossing symmetry.
	The V--A structure of weak interactions means that this current  
	has the form,
	\begin{equation}
	V_\mu^{\mathrm{lept}}-A_\mu^{\mathrm{lept}}=
	\bar{e}\gamma_\mu \nu -\bar{e}\gamma_\mu\gamma_5 \nu=
	2 \bar{e} \gamma_\mu P_{\mathrm{L}} \nu
	\end{equation}
	where we  introduced the chirality projector $P_\mathrm{L}=(1-\gamma_5)/2$. This structure implies that the wave-functions of the neutrinos and of the antineutrinos appear necessarily in the combinations,
		\begin{equation} \label{wuvuf}
		\psi_\nu(x)=e^{-i px} P_{\mathrm{L}}\; u
		\quad\text{and}\quad\psi_{\bar\nu}(x)=e^{-i px}
		P_{\mathrm{R}}\; u
	\end{equation}	
where the 4-spinors $u$ obey the Dirac 
equation and where we have considered plane waves for definiteness, thus  $ p\; x = p_\mu\,x^\mu=E t-\vec{p}\cdot\vec{x}$; 
of course, the energy is $E=\sqrt{p^2+m^2}$ where 
$p=|\vec{p}\,|$ is the momentum and $m$ is the mass.
(The wave-functions  that appear in eq.~\eqref{wuvuf} are proportional to the functions $\psi_i$ 
that appear in the field definition that will be given in \eqref{definella}; their properties concerning the chirality  projectors  
can be derived using the results on charge conjugation described in \sectionname~\ref{app:ccm}.)

In the ultrarelativistic limit $p\gg m$ we have $E\approx p$ and 
 the Dirac hamiltonian that rules the propagation of a massive fermion can be written as,
  \begin{equation}
 H_{\mathrm{D}}=\vec{\alpha}\cdot\vec{p}+\beta m\approx \vec{\alpha}\cdot\vec{p} =
 \vec{\Sigma}\cdot\vec{p}\,  \gamma_5 \approx
 \vec{\Sigma}\cdot\vec{n}\,  \gamma_5 \, E
 \end{equation}
 where $\vec{n}=\vec{p}/p$ is the direction of the momentum. 
 The projection of the spin in 
			the direction of the momentum,
 $	{\vec{\Sigma}\cdot\vec{n}}$,  is called the {\em helicity}.
 Thus,  when the kinetic energy is much larger than the mass, we find that the energy eigenstates given in \eqref{wuvuf} satisfy,
 \begin{equation}
		{\vec{\Sigma}\cdot\vec{n}}\; \psi_\nu(x)\approx  -\psi_\nu(x)
		\quad\text{and}\quad
		{\vec{\Sigma}\cdot\vec{n}} \; \psi_{\bar\nu}(x)\approx  +\psi_{\bar\nu}(x)
			\end{equation}
			Stated in plain words,
we see that, in the ultrarelativistic limit,  
neutrinos have negative helicity whereas antineutrinos have positive helicity. In other terms, we have another way to identify what is a neutrino and what is an antineutrino, as illustrated in \figurename~\ref{fig:hel}.
However, in this manner the question arises: what happens if the mass is non-zero and we invert the direction of the momentum? 
Is it possible to retain a distinction between neutrinos and 
antineutrinos?

\begin{figure}[t]
\centering
\subfigure[Dirac states.]
{
\includegraphics[width=0.3\textwidth,bb=0 0 202 261]{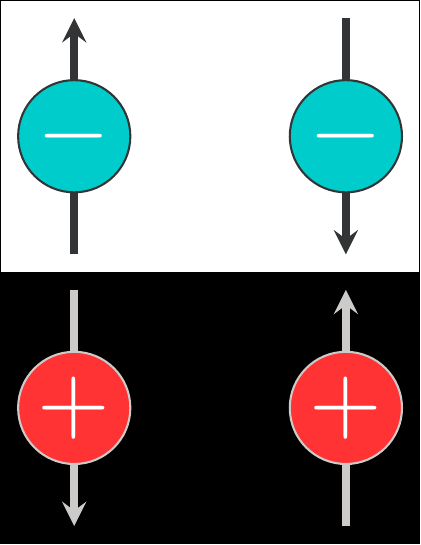}
\label{fig:MassDirac}}
\hspace*{0.15\textwidth}
\subfigure[Majorana states.]{
\includegraphics[width=0.3\textwidth,bb=0 0 202 187]{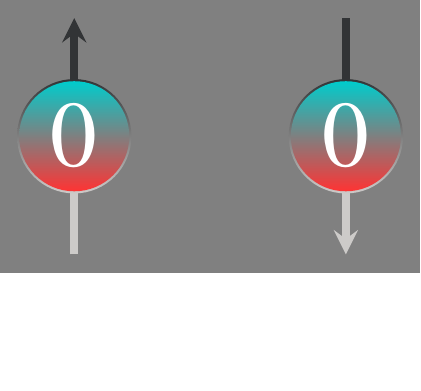}
\label{fig:MassMajo}}
\caption{Massive fields in their rest frames. The arrows show 
		the possible directions of the spin. 
		\protect\subref{fig:MassDirac} shows the 4 states of a Dirac massive field. The signs
		indicate the charge that distinguishes particles and 
		antiparticles, e. g. the electric charge of an electron. 
		\protect\subref{fig:MassMajo} shows the 2 states of a Majorana massive field. 
		The symbol ``zero'' indicates the absence of any  charge (not 
		only of the electric charge): particles and antiparticles 
		coincide.}
\label{fig:DiracMajoDiagr}
\end{figure}
	
	\subsubsection{Neutrino masses}
	The answers to the questions raised just above are not known and to date we have to rely on theoretical considerations. The main hypotheses that are debated are the following two:
\begin{enumerate}[itemsep=-0.7ex,partopsep=1ex,parsep=1ex]
	\item The mass of neutrinos has the same character as the mass of any other charged spin 1/2 particles; in more formal terms, we assume  the same type  of mass originally 
	hypothesized by Dirac for the electron. 
	A closer example is the neutron, that is a neutral particle just as the neutrino.
In a relativistic quantum field theory, this type of mass entails  a strict separation between particle and antiparticle states. More in details,  it means  that in the rest frame there are four distinct states, as for the neutron or the electron, 
namely, 2 spin states for the neutrino and 2 spin states for the antineutrino. This is illustrated in \figurename~\ref{fig:MassDirac}. 
If we accelerate the 2 neutrino states --- depicted in blue  
in \figurename~\ref{fig:MassDirac} ---
at ultrarelativistic velocity and in the direction of the spin, 
one of these state will be allowed to react with the matter through weak interactions, whereas the other state will not; the same is true for antineutrinos.
\item The second hypothesis is the one put forward by Majorana. In this case, there are just 2 spin states in the rest frame; the symmetry under rotations implies that these are two states of the same particle, or in other words, particle and antiparticle coincide. 
 This is illustrated in \figurename~\ref{fig:MassMajo}. This hypothesis can be reconciled with the property of weak interactions, summarized in
 \figurename~\ref{fig:hel}, simply remarking that helicity is {\em not} an invariant quantity  for a massive particle. Therefore, the distinction between neutrinos and antineutrinos is just a feature of ultrarelativistic motions and not a fundamental one. This hypothesis is more economical than the previous one, being based on a smaller number of states, and it is considered plausible by many theorists, for various reasons that we will not examine here in detail.
\end{enumerate}
	
Arguably, 
the question of settling which of these hypotheses is correct is 
the most important open question in neutrino physics to date.
	In principle, it would be possible to observe   
	the difference between the two types of masses in some experiments\footnote{Most plausibly, this can be done  
	through the search of lepton number violating phenomena, such as the decay $(A,Z)\to (A,Z+2)+2 \Pe$, called 
	{\em neutrinoless double beta decay} and discussed elsewhere e.g.\ in \cite{DellOro:2016tmg}.}
	even if we know that the effects  we are searching for experimentally are quite small.  
	
	The difference between the two type of masses is 
	however irrelevant for the description of other important phenomena, including neutrino oscillations. In fact the main experimental evidence arising from the Majorana mass would be the total lepton number violation. On the contrary, neutrino oscillations concern $\nu_\ell \rightarrow \nu_\ell'$ and $\bar{\nu}_\ell \rightarrow \bar{\nu}_\ell'$ transitions, i.e., a transformation from one lepton to another lepton, and not a violation of the total lepton number. This shows that the distinctive feature of Majorana mass is not probed. This was first remarked and proved in Ref.~\cite{sergey} and will be discussed in the following. 
	
\subsection{The major role of relativity for neutrinos\label{sec:rel}}
From the standpoint of the standard model of elementary particle, 
all fermions have a similar status. However,
electrons and neutrinos behave very differently in many situations. One reason of  this difference is just the velocity of the motion. 
The external electrons of an atom 
revolve with a slow velocity $v\sim c\times \alpha$, where  
$\alpha\approx 1/137$ is the fine structure constant:
thus, the role of relativistic considerations is not very central. 
By contrast, relativity is of paramount relevance for neutrinos, 
for several (more or less evident) 
reasons that will be recalled here.

\subsubsection{Smallness of neutrino mass}
Since the beginning it was believed that neutrinos have a small mass
 (Pauli, Fermi, Perrin). Its existence was demonstrated only recently and, despite the fact we do not know its precise value yet, we are sure that it is more than one million times smaller than the mass of the electron. The main experiment that will investigate neutrino mass in laboratory is the KArlsruhe TRItium Neutrino (KATRIN): it aims to study the endpoint region of the tritium decay hoping to improve by 10 times the current upper limit of \SI{2.0}{\electronvolt}, combined results of Mainz and Troitsk experiments. 
In the context of three neutrino oscillations, these results can be compared directly with other ones:  those from SN1987A (or those from pion and tau decay) are (much) weaker than those described above, while 
the most recent combined cosmological analyses claim much tighter limits, that in fact have the potential to discriminate the neutrino mass spectrum.
As discussed above, not only the ``absolute'' value of the neutrino masses, but also the nature of the mass is at present unknown.
 
\subsubsection{Features of the main phenomena of neutrino emission}
Neutrinos are emitted in nuclear transitions where a typical kinetic energy ranges from few keV to some MeV; the lowest energy neutrinos observed by Borexino  (solar neutrinos) 
have few hundreds \si{\kilo\electronvolt}. 
In many high energy processes neutrinos are emitted with much larger 
kinetic energies: the highest energy events attributed to neutrinos are those seen by IceCube 
with energy of few \si{\peta\electronvolt}. In 
all practical cases in which we will be interested here, the kinetic energy is much more than the mass, and 
neutrinos propagate in the 
{\em ultrarelativistic} regime.

\subsubsection{Cross sections growth with energy}
Neutrino interactions, as a rule, increase with their energy $E_\nu$. In fact there is a characteristic constant named after Fermi,  $G_{\mathrm{F}}$, that appears in the amplitude of neutrino interactions. This 
has dimensions of an area (or equivalently an inverse energy squared) thus any  cross section behaves typically as $\sigma \sim  G_{\mathrm{F}}^2 E^2_\nu$ or $G_{\mathrm{F}}^2 m E_\nu$, where $m$ is some characteristic mass. Incidentally, the main feature of neutrinos, namely  
the smallness of their cross sections of interactions (weak interactions) is evident from the numerical 
value of,
\begin{equation}
G_{\mathrm{F}}^2=\SI{5.297e-44}{\centi\meter\squared\per\square\mega\electronvolt}
\end{equation}
This is why these particles have been observed only at relatively high energies.  Moreover, as discussed above  
in most cases (and all cases of interest for the study of oscillations) we can assume that neutrinos are produced in ultrarelativistic conditions.

\subsubsection{Helicity (spin-momentum correlation)}
 In weak interactions ultrarelativistic 
  neutrinos (resp., antineutrinos) have as a rule the spin antialigned
(resp., aligned) with the momentum. This is completely different from what happens to the electrons in atomic physics, where the spin can be both up or down.  
Various aspects of this connection have been 
discussed in \sectionname~\ref{sec:npr}, stressing the assumption $p\gg m$. We add here only one formal remark. When we consider 
 ultrarelativistic motions 
 the wave equations  for neutrinos and antineutrinos 
 can be written 
 as a pair of Weyl equations, 
 namely,
\begin{equation}
i\; \partial_t\; \phi =\mp\, \vec{\sigma}\cdot\vec{p}\; \phi
\end{equation}
where $\sigma_i$ are the Pauli matrices 
and $\vec{p}=-i\vec{\nabla}$. In this formalism, 
the role of helicity is very transparent.
(Formally, this remark  
 concerns the structure of the Dirac equation and becomes evident   
 in the representation of the Dirac matrices in which the chirality
 matrix $\gamma_5$ is diagonal: see \sectionname~\ref{pat} and \ref{frz}.)

\subsubsection{Role of antineutrinos}
Finally, unlike atomic physics where the existence of positrons and anti-baryons can be neglected, in most practical applications neutrinos and antineutrinos  have the same importance:
consider, for instance, big bang or supernovae events that have energies at which the particle-antiparticle production cannot be neglected.
Note that the possibility to create (or to destroy) particles is a specific feature of relativistic phenomena. Indeed, the model of Fermi was proposed since the start as a relativistic quantum field theory.

To summarize, we have described the main reasons why neutrinos are, as a rule, ultrarelativistic and the theory of neutrinos has to be  a relativistic theory. This implies that the formalism to describe neutrinos differs to large extent from the one commonly used for the physics of the electrons in the atoms, i.e., of ordinary matter.  

\section{Introduction to oscillations\label{sec:no}}

As mentioned above, three different 
types of neutrinos (called also flavor) exist.
It is common usage to call neutrino oscillations the observable transformation of a neutrino from one type to another, from the moment when it was produced to the moment when it is detected.\footnote{In a more restrictive and precise sense, 
this term refers to the sinusoidal/cosinusoidal character of some connected phenomena.} 
In his Nobel lecture (1995), Reines depicts this phenomenon by using the vivid allegory of the ``dog which turns into a cat'' which, however, causes a rather disconcerting feeling, being quite far from what we experience in reality.

As remarked immediately by its discoverers, neutrino oscillation is a typical  quantum mechanics phenomenon, that can be easily described resting on the wave nature of neutrinos. 
In order to introduce it as effectively as possible, we consider 3 different physical systems in the following. 
The first is just the propagation of polarized light in a birefringent crystal, the second is the spin states of an electron in a magnetic field, the third one is the neutral kaon system.  All these systems can be considered as sources of precious analogies for us; the last one was originally invoked to introduce neutrino oscillations (Pontecorvo, 1957).

\paragraph{Transformation of the polarization states of the light}

Consider a wave guide (or a transparent crystal) on a table,
oriented along the (horizontal) $x$-axis and assume that 
the two orthogonal directions  
$y$ and $z$ have different refractive indices $n_y$ and $n_z$;
e.g., a birefringent crystal.
Suppose that a plane wave that propagates along the $x$-axis enters the crystal and that it is linearly polarized at $45^\circ$ in the $y-z$ plane, bisecting the 1st-3rd quadrants. 
This polarization is realized when the two oscillating components of the electric field on the $y$ and $z$ axes have a phase difference $\delta = 0$.
 Inside the crystal, the first component propagates according to,
\begin{equation}
 \exp\!\left[ { 2\pi i \left(\frac{x }{\lambda_{y}} -\frac{t}{T}\right)} \right]
 \quad\text{with}\quad
 \lambda_y=\frac{\lambda_0}{n_y}
\end{equation}
and similarly for the $z$-component. Owing to the fact that  the two refractive indices are different, the relative phase between the $y$ and  the $z$ components changes. In particular, if the length of the crystal $L$ satisfies the condition,
\begin{equation}
\delta(L) = 2 \pi (n_y -n_z) \frac{L}{\lambda_0}=\pi
\end{equation} 
the wave will exit with a polarization  
at $135^\circ$, that 
bisects the 2nd-4th quadrants. Thus, the wave 
will be orthogonal to the 
initial wave, the one that entered the crystal. (This device and arrangement is called in optics {\em waveplate}.)

Another description of the same phenomenon 
is as follows.
The definition 
$n={\lambda_0}/{\lambda}$ used above\footnote{We recall that   
$\lambda=2\pi/k$,
$\lambda_0=2\pi c/\omega$ and 
$\omega=2\pi/T$
since the wave behaves as $
 \exp\!\left[ { i \left( k {z } -\omega {t}\right)} \right]$.} 
implies,
\begin{equation}
n=\frac{c}{v_{\Pproton}}
\end{equation}
where  
$v_{\mathrm{p}}=\lambda/T=\omega/k$ is the phase velocity; thus, the  phase velocities are different in the $y$ and $z$ directions. 
We can then say that the components of the wave in the $y$ and $z$ directions propagate with different velocities and this causes the relative change of phase, that modifies the polarization of the propagating wave (and/or of the photons).

An important remark concerns 
the interpretation of this situation in terms of photons,
i.e.\ the impinging photons get transformed  into their orthogonal states.
This conclusion  is quite dramatic indeed and we are entitled to talk of transmutation or transformation of photons.

\paragraph{Transformation of the spin states of an electron}
Consider a region where there is a magnetic field aligned along the $z$ axis with intensity $B$ and an electron whose spin lies in the $x-y$ plane.
It is evident to any physicist that, in this situation, the spin of the electron will remain in the  plane and simply revolve, in a motion of precession around the $z$ axis. 

Let us examine this situation from the point of view of quantum mechanics.
The (matrix elements of the spin) wave-function  that describes 
an electron, polarized in the $x$ direction  
at time $t=0$, is,
\begin{equation}
\psi_{x+}(0)=\frac{1}{\sqrt{2}}
\left(
\begin{array}{c}
1\\
1
\end{array}
\right)
\end{equation}
The (matrix elements of the spin) hamiltonian, that describes the coupling of $\vec{B}$ to the magnetic moment $\vec{\mu}$, is,
\begin{equation}
H=-\vec{B}\cdot\vec{\mu}=
\left(
\begin{array}{cc}
-B \mu & 0 \\
0 & B \mu
\end{array}
\right)
\end{equation}
The state of the electron will be not stationary and 
it is easy to find the solution of the Schr\"{o}dinger equation as,
\begin{equation}
\psi_{x+}(t)=\frac{1}{\sqrt{2}}
\left(
\begin{array}{c}
e^{i \Phi}\\
e^{-i \Phi}
\end{array}
\right)\quad\text{where}\quad\Phi=\frac{B\ \mu\ t}{\hbar}
\end{equation}
This implies for instance that the electron will not remain polarized in the direction of the $x$-axis in the course of the time and in fact, it can eventually turn into an orthogonal state.  
In order to verify this statement, we evaluate the probability to find it in the state with opposite polarization, namely aligned along,
\begin{equation}
\psi_{x-}(0)=\frac{1}{\sqrt{2}}
\left(
\begin{array}{c}
1\\
-1
\end{array}
\right)
\end{equation}
It is straightforward to verify that,
\begin{equation}
P_{x+\to x-}(t)\equiv \left|\: \langle \psi_{x-}(0)  | \psi_{x+}(t) \rangle \: \right|^2=\left| \frac{e^{i \Phi}-e^{-i \Phi}}{2} \right|^2 =
\sin^2\! \left[ \frac{B\ \mu\ t}{\hbar} \right]
\end{equation}
which is non-zero, it is between $0$ and $1$ as any respectable probability,  it has an oscillatory character and in fact it becomes 1 (signaling a full transformation) when $\Phi=\pi/2$.

\paragraph{Transformation of neutral kaons}
Shortly after the discovery of $K$-mesons (kaons) and $\Lambda$-particles, it was realized that a new quantum number, strangeness, is  conserved by strong interactions. The neutral kaon \PKzero  produced e.g., in $\Ppiminus +\Pproton\to \PKzero+\Lambda$, having an internal non-vanishing strangeness quantum number, is different from its own antiparticle \APK. 
Gell-Mann and Pais in 1955 remarked that the common decay channel 
\PKzero or $\APK\to \Ppiplus\Ppiminus$ 
implies necessarily the existence of non-zero transition amplitudes 
$\langle \Ppiplus\Ppiminus | H_{\mathrm{weak}} | \PKzero\rangle$ and $\langle \Ppiplus\Ppiminus | H_{\mathrm{weak}} | \APK\rangle$ at the order $G_\mathrm{F}$. Thus, strangeness is not respected in weak interactions. 
The effective hamiltonian of the two neutral kaons has non-zero  transition  element 
$\langle \PKzero  |  H_{\mathrm{eff}} | \APK\rangle=\delta m -  i\: {\delta \Gamma}/{2}$, namely,  
\begin{equation}
H_{\mathrm{eff}}=
\left(
\begin{array}{cc}
M & \delta m \\
 \delta m^*  & M 
\end{array}
\right)-\frac{i}{2}
\left(
\begin{array}{cc}
\Gamma & \delta \Gamma \\
 \delta \Gamma^*  & \Gamma
\end{array}
\right) \label{bambulo}
\end{equation}
where,
\begin{equation}
	M=\mathcal{O}(1)\quad\text{and}\quad
	\Gamma, \delta m,\delta\Gamma=\mathcal{O}(G_{\mathrm{F}}^2)
\end{equation}
with the non-hermitian part accounting for the weak decays. 
This hamiltonian is non-diagonal, thus the propagation eigenstates differ from the  strangeness eigenstates \PKzero and \APK. These  are indicated with \PKlong and \PKshort, from ``long'' and ``short'' (with reference to their very different lifetimes) and are 
occasionally called  also weak eigenstates.\footnote{In first approximation, they coincide with the CP eigenstates $K_1$ and $K_2$.} 
Even if their mass  difference is so small that it cannot be measured directly, it  entails the occurrence of 
{``virtual transitions 
$\PKzero\leftrightarrows \APK$}'', quoting the words of Gell-Mann and Pais.  More in detail, 
a \PKzero produced at the time $t=0$ is a combination of \PKlong and \PKshort, so 
 it is possible that  it will be 
 detected as a \APK at a subsequent time $t$. This phenomenon is called kaon transformations or transmutation or (with the modern 
language) {\em kaon oscillations}.\footnote{
The first proof was as follows:  a beam of kaons 
originally deriving from $\Pproton+\Pneutron\to \Pproton+\Lambda+\PKzero$
and therefore 
composed by
mesons of positive strangeness, \PKzero, was able upon propagation to produce hyperons of 
negative strangeness, \PSigmaminus.} 

Another interesting behavior of the neutral kaon system 
was predicted in 1955 by Pais and Piccioni: since \PKzero and \APK interact differently with nuclei, the eigenstates of propagation in the ordinary matter are not \PKlong and \PKshort. 
 Thus, when a beam of \PKlong traverses a slab of matter, we will have also \PKshort at the exit. 
This phenomenon is called 
 {\em kaon regeneration}. 

\paragraph{The inception of neutrino oscillations}

After the experimental observation of the kaon transformation  phenomenon, Pontecorvo (1957) asked whether something similar could occur to other systems such as neutrino-antineutrino, neutron-antineutron, atoms-antiatoms. While this does not correspond to the current physical picture of the phenomenon, it is its first specific description in the scientific literature. Later, Maki, Nakagawa and Sakata (1962) mentioned the possible occurrence of {\em virtual transmutation} of neutrinos, again without elaborating the details. This was emended, once again, by Pontecorvo who described the modern  
formalism in 1967. 
It was the phenomenon of kaons regeneration that inspired Emilio Zavattini to ask  Lincoln Wolfenstein about the possible occurrence of something similar in neutrino physics. 
(Note the interesting fact,  the former was an experimentalists and the latter a theorist, just as in the case of Stas Mikheyev and Alexei Smirnov. Similarly, Pontecorvo belonged to a school of physicists where the distinction between theorists and experimentalists was quite vague.)

	\chapter{Leptonic mixing\label{ch:m}}
\minitoc

\section{General considerations\label{lmm}}

\subsection{Definition and context}

\subsubsection{Neutrino flavor states}

The concept of {\em neutrino with given flavor}  identifies the neutral particle associated by weak interactions to the charged leptons with given flavor, namely the electron, the muon or the tau. For what we know to date, the association works in such a manner that the flavor is conserved in the interaction point.  
This allows us to define, e.g., a \Pnue as the particle associated 
to the \APelectron in the $\beta^+$ decay
of \textsuperscript{30}P, or equivalently, the 
state emitted in electron capture processes, say,
$\Pelectron + \text{\textsuperscript{37}Ar} \to \Pnue + \text{\textsuperscript{37}Cl}$.
A similar definition holds true for antineutrinos. 

At the basis of this definition is the assumption, consistent with all known facts,  that charged-current weak-interactions are described by a relativistic quantum field theory, that is Fermi theory at low energies and the standard electroweak model 
--- based on $\mathrm{SU}(3)_\mathrm{C}\times\mathrm{SU}(2)_\mathrm{L}\times\mathrm{U}(1)_{\mathrm{Y}}$ ---
at higher ones. For this reason, many different processes involve to the same type of neutrino, just as in the previous example with \Pnue. Another hidden assumption is that neutrino masses play a negligible role in these interactions, as we will discuss here. 
\paragraph{Neutrino mixing}
 Let us assume that the quantized neutrino fields that have given flavor $\nu_\ell$ 
 {\em do not coincide} with quantized neutrino fields  that have given mass but rather they coincide with linear combinations of fields 
 $\nu_i$ 
 that have given mass $m_i$, namely,
 \begin{equation}
\nu_\ell=\sum_{i=1}^{N}U_{\ell i}\ \nu_i
\quad\text{with}\quad
\begin{cases}
	\ell = \mathrm{e},\mu,\tau & \text{[flavor]}\\
	i=1,2,3 & \text{[mass]}
\end{cases}
\label{mixfi}
\end{equation}
 where $U_{\ell i}$ are the elements of the {\em leptonic mixing matrix}. 
In the following, we will emphasize the case of three light neutrinos, 
$N=3$,  and we will assume that $U_{\ell i}$ are the elements of a $3\times 3$ unitary matrix. These assumptions are consistent:
1.\ with the measured width of \PZzero, that receives a contribution from the 3 light neutrinos;
2.\ with the fact that the neutrinos with given flavors have all the same interactions --- they are {\em universal};
3.\ with cosmological observations. 
The existence of a sizable admixture with other neutrinos would imply new phenomena 
(if they are heavy, it would result into non-unitarity of the $3\times 3$ part of the mixing matrix and it could lead to observable violations of flavor universality) that currently are not observed.  

\subsubsection{Connection with Lagrangian densities} 
The above situation holds true, for instance, if we assume that the 
mass of neutrinos are described by a Dirac  
Lagrangian density,
\begin{equation}
\mathcal{L}_{\mathrm{D}}=- 
\sum_{\ell,\ell'} \left( \overline{\nu_{\ell \mathrm{R}}}
\mathcal{M}_{\ell \ell'}^{\mathrm{D}} \nu_{\ell' \mathrm{L}} +
\overline{\nu_{\ell' \mathrm{L}}} (\mathcal{M}_{\ell \ell'}^{\mathrm{D}})^\dagger  \nu_{\ell \mathrm{R}}\right)
\end{equation}
with,
\begin{equation}
\mathcal{M}^{\mathrm{D}} = V\; \mathrm{diag}(m_1,m_2,m_3)\, U^\dagger
\end{equation}
where $\nu_{\ell' \mathrm{L}}=P_{\mathrm{L}}\nu_{\ell'}$ and $\nu_{\ell \mathrm{R}}=P_{\mathrm{R}}\nu_{\ell}$ are the projections of the neutrino field 
--- later expanded in oscillators, see \eqref{definella}
and \eqref{ridefinella} ---
in the left handed and right handed subspaces, with $P_{\mathrm{R,L}}=(1
\pm\gamma_5)/2$, and $V$ is another mixing matrix, that concerns only right neutrino fields; it can be set equal to $U$ when we are interested only in neutrino oscillations and if we assume that only the known forces are present. The same mixing matrix $U$ stems from 
a Majorana Lagrangian density,
\begin{equation}
\mathcal{L}_{\mathrm{M}}=- 
\frac{1}{2} \sum_{\ell,\ell'} \left( \overline{\nu_{\ell \mathrm{R}}^{\mathrm{c}}}
\mathcal{M}_{\ell \ell'}^{\mathrm{M}} \nu_{\ell' \mathrm{L}} +
\overline{\nu_{\ell' \mathrm{L}}} (\mathcal{M}_{\ell \ell'}^{\mathrm{M}})^\dagger  \nu_{\ell \mathrm{R}}^{\mathrm{c}}\right)
\end{equation}
with,
\begin{equation}
\mathcal{M}^{\mathrm{M}} = U^* \mathrm{diag}(m_1,m_2,m_3) U^\dagger
\end{equation}
where we can define one right-spinor in terms of the left-spinor by using the charge conjugation matrix as follows, 
$\nu^{\mathrm{c}}_{\mathrm{R}}\equiv C \overline{\nu_{\mathrm{L}}}^{\mathrm{t}}$. From this definition,   we derive
$\overline{\nu_R^{\mathrm{c}}}= -\nu^{\mathrm{t}}_{\mathrm{L}} C^{-1}$;
see  \sectionname~\ref{app:ccm} 
for further discussion on the charge conjugation matrix and on its properties.
Let us clarify that the fields $\nu_\ell$ with
$\ell=\ee,\mu,\tau$ that we consider in the rest of our discussion 
are the left chiral fields implied  by ordinary weak interactions and the presence of the chiral projector $P_{\mathrm{L}}$ is not indicated but only to simplify the notations. 

\subsection{Relation between flavor and mass states\label{sec:pera0}}

\subsubsection{Plane waves}
Let us assume that the free neutrinos 
are confined in a cube of volume $V=L^3$ subject to periodicity conditions, so that the momenta are `quantized', i.e., are given in terms of integer numbers $n_x,n_y,n_z$ 
as 
$p_x=n_x\times 2\pi/L $, $p_y=n_y\times 2\pi/L $, $p_z=n_z\times 2\pi/L $. We identify the states fully by means of helicity, 
and use bi-spinors with given helicity, namely, 
$(\vec{\sigma}\cdot\vec{p}/p) \varphi_\lambda=\lambda \varphi_\lambda$, 
where $\lambda=\pm 1$. 
We group the momentum and the helicity in the collective label,
\begin{equation}
\alpha=\vec{p},\lambda
\end{equation}
to shorten a bit the notation. 

In the ``standard''
representation of the Dirac matrices\footnote{Also named 
after Dirac, after Pauli, after both, or also ``non-relativistic'' representation.
Recall that in this representation, the correspondence with ordinary quantum mechanics is more transparent; however, it can be used for any particle, including the ultrarelativistic ones (see \sectionname~\ref{pat}).}  
the {\em plane waves}  of free neutrinos of mass $m_i$ are the eigenvectors 
of the Dirac hamiltonian $H_{\mathrm{D}}$ given by,
\begin{equation}
\psi_i(\vec{x},\alpha)=\frac{\exp( i\, \vec{p}\, \vec{x})}  {\sqrt{2 V}}
\left(
\begin{array}{c}
\sqrt{1+\varepsilon_i}\ \varphi_\lambda \\
\lambda \sqrt{1-\varepsilon_i}\ \varphi_\lambda \\
\end{array}
\right)\quad\text{where}\quad\epsilon_i=\frac{m_i}{E_i}
\label{eq:defepsilon}
\end{equation}
namely $H_{\mathrm{D}} \psi_i=E_i \psi_i$ with 
$E_i =\sqrt{p^2+m_i^2}$. 
We recall incidentally that the bi-spinors are usually grouped into 
four-spinors,
\begin{equation}
u_i(\alpha)=\frac{1}{\sqrt{2}}
\left(
\begin{array}{c}
\sqrt{1+\varepsilon_i}\ \varphi_\lambda \\
\lambda \sqrt{1-\varepsilon_i}\ \varphi_\lambda \\
\end{array}
\right)\label{uuu}
\end{equation}
that are normalized according to,
$u^\dagger u=1$. We use this notation occasionally, see e.g.,  
\eqref{piepolo}.
The single-particle (non-relativistic) normalization conditions hold true:
$\int\mathrm{d}^3 x\; \psi^\dagger \psi=1$ for the plane waves  and $\varphi^\dagger_\lambda \varphi_\lambda=1$ for the bi-dimensional spinors.

In the typical situation when the neutrinos are ultrarelativistic, these  functions take a mass independent universal form,
\begin{equation}
\psi_i(\vec{x},\alpha)\approx \psi(\vec{x},\alpha)
\quad\text{with}\quad
\psi(\vec{x},\alpha)
=\frac{\exp( i\, \vec{p}\, \vec{x})}  {\sqrt{2 V}}
\left(
\begin{array}{c}
\varphi_\lambda \\
\lambda\; \varphi_\lambda \\
\end{array}\label{univrs}
\right)
\end{equation}
The reason why we do not show the dependence of the time is that, immediately below, we introduce quantized fields, given in the interaction   representation.

\subsubsection{Fields and oscillators}
We begin from the quantized field of a neutrino with mass $m_i$,
\begin{equation}\label{definella}
\nu_i(x)=\sum_{\alpha} \left( a_i(t,\alpha)\,  \psi_i(\vec{x},\alpha)+ b_i^\dagger(t,\alpha) \,
\psi_i^{\mathrm{c}}(\vec{x},\alpha)  \right)
\end{equation}
where,
\begin{equation}
\begin{cases}
	a_i(t,\alpha) =a_i(\alpha) e^{-i E_i t}  \\
	b_i(t,\alpha) =b_i(\alpha) e^{-i E_i t}
\end{cases}
\end{equation}
Once again, the index $i$ corresponds to the mass of the neutrino.
Above, we introduced the charge conjugate spinor, 
\begin{equation}
\psi_i^{\mathrm{c}}=C \overline{\psi_i}^{\mathrm{t}}
\end{equation}
see \sectionname~\ref{app:ccm} for a reminder.
 The oscillators $a_i$ and $b_i$  are dimensionless operators that obey the condition 
 \begin{equation}
 \{ a^\dagger(\alpha),  a(\beta)\}=\delta_{\alpha\beta}
 \end{equation}
 namely there is one independent fermionic oscillator each value of $\alpha$.\footnote{In other words, 
 $\delta_{\alpha\alpha'}=\delta_{\vec{p}\vec{p}'} \ 
 \delta_{\lambda \lambda'}$ is a product of Kronecker-deltas;
 recall that in our formalism the momenta are quantized in order to provide a transparent physical interpretation.} 
The above field describes a Dirac neutrino or also a Majorana neutrino, simply replacing $b_i=a_i.$

We are interested in the neutrino field,
\begin{equation}\nu_\ell(x)=U_{\ell i}\ \nu_i(x)
\end{equation} 
where the repeated index $i$ is summed over.
If we consider ultrarelativistic neutrinos $p\gg m_i$, and if we do not 
measure the energy too precisely (so to identify the various mass components), we see from \eqref{univrs} 
that the field with given flavor, that is associated to the corresponding charged lepton, 
can be approximated to,
\begin{equation}
\nu_\ell (x)\approx \sum_{\alpha} \left( a_\ell(t,\alpha)  \,\psi(\vec{x},\alpha)+ b_\ell^\dagger(t,\alpha) \,
\psi^{\mathrm{c}}(\vec{x},\alpha)  \right)
\label{ridefinella}
\end{equation}
with the very important identification,
\begin{equation}
\begin{cases}
	a_\ell(t,\alpha) =U_{\ell i}\ a_i(t, \alpha)   \\
	b_\ell(t,\alpha) =U_{\ell i}^*\ b_i(t,\alpha)
\end{cases}
\label{cirulla}
\end{equation}
It is crucial to note that, in this approximate expression:
\begin{itemize}[itemsep=-0.7ex,partopsep=1ex,parsep=1ex]
\item the plane waves have the same `universal form' $\psi$
valid in the ultrarelativistic limit;
\item the operators with given flavor $\ell$,  namely $
a_\ell$ and $b_\ell$,  
are weighted sums of operators with given mass;
\item the particle ($a_\ell$-type) operators are summed with the matrix $U_{\ell i}$ whereas the antiparticle 
($b_\ell$-type) operators with its conjugate $U_{\ell i}^*$.
\end{itemize}

\subsubsection{Momentum eigenstates}
We can simplify the formulae even further when we consider the fact that the neutrino fields are always multiplied by the 
left chiral projector $P_{\mathrm{L}}=(1-\gamma_5)/2$ in the known weak interactions. For this reason, as it is well known, we will have only 
neutrinos with negative helicity $\lambda=-1$ and antineutrinos with positive helicity $\lambda=+1$ in the ultrarelativistic (UR) limit. 
Therefore, we can write,
\begin{equation}
\begin{cases}
	a_\ell(t,\vec{p},-) =U_{\ell i}\ a_i(\vec{p},-)\  e^{-i E_i t}
	&\text{for UR neutrinos}\\
	b_\ell(t,\vec{p},+) =U_{\ell i}^*\ b_i(\vec{p},+)\  e^{-i E_i t}
	&\text{for UR antineutrinos}
\end{cases}
\label{pera1}
\end{equation}
In this manner, it is possible to define formally 
the states of ultrarelativistic neutrino and antineutrinos as,
\begin{equation}
\begin{cases}
	|\nu_\ell ,\vec{p}\; \rangle\equiv a_\ell^\dagger (0,\vec{p},-) |0\rangle = U_{\ell i}^* |\nu_i,\vec{p}\; \rangle
	&\text{for UR neutrinos}\\
	|\bar\nu_\ell ,\vec{p}\; \rangle\equiv b_\ell^\dagger (0,\vec{p},+) |0\rangle = U_{\ell i} |\nu_i,\vec{p}\; \rangle
	&\text{for UR antineutrinos}
\end{cases}
\label{pera2}
\end{equation}
and also the evolved vectors in the representation of Schr\"{o}dinger, 
\begin{equation}
\begin{cases}
	|\nu_\ell ,\vec{p},t\rangle\equiv a_\ell^\dagger (t,\vec{p},-) |0\rangle = U_{\ell i}^* |\nu_i,\vec{p}\; \rangle e^{-iE_i t}
	&\text{for UR neutrinos}\\
	|\bar\nu_\ell ,\vec{p},t\rangle\equiv b_\ell^\dagger (t,\vec{p},+) |0\rangle = U_{\ell i} |\nu_i,\vec{p}\; \rangle  e^{-iE_i t}
	&\text{for UR antineutrinos}
\end{cases}
\label{pera3}
\end{equation}
A few remarks are in order:
\begin{itemize}[itemsep=-0.7ex,partopsep=1ex,parsep=1ex]
	\item The mixing matrix  enters differently in the relation for neutrinos and antineutrinos states (unless this matrix is real).
	\item It should be noted that the previous 
very important equations, 
derive directly from quantum field theoretical relations.
\item Note that we have omitted the 
helicity labels in the states, or more 
precisely, we have replaced them with an explicit indication of the  
character of the state, either neutrino or antineutrino, which is well defined in the ultrarelativistic limit.
\end{itemize}

\subsubsection{Again on the ultrarelativistic limit}
In order to complete the discussion in \sectionname~\ref{sec:rel} making it more specific, we collect here some final 
  remarks on the ultrarelativistic approximation.
\begin{itemize}
	\item Consider the lowest energy neutrinos that can be detected to date, namely  the solar neutrinos above $\sim \SI{300}{\kilo\electronvolt}$ that can be seen  in Borexino, along with the (very conservative) 
bound of 2 eV on neutrino masses. In this case, the value of the parameter that quantifies the deviation from the UR limit \eqref{eq:defepsilon} is $\epsilon_i < 10^{-5}$.
\item It is possible to  think to various situations when the above assumptions do not apply. For example, very near the endpoint of the $\beta$ spectrum, only the lightest neutrino mass state $\nu_1$ is emitted together with the electron, and not all three of them, simply to conserve energy. However it is not easy to imagine how to form a usable neutrino beam with this composition.
\item More interestingly, the neutrinos produced in the big-bang have now momenta $p\sim kT\approx \SI{0.2}{\milli\electronvolt}$; the two heavier states have masses larger than $\sqrt{|\Delta m^2_{31}|}\approx \SI{50}{\milli\electronvolt}$ and $\sqrt{\Delta m^2_{21}}\approx \SI{8.6}{\milli\electronvolt}$ and are now non-relativistic. However,  a discussion of oscillations is hardly needed, as we can simply treat these states as mass states.
\end{itemize}

Therefore, for all practical applications, we are interested to discuss neutrino oscillations of the ordinary neutrinos only when 
the ultrarelativistic limit applies.


\section{The parameters relevant to oscillations}

\subsection{General considerations}
\subsubsection{Number of parameters in a unitary matrix}
The number of free parameters of a $n\times n$ mixing matrix $U$ satisfying unitarity $U U^\dagger=U^\dagger U=\mathbb{1}$ is easily found taking into account that there are $n$ constraints $\sum_{i=1}^n |U_{\ell i}^2| =1$ and two times $n(n-1)/2$ further constraints $\sum_{i=1}^n  U_{\ell i}^* U_{\ell' i} =0$, with $\ell\neq\ell'$.\footnote{In fact, the  number of complex out-of-diagonal entries of the hermitian matrix $M=U^\dagger U$ is $n(n-1)$.} 
Thus we have $n^2$ real constraints and $n^2$ real 
parameters in the matrix $U$. The determinant is a  phase factor that can be explicitly factored out, 
corresponding to the factorization of the group $\mathrm{U}(n)=\mathrm{U}(1)\otimes\mathrm{SU}(n)$. The same counting can be done even more simply, by considering that a unitary matrix can be written as $U=\exp(i H)$ where $H$ is hermitian; the subgroup $\mathrm{SU}(n)$ corresponds to the subset of the traceless hermitian matrices. Summarizing, in the $2\times 2$ case we have 4 parameters and in the $3\times 3$ case we have 9 of them.

\subsubsection{Overall phases do not matter}  
As we have mentioned just above, the states with given flavor are superpositions of states with given mass and the former evolve in time in a non-trivial manner, acquiring different phase factors. For this reason, flavor transformation occur: this is the conceptual core of  the neutrino oscillation phenomena.

 Note that in quantum mechanics there is the freedom to define at our will which are the overall phase factors of the   
states, and in particular,  which are the phase factors of the states 
with given flavor and those with given mass. 
To be sure, we note that when these phases are changed,  the mixing matrix {\em does change}. In fact, if we change the phases of 
the flavor and of the mass states, according to,
\begin{equation}
|\nu_\ell \rangle' =e^{i\zeta_\ell} |\nu_\ell \rangle
\quad\text{and}\quad
|\nu_i \rangle' =e^{i\zeta_i} |\nu_i \rangle
\end{equation}
we will have,
\begin{equation}
|\nu_\ell \rangle'=  (U_{\ell i}')^* |\nu_i \rangle' 
\end{equation}
with a mixing matrix where we have changed the phases in all rows and columns,
\begin{equation}
U_{\ell i}'=e^{- i\zeta_\ell}\ U_{\ell i}\  e^{i\zeta_i}
\label{ridif}
\end{equation}
On the other hand, it is evident that this redefinition will not change 
the absolute values,
\begin{equation} \label{ant}
|U_{\ell i}'|=|U_{\ell i}|
\end{equation}
and, most importantly, it will not change the probabilities of transition in which we are interested, such as,
\begin{equation}
| \langle  \nu_{\ell_2} |\nu_{\ell_1}, t \rangle |^2 = 
| \text{\reflectbox{$'$}}{\langle  \nu_{\ell_2}} {|\nu_{\ell_1}, t \rangle'} |^2 
\end{equation}
where $\text{\reflectbox{$'$}}\langle  \nu|$ is the bra corresponding to the ket $|\nu\rangle'$;  recall that we have $| \nu_{\ell} \rangle\equiv 
| \nu_{\ell}, t=0 \rangle$.
The possibility to redefine the overall phase factors has however a prominent consequence:
not all the $n^2$ parameters of a unitary matrix are relevant 
for the phenomenon that we are interested to discuss.
Next, we proceed to count the parameters that, instead, do matter.

\subsubsection{Number of parameters relevant to oscillations}
The parameters that are relevant to neutrino and antineutrino oscillations  are those that are invariant under a redefinition of the phases of the flavor fields and of the mass neutrino fields. Evidently these are less than $n^2$, the parameters of a unitary matrix. 
We present three different ways to count them,
in view of the importance of this counting.\footnote{The same counting applies to the quark sector; for this reason, it is common usage to call this subset of  parameters as {\em physical parameters}. However, this terminology is misleading since for neutrinos other phases are potentially measurable --- even if not by means of  (ultrarelativistic) neutrino oscillations.}
\begin{enumerate}[itemsep=-0.7ex,partopsep=1ex,parsep=1ex]
	\item The standard way to count these parameters is to consider that for each flavor or mass field, we can impose a constraint on $U_{\ell i}$ exploiting the redefinition of the phases of the fields.  For instance we can choose to make real a full row and a full column of the mixing matrix; but the two have one parameter in common, thus we can arrange $2n-1$ constraints only.
	\item Another consideration that leads to the same conclusion is that we have $n$ phase factors for the flavor fields and $n$ of them for the mass fields, but  the global phase factors of the flavor fields 
and of the mass fields amount to  a single effective parameter, not to two.
\item As already noted in \eqref{ant}, 
 using different  phases of the fields 
 the modulus of the elements 
of the mixing matrix stays 
unchanged. 
Thus, let us count the number of independent 
parameters $|U_{\ell i}^2|$.
The unitarity relations $\sum_{i=1}^n |U_{\ell i}^2|  =1$ fix one parameter for each  row and similarly for   each column. This consideration implies that the real matrix with elements $|U_{\ell i}^2|$ has only $(n-1)^2$ independent parameters. 
For example, in the $3\times 3$ example, we have four of them,
\begin{equation}
\left(
\begin{array}{lll}
\mathrm{free} & \mathrm{free} & \mathrm{fixed} \\ 
\mathrm{free} & \mathrm{free} & \mathrm{fixed} \\
\mathrm{fixed} & \mathrm{fixed} & \mathrm{fixed} 
\end{array}
\right)
\end{equation}
\end{enumerate}
All in all, the result is that the number of 
free (or independent) parameters is,
\begin{equation}
\text{Number of free parameters}=
n^2-(2n-1)=(n-1)^2
\end{equation}
This means 4 in the   $3\times 3$ case while in the $2\times 2$ case we have only 1 relevant parameter.

\subsubsection{Angles and phases in the mixing matrix}
It is not difficult to identify these free parameters more precisely.
First, consider the subclass of unitary matrices that are also real,  $U=U^*$. It is easy to verify that this corresponds to the orthogonal real matrices, that have $n(n-1)/2$ parameters. These parameters are {\em angles}; e.g., in the $2\times 2$ case we are dealing with a single angle, whereas in the 
 $3\times 3$ case we have the three Euler angles. The remaining physical parameters are {\em phases} relevant to oscillations,
\begin{equation}
\mbox{Number of phases}=(n-1)^2-n(n-1)/2=(n-1) ( n - 2 )/2\quad\text{for}\quad n\ge 2
\end{equation}
In other words, a $2\times 2$ matrix can be 
made real by a suitable choice of phase factors, so that no phases have relevance to oscillations, 
whereas a 
 $3\times 3$ matrix has one physical phase factor relevant to neutrino oscillations. 
 
\subsubsection{Quartets} 
As we will see later, a complex combination of mixing elements, sometimes called {\em quartet}, enters the explicit expressions of 
the oscillation probabilities,
\begin{equation}
\mathcal{Q}_{\ell \ell', ij}= U_{\ell i}\; U_{\ell' i}^*\;  U_{\ell j}^*\; U_{\ell' j} 
\label{quartet}
\end{equation}
and indeed it has the same value for any 
choice of the phases of the $\nu_\ell$ and $\nu_i$ fields --- i.e., it is not convention-dependent. Moreover, it satisfies various interesting properties, 
such as,
 \begin{equation}
 \begin{cases}
 \mathcal{Q}_{\ell \ell', ij}=\mathcal{Q}_{\ell' \ell, ij}^*
= \mathcal{Q}_{\ell \ell', ji}^*\;\\
 \sum_{i} \mathcal{Q}_{\ell \ell', ij}=\delta_{\ell \ell'} |U_{\ell j}^2|
  \;\\
   \mathcal{Q}_{\ell \ell, ij}=|U_{\ell i}^2 U_{\ell j}^2 |
    \; \\
    \dots
    \end{cases}
 \end{equation}
We can usefully define the real and imaginary parts,
\begin{equation}
\mathcal{Q}_{\ell \ell', ij}=R_{\ell \ell', ij}+ i \ I_{\ell \ell', ij}
\quad\text{namely}\quad
\begin{cases}
	R_{\ell \ell', ij} =\mathrm{Re}[ \mathcal{Q}_{\ell \ell', ij} ]\\
I_{\ell \ell', ij} =\mathrm{Im}[ \mathcal{Q}_{\ell \ell', ij} ]
\end{cases} \label{caccase}
\end{equation}
which are of interest since the expression of the oscillation probabilities depend upon these two real quantities, as will be  discussed in \sectionname~\ref{sec:nfl}. These quantities can be always expressed in terms of the mixing angles and of the physical phase factor, however they are interesting on their own, being parameterization independent. Therefore, we examine  here the real part of the quartets $R$ and their imaginary part $I$. 
We will show (focussing on the 3 flavor case) that it is possible to  
calculate both of them 
in terms of the simplest phase-independent quantities, 
namely $|U_{\ell i}^2|$, up to a sign that remains undetermined.

\subsubsection{Expression of \boldmath$R$ and $I$ in terms of $|U_{\ell i}|$}
The symmetry properties are as follows:
the real part is even in the exchange of $\ell \leftrightarrow \ell'$ and
$i\leftrightarrow j$, whereas the imaginary one is odd in both exchanges. Thus,  in the case  $\ell=\ell'$, we have,
\begin{equation}
R_{\ell \ell, ij}=|U_{\ell i}^2 U_{\ell j}^2 |
\quad\text{and}\quad
I_{\ell \ell, ij}= 0
\end{equation}
We are also very interested in the cases when
$\ell\neq \ell'$ and $i\neq j$. The expressions are as follows,
\begin{eqnarray}\label{pranga1}
		R_{\ell \ell', ij} &=& \frac{1}{2}\left( a_k^2-a_i^2-a_j^2  \right) \\
		|I_{\ell \ell', ij}| &=& 2 \sqrt{ 
		p (p-  a_k) (p-a_j) (p-a_i) }
	\label{pranga2}
\end{eqnarray}		
with the following definitions,
\begin{equation} \nonumber
				a_i=|U_{\ell i} U_{\ell' i}|\,;\quad
a_j=|U_{\ell j} U_{\ell' j}|\,;\quad
a_k=|U_{\ell k} U_{\ell' k}|\,;\quad
p=(a_i+a_j+a_k)/2
\end{equation}
where $\ell\neq \ell'$ and 
$i,j$ and $k$ are a permutation of $1,2$ and 3. 
The first expression is immediately obtained by considering  the relation $U_{\ell k} U_{\ell' k}^*=  -U_{\ell i} U_{\ell' i}^*-U_{\ell j} U_{\ell' j}^*$, taking its absolute value and  then 
recalling the definition of $R_{\ell\ell',ij}$.
In order to derive the 
second expression, one can proceed by following steps:
\begin{enumerate}[itemsep=-0.7ex,partopsep=1ex,parsep=1ex]
\item Let us define the following complex numbers,
 $z_i=U_{\ell i} U_{\ell' i}^*$ with $i=1,2,3$, 
 and let us express them as $z_i=a_i \exp(i \theta_i)$, where of course $a_i=|z_i|$.
\item The orthogonality relation 
$U_{\ell k} U_{\ell' k}^* + U_{\ell i} U_{\ell' i}^* + U_{\ell j} U_{\ell' j}^*=0$ reads $z_k+z_i+z_j=0$, that 
can be seen as a triangle in the complex plane.
\item Its area is 
$A=|a_i a_j \sin\theta_{ij}/2|$ where we introduced the angle $\theta_{ij}=\theta_i-\theta_j$ between two sides.
\item
If we rewrite 
$I_{\ell \ell',ij}=\mathrm{Im}[z_i z_j^*]=\mathrm{Im}[ a_i e^{i\theta_i} a_j e^{-i\theta_j}]=a_i a_j \sin\theta_{ij}=2\times A$,  
we realize that the absolute value of this quantity is twice the area
of the above triangle. 
\item But the area  can be  evaluated also with Heron's formula, since  length of the sides is $|z_i|=a_i$ with $i=1,2,3$. 
\end{enumerate}
It is remarkable that just as we have a single phase, we have a single  imaginary quantity in the three neutrino case, e.g.,
\begin{equation} 
I_{\ee\mu,12}=- I_{\ee\tau,12}=+I_{\ee\tau,13}=...
\label{lostess}
\end{equation}
The proof of the first relation goes as follows: let us consider,
\begin{equation}
	I_{\ee\mu,12}=\mathrm{Im}\left[ U_{\ee 1} U_{\mu 1}^* U_{\ee 2}^* U_{\mu 2}\right]
\end{equation} 
From the orthogonality relation, we have,
\begin{equation}
	I_{\ee\mu,12}=
\mathrm{Im}\left[ U_{\ee 1}  U_{\ee 2}^* \left( -U_{\ee 1}^* U_{\ee 2}-U_{\tau 1}^* U_{\tau 2}\right) \right]
\end{equation}
that is also equal to,
\begin{equation}
	-\mathrm{Im}[ U_{\ee 1}  U_{\ee 2}^* U_{\tau 1}^* U_{\tau 2} ]=- I_{\ee \tau,12}
\end{equation}
which is the desired relation. Similarly, we can obtain the other ones.  
An interesting consequence of this result is that, if some element of the mixing matrix is zero, this quantity will be zero. Conversely, a necessary condition for this quantity to be large is that all the elements of the mixing matrix are large. (More on that in the next pages --- see \eqref{jarlsk} and \eqref{cicciri} --- and then in \sectionname~\ref{pipin}.)

\subsection{The standard parameterization}

	
	\paragraph{Explicit expression}
	The leptonic (PMNS) mixing matrix 
	$U$ is conventionally written 
	in terms of the mixing angles $\theta_{12}$, 
	$\theta_{13}$ and $\theta_{23}$ and of the CP-violating phase  $\delta$ that  plays a role in neutrino oscillations.
	The most common convention is,
	\begin{equation}
		U = \left(  
			\begin{array}{ccc} 
				c_{12}c_{13} & s_{12}c_{13} & s_{13}e^{-i\delta} \\ 
				-s_{12}c_{23}-c_{12}s_{13}s_{23}e^{i\delta} & c_{12}c_{23}-s_{12}s_{13}s_{23}e^{i\delta} & c_{13}s_{23} \\ 
				s_{12}s_{23}-c_{12}s_{13}c_{23}e^{i\delta} & -c_{12}s_{23}-s_{12}s_{13}c_{23}e^{i\delta} & c_{13}c_{23}
			\end{array}  
		\right) \label{stdd}
	\end{equation}
	where $s_{ij},c_{ij} \equiv \sin\theta_{ij},\cos\theta_{ij}$ 
	and where the angles lie in the first quadrant whereas the phase $\delta$ is generic, $\delta \in [0,2\pi)$.
	Note the usage of the same phase convention and parameterization of the quark (CKM) mixing matrix even if, of course,
	the values of the parameters are different. This convention makes evident that the CP phase is unphysical if $\theta_{13}\to 0$. Indeed this limit holds if {\em any} mixing angle is zero.

\subsubsection{Expression in terms of product of matrices}
The same matrix can be written as a product of three Euler rotation and a  
matrix of phases as follows,
\begin{equation}\label{mixingmatrix}
U= 
R_{23}\ \Delta \ R_{13}\ \Delta^*\ R_{12}
\end{equation}
where,
\begin{equation}
\begin{aligned}
R_{23} &=
\left(
\begin{array}{ccc}
1&0& 0\\
0& c_{23} & s_{23} \\
0&-s_{23} & c_{23}
\end{array}
\right)\quad
R_{13}=
\left(
\begin{array}{ccc}
c_{13} & 0 & s_{13} \\
0 & 1 & 0 \\
-s_{13}& 0  & c_{13}
\end{array}
\right)\\
R_{12} &=
\left(
\begin{array}{ccc}
c_{12}  & s_{12} & 0 \\
-s_{12}  & c_{12} & 0 \\
0 & 0 & 1 
\end{array}
\right)\quad\hphantom{R_{13}}\negphantom{$\Delta$}
\Delta=
\left(
\begin{array}{ccc}
1 &0& 0\\
0& 1 & 0 \\
0 & 0 & e^{i\delta}
\end{array}
\right)
\end{aligned}
\end{equation}

\begin{figure}[t]
\centering
\includegraphics[width=0.8\textwidth,bb=0 0 274 139]{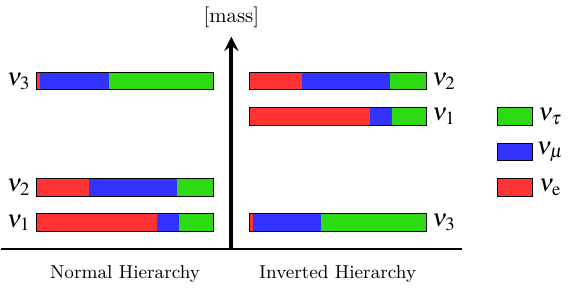}
	\caption{
   Illustration of the mass spectra compatible with 
   the data from neutrino oscillations; left, normal hierarchy; right,
   inverted hierarchy. 
   The length of the colored bars are proportional
   to the flavor content of each mass state, $|U_{\ell i}|^2$. \label{an:fig.hie}}
\end{figure}

\subsubsection{An alternative expression}
By changing the phases, one can obtain apparently different matrices, whose parameters however maintain the same meaning and therefore describe exactly the same oscillation phenomena. An example of such a choice of some interest 
is,
\begin{equation*}
U'= \left(  
			\begin{array}{ccc} 
				c_{12}c_{13} & s_{12}c_{13} & s_{13} \\ 
				-c_{12}s_{13}s_{23}-s_{12}c_{23} e^{-i\delta} & -s_{12}s_{13}s_{23} +c_{12}c_{23} e^{-i\delta}& 
				c_{13}s_{23} \\ 
				-c_{12}s_{13}c_{23}+s_{12}s_{23} e^{-i\delta}& 
				-s_{12}s_{13}c_{23} -c_{12}s_{23}e^{-i\delta}& c_{13}c_{23}
			\end{array}  
		\right) 
\end{equation*}
This new choice is somewhat advantageous for the discussion of neutrinoless double beta decay, since the phase $\delta$ does not play any role for the parameter that summarizes the effect of the light Majorana neutrino masses, i.e.,  
$m_{\beta \beta}\equiv \sum_{i}U_{\ee i}^2 m_i$.\footnote{The parameter $m_{\beta\beta}$  depends only upon the phases of the states with given masses, named ``Majorana phases'', that could be simply and usefully included in $m_i$, 
namely $m_{\beta \beta}=\sum_{i}|U_{\ee i}^2| (m_i e^{i 2 \phi_i})$.}
Below, we exhibit a direct construction of the unitary mixing matrix where the single physical phase factor is singled out, that yields the above form of the matrix.
 
It is quite natural to choose  the phases of the fields in order to make 
real and positive one row and one column. Choosing the ones that cross in $U_{\ee 3}\equiv \sin\theta\equiv s$ we get,
\begin{equation}
    U'' =
    \left(
    \begin{array}{ccc|c}
    & c\;\vec{n} & & s\\
    \hline
    & & & \\
    & u & & c\;\vec{m}\\
    & & & \\
    \end{array}
    \right)
\end{equation}
where of course $c=\cos\theta$. 
Note that in order to 
have a simpler notation, we assume that the 
row and column vectors are evident from the context, and do not
indicate the transpose sign on  the vectors $\vec{m}$ and  $\vec{n}$.
Here, we introduced two real unit vectors that we can decompose as,
\begin{equation}
\vec{n}=(n_1,n_2)=(c_{12},s_{12})\quad\text{and}\quad
\vec{m}=(m_1,m_2)=(s_{23},c_{23})
\end{equation}
The  $2\times 2$ matrix $u$ is still to be determined. Now,
consider the second row $r_2$ of the full matrix $U''$, with components
$r_2=( \alpha \; \vec{n}+ \beta \; \vec{n}_\perp\ ,\ c m_1) $, where 
$\alpha$ and $\beta$ are in general two complex numbers and 
$\vec{n}_\perp=(-s_{12},c_{12})$ is a real vector orthogonal to $\vec{n}$.
When we impose the orthogonality with the first row, we find that 
$\alpha=-s\; m_1$, and when we impose that $||r_2||^2=\sum_i |U_{\mu i}''|^2=1$, we find that $|\beta^2|=m_2^2$. Thus we conclude that the second row is
$r_2=( -s m_1 \; \vec{n}+ e^{i\phi}\ m_2 \; \vec{n}_\perp\ ,\ c m_1) $. For similar reasons, the third row is 
$r_3=( -s m_2 \; \vec{n}- e^{i\psi}\ m_1 \; \vec{n}_\perp\ ,\ c m_2) $. Finally, imposing the orthogonality with the third row, one realizes that $\phi=\psi$. In short, the matrix $u$ defined above is,
\begin{equation}
u=- s\ \vec{m}\otimes \vec{n} + e^{i\phi}\ \vec{m}_\perp\otimes \vec{n}_\perp
\end{equation}
Thus, by choosing the parameters
$\phi=-\delta$, $\theta=\theta_{13}$ and with the phase choice for  $\vec{m}_{\perp}=(c_{23},-s_{23})$,  
we conclude that $U''=U'.$

\begin{figure}[t]
\centering
\subfigure[Quark Mixing Elements.]
{\includegraphics[width=0.35\textwidth,bb=0 0 132 133]{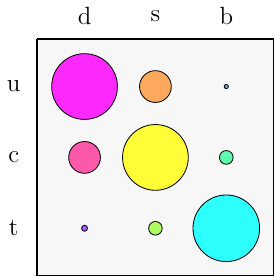}
\label{fig:CKM}}
\hspace*{0.1\textwidth}
\subfigure[Lepton Mixing Elements.]{
\includegraphics[width=0.35\textwidth,bb=0 0 135 133]{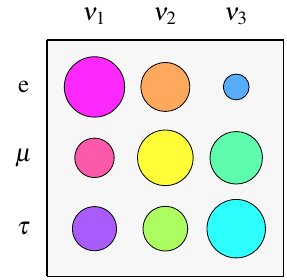}
\label{fig:PMNS}}
\caption{The surfaces of the circles represent  the size of the mixing elements, assuming normal mass hierarchy. From top to bottom,  from left to right:
   \protect\subref{fig:CKM} quark mixing (CKM) elements $|V_{\mathrm{ud}}|,  |V_{\mathrm{us}}|,
   |V_{\mathrm{ub}}|,|V_{\mathrm{cd}}|, ...$;
   \protect\subref{fig:PMNS} lepton mixing (PMNS) elements $|U_{\ee 1}|,  |U_{\ee 2}|,
   |U_{\ee 3}|,|U_{\mu 1}|, ...$ 
      in both panels, the mixing matrices are supposed to be unitary.  
   The hierarchical structure of quark mixing elements contrasts with the one of lepton mixing elements.}
\label{an:fig.mix}
\end{figure}


\subsubsection{Measure of CP violation}
As we have already discussed, for 3 neutrinos 
there is a single imaginary quantity that signals the  presence of a complex mixing matrix, and that, therefore, rules the leptonic CP violation phenomena in neutrino oscillations: see \eqref{lostess}. In the 
standard parameterization, the expression of this quantity is, 
\begin{equation}
J_{\mathrm{CP}}\equiv 
I_{\ee\mu,12}=s_{13} c_{13}^2\; s_{12} c_{12}\; 
s_{23} c_{23}\; s_\delta
\label{jarlsk}
\end{equation}
notice that (consistently with the above discussion) 
this quantity is bound to vanish if anyone of the angles is zero, or if $\delta=0$. It is easy to show that the maximum value 
is $J_{\mathrm{CP}}^{\mathrm{max}}=1/(6\sqrt{3})$, obtained when $\theta_{12}=\theta_{23}=\pi/4$ and $\sin\theta_{13}=1/\sqrt{3}$ and $\delta=\pi/2$; the minimum value is just the opposite one.
The symbol $J$ is used to honor C.~Jarlskog, who originally introduced such a quantity to describe CP violation in quark systems. 

\subsection{What we know on the parameters of neutrino oscillations}

As we will see better below, neutrino oscillations depend upon the 
difference of neutrino masses squared, and this is why the relevant massive parameters are sometimes indicated
symbolically as delta-m-squared.  The 
analysis of oscillation data have allowed us to discover and measure two  different values of delta-m-squared, 
in a manner that  will be recalled shortly later and apart from an important remaining ambiguity. 
The results of these experiments and analyses are 
illustrated graphically in 
\figurename~\ref{an:fig.hie}; note that  two different types of neutrino mass hierarchies (or orderings, or spectra) are compatible with the existing data.

	The values of the parameters of the leptonic mixing matrix, obtained from a global analysis of all oscillation data available 
	in 2016	\cite{capo}, are presented in \tablename~\ref{tab:lisi16}.
We present the best fit values and 
an estimation of the accuracy, obtained from the two sigma 
	ranges as follows: for any parameter $x$, at a given confidence level corresponding to $n$ sigma in the gaussian approximation, the uncertainty is given  \cite{capo} as an asymmetric interval $x_{\mathrm{min}} < x_{\text{best fit}} < x_{\mathrm{max}}$. In order to compute a single number able to give, at first glance, an overview of the order of magnitude of the relative uncertainty $\Delta x / x_{\text{best fit}}$ we set $\Delta x = (x_{\mathrm{max}} - x_{\mathrm{min}}) / 2n$.	
	 Note that these results do not depend strongly upon the type of mass hierarchy, except for the parameter $\theta_{23}$.
	Just for reference, the corresponding angles, in the case 
	of normal mass hierarchy, are,
	$\theta_{13}=8.5^\circ$, 
	$\theta_{12}=33^\circ$, 
	$\theta_{23}=42^\circ$, 
	$\delta=250^\circ$.  For what concerns the value of $\delta$, it was determined by the combined data from T2K and NOvA, that are experiments with appearance channels (T2K was the first able to constrain the CP violation phase). 

	For further illustration, we compare
	in \figurename~\ref{an:fig.mix} the absolute values of the lepton mixing elements 
	(i.e., the PMNS mixing matrix) with  the corresponding values of the quark mixing elements (i.e., the CKM mixing matrix).
	As a direct application, we can evaluate the universal CP violating quantity that has been defined in \eqref{jarlsk}. With the present central values and assuming normal mass hierarchy, we have, 
\begin{equation} \label{cicciri}
J_{\mathrm{CP}} = -0.03
=-32\%\ J_{\mathrm{CP}}^{\mathrm{max}}
\end{equation}
the result for inverted mass hierarchy is similar.

\begin{table}[t]
	\centering
	{\begin{tabular}{@{}lcrc@{}} 
 			\toprule
 			& {Normal (Inverted)} & {Error} & {Units} \\
			\midrule 
  			$\Delta m^2$ & 2.50 (2.46) & $18\,\%$ & $10^{-3}\,\mathrm{eV}^2$\\
  			$\delta m^2$ & 7.37 (7.37) & $2.4\,\%$ & $10^{-5}\,\mathrm{eV}^2$ \\ 
  			\midrule
  			$\sin^2\!\theta_{13}$ &  2.17 (2.19) &  $4.8\,\%$ &   $10^{-2}$ \\ 
			$\sin^2\!\theta_{12}$ &  2.97 (2.97) &  $6.2\,\%$ &   $10^{-1}$ \\ 
			$\sin^2\!\theta_{23}$ &  4.43 (5.75) &  $16\,\%$ &   $10^{-1}$ \\ 
			$\delta$ &  1.39 (1.39) &  19\% &   $\pi$ \\
 			\bottomrule
  		\end{tabular}}
  		\caption{Results of the global analysis of oscillation data of the Bari group (2016) \cite{capo}. The precise meaning of the parameters and the error estimate are discussed in the text.}
\label{tab:lisi16}
\end{table}

	The specific choices of the delta-m-squared parameters 
	used in this analysis is,
	\begin{equation}\delta m^2=m_2^2-m_1^2\;;\quad
	\Delta m^2=|m_3^2-(m_2^2+m_1^2)/2|\end{equation} 
that, denoting the lightest neutrino mass by $m$, 
is equivalent to the following set of relations, 
\begin{equation}
\left\{
\begin{array}{lll}
m_1\ \  &=m & (=\sqrt{m^2+\Delta m^2-\delta m^2/2}) \\
m_2\ \  & =\sqrt{m^2+\delta m^2}  & (=\sqrt{m^2+\Delta m^2+\delta m^2/2}) \\
m_3\ \   &=\sqrt{m^2+\Delta m^2+\delta m^2/2} & (=m)\\
\end{array}
\right.
\end{equation}
where the expression outside (inside) the brackets applies for normal (inverse) mass hierarchy.

Note incidentally 
that the minimum mass is not probed by oscillations.
	It should be stressed that the case of 
	normal mass hierarchy 
	is slighly favored from the 
	present experimental information at $\Delta \chi^2=2.8$
	(namely, about $1.7\sigma$) from the same analysis	\cite{capo}.
	
	\chapter{Vacuum neutrino oscillations\label{ch:vo}}
\minitoc

In the first section (\ref{sec:c2p}), we derive the 
formulae of the oscillation probability, namely the transition probability of neutrino of a given flavor $\ell$ to a flavor $\ell'$, and 
discuss the standard manipulations that are needed to 
understand thoroughly the underlying physics.
In the next section (\ref{sec:ae}) we will consider various applications of these formulae, that will allow to appreciate their usefulness and to explore the flexibility of this formalism.

\section{General formalism\label{sec:c2p}}
Using the results on the mixing matrix $U$ obtained in the previous section, we discuss how to describe the oscillation probabilities in the  situations when the flavor is unchanged (survival or disappearance probability) or instead when the flavor of the neutrino changes (appearance probability). We consider the general case (\sectionname~\ref{sec:nfl}) and also the special case when we have two neutrinos only (\sectionname~\ref{sec:2fl}). Then we consider a formal developments, namely
we discuss neutrino oscillations in 
the context of the field theoretical formalism 
(\sectionname~\ref{sec:fiello}),
we introduce the effective hamiltonians (\sectionname~\ref{sec:hhh}), and finally we discuss which new effects are expected describing the neutrinos as wave-packets (\sectionname~\ref{sec:wp}). 

\subsection{Oscillations with n-flavors and n-mass states \label{sec:nfl}}

\subsubsection{Time dependence and transition amplitude}
The most general case we can consider, describes $n$ neutrino flavor eigenstates $|\nu_\ell\rangle$ as linear combinations of the $n$ neutrino mass eigenstates $|\nu_j\rangle$. The flavor states change in time as the energy eigenstates evolve,
\begin{equation}
|\nu_\ell,t\rangle=\sum_{j=1}^nU^*_{\ell j}|\nu_j,t\rangle 
\quad\text{where}\quad
|\nu_j,t\rangle =|\nu_j,0\rangle\:e^{-i E_j t}
\label{soluzella}
\end{equation}
where we use a somewhat redundant but transparent notation
for $|\nu_j,0\rangle=|\nu_j\rangle$ and similarly for 
$|\nu_\ell ,0\rangle=|\nu_\ell \rangle$.
Therefore, we find easily,
\begin{equation}
|\nu_\ell,t\rangle =\sum_{j=1}^nU^*_{\ell j}\,e^{-i E_j t}|\nu_j,0\rangle =\sum_{j=1}^nU^*_{\ell j}\,e^{-i E_j t} \left( \sum_{\ell'}U_{\ell'j}|\nu_{\ell'},0\rangle \right)
\end{equation}
and eventually,
\begin{equation}
|\nu_\ell,t\rangle =
\sum_{\ell'}  
\left( \sum_{j=1}^n 
U^*_{\ell j}
\,e^{-i E_j t}\,    
U_{\ell'j} 
\right)
|\nu_{\ell'},0\rangle
\label{bum1}
\end{equation}
where we used the basic property of unitary matrices 
$U^\dagger = U^{-1}$.
 From this expression, 
 the formula for the {\em transition amplitude}
 follows immediately, 
\begin{equation}
 \mathcal{U}_{\ell \ell'}(t) \equiv \langle \nu_{\ell'},0|\nu_\ell,t\rangle = \langle \nu_{\ell'},0| \sum_{\ell''} \left( 
 \sum_{j=1}^n U^*_{\ell j}\, e^{-i E_j t}\, U_{\ell''j}\right) |\nu_{\ell'',0}\rangle
 \end{equation}
 leading to,
 \begin{equation}
 \mathcal{U}_{\ell \ell'}(t) =
 \sum_{j=1}^n U^*_{\ell j}\,e^{-i E_j t} \,U_{\ell'j}
 \label{bum2}
 \end{equation}
  The reason why we prefer to indicate this amplitude with the  symbol $\mathcal{U}$, rather than  with the common symbol $\mathcal{A}$ used by other authors, will be discussed in \sectionname~\ref{sec:fiello} and \ref{sec:hhh}.
  
\subsubsection{Probabilities}
With few more straightforward manipulations\footnote{In the
specific,
\begin{eqnarray*}
P_{\nu_\ell \rightarrow \nu_{\ell'}} &=& \mathcal{U}_{\ell\ell'}\mathcal{U}^*_{\ell\ell'} =  \left(\sum_i U^*_{\ell i}U_{\ell'i}e^{-i E_i t} \right)\left(\sum_j U^*_{\ell j}U_{\ell'j}e^{-i E_j t} \right)^* = \\
 &=&\sum_{ij} U^*_{\ell i}\;U_{\ell'i}\;U^*_{\ell'j}\;U_{\ell j}\; e^{-i(E_i-E_j)t}= \sum_{ij}\mathcal{Q}_{\ell\ell',ij}^* e^{-i(E_i-E_j)t}
 \end{eqnarray*}
 Finally, we note that 
$\mathcal{Q}_{\ell\ell',ii}=|U_{\ell i}^2| |U_{\ell' i}^2| $.} 
we obtain the probability,
\begin{equation}
P_{\nu_\ell\rightarrow\nu_{\ell'}}(t)=|\mathcal{U}_{\ell\ell'}(t)|^2
=\sum_{i=1}^n|U_{\ell i}|^2|U_{\ell' i}|^2+\sum_{i\neq j}^n   \mathcal{Q}_{\ell\ell',ij}^* \;e^{-i(E_i-E_j)t}
\label{Pll'}
\end{equation}
that has a compact expression 
using the quartet
$\mathcal{Q}_{\ell\ell',ij}$, 
defined and discussed in the previous section ---
see \eqref{quartet}.
 Noting that  
 $\mathcal{Q}_{\ell\ell',ij} = \mathcal{Q}_{\ell\ell',ji}^*$, we see 
 that the second term in the last expression can be
 rewritten as,
 \begin{equation}
 P_{\nu_\ell\rightarrow\nu_{\ell'}}(t) 
 =\sum_{i=1}^n|U_{\ell i}|^2|U_{\ell' i}|^2+ 
\sum_{i> j}^n 2\; \mbox{Re}[  \mathcal{Q}_{\ell\ell',ij}^* \;e^{-i(E_i-E_j)t}\; ]
\end{equation}
At this point, it is useful to consider two opposite limits. The first one is when the time is very small, so that the phases are negligible. From $|\mathcal{U}_{\ell\ell'}(0)|=\delta_{\ell \ell'}$, 
we get the useful identity,
\begin{equation}
\delta_{\ell \ell'}
 =\sum_{i=1}^n|U_{\ell i}|^2|U_{\ell' i}|^2+ 
\sum_{i> j}^n 2\; \mbox{Re}[  \mathcal{Q}_{\ell\ell',ij}^*  ]
\label{identa}
\end{equation}
The second limit is when the time is very large, so that all the phases 
$(E_i-E_j) t\gg 1$; we suppose that there is no degeneracy in the energy level (i.e., none of the masses are equal). In this case, we can consider the average over time,
\begin{equation}
\langle P_{\nu_\ell\rightarrow\nu_{\ell'}}\rangle\equiv 
  \lim_{T\to \infty } \frac{1}{T}\int_0^T  P_{\nu_\ell\rightarrow\nu_{\ell'}}(t) \:\mathrm{d} t=
\sum_{i=1}^n|U_{\ell i}|^2|U_{\ell' i}|^2
\end{equation}
The {\em averaged} value of the probability is an important quantity whose physical meaning will be discussed in \sectionname~\ref{sec:ave}. Adopting the definition of the phase,
\begin{equation}
\varphi_{ij}=\frac{E_i-E_j}{2}\ t
\end{equation}
along with the identity given in eq.\ \eqref{identa}
we find  the final, general expression \cite{strumiavissani},
\begin{equation}
P_{\nu_\ell\rightarrow\nu_{\ell'}}=\delta_{\ell\ell'}- \sum_{i>j} \left( 4\;\mbox{Re}[\mathcal{Q}_{\ell \ell', ij}]\;\sin^2 \varphi_{ij}+2\; \mbox{Im}[\mathcal{Q}_{\ell \ell',ij}]\;\sin 2\varphi_{ij} \right)
\end{equation}
A further simplification of the expression is possible in
the $3 \times 3$ case, since the imaginary part of the quartet, 
 $I_{\ell \ell',ij}$, can be rewritten 
 using \eqref{lostess} and  \eqref{jarlsk},
\begin{equation} 
I_{\ell \ell',ij} =  J_{\mathrm{CP}}   
\sum_{\ell'' , k} \epsilon_{\ell \ell' \ell''}\ \epsilon_{ijk} =
\begin{cases}
	J_{\mathrm{CP}} &\text{for
	$\ell,\ell',i,j=e,\mu,1,2$; etc.}\\
	-J_{\mathrm{CP}} &\text{for
	$\ell,\ell',i,j=e,\mu,2,1$; etc.}\\
	 0 &\text{for $\ell = \ell'$ or $i=j$}\\
\end{cases}
\end{equation}
where $J_{\mathrm{CP}} $ is the single (Jarlskog) invariant that quantifies 
CP violation.
Thus, 
with the definitions of \eqref{caccase}, the expression for the probabilities of   
disappearance (or survival, $\ell=\ell'$) and of appearance ($\ell\neq\ell'$), 
valid in the three flavor case, read,
\begin{equation}
\begin{cases}\displaystyle
\label{cit}
	P_{\nu_\ell\rightarrow\nu_{\ell}} = 1 - \sum_{i>j}  4\; 
	| U_{\ell i}^2 U_{\ell j}^2|\;\sin^2 \varphi_{ij}
	&[\ell=\ell'] \\ \displaystyle\vphantom{\sum_i^\text{M}}
	P_{\nu_\ell\rightarrow\nu_{\ell'}} = - \sum_{i>j}  4\;R_{\ell \ell', ij}\;\sin^2 \varphi_{ij}
\pm 8 \; J_{\mathrm{CP}}   
\prod_{i>j} \sin \varphi_{ij} &[\ell\neq\ell']
\end{cases}
\end{equation}
 where for the second formula we used a few standard  trigonometrical manipulations and where the sign in front of the CP violating part   
 is $\pm=\sum_{\ell''} \epsilon_{\ell \ell' \ell''}$.
 These equations can be possibly expressed in terms of the 
 elements of the mixing matrix squared by using 
 \eqref{pranga1} and \eqref{pranga2} if one wishes so, but recall that the sign of $J_{\mathrm{CP}}$ is not fixed by those relations.

\subsection{The case with 2 flavors and 2 masses\label{sec:2fl}}
The simplest (but very useful!) model one can think of in order to explain neutrino oscillations is the one with two flavor states. 
We consider for definiteness the two mass eigenstates $|\nu_i\rangle \quad i={2,3}$ and the flavor eigenstates $|\nu_\ell \rangle \quad \ell={\mu,\tau}$, which is a reasonable approximation of a 
real situation (as we discuss later in \sectionname~\ref{ss:2f}). 
There is just a single parameter that fully describes the relevant part of the 
mixing matrix.
Thus, we can assume the following natural parametrization,
\begin{equation}
\left. 
  \begin{pmatrix} \Pnum \\ \Pnut \end{pmatrix}
  =	
  \begin{pmatrix} U_{\mu2} & U_{\mu3} \\ U_{\tau2} & U_{\tau3} \end{pmatrix} 
  \begin{pmatrix} \nu_2 \\ \nu_3 \end{pmatrix}
\right.
\quad\text{where}\quad
U=\begin{pmatrix} c & s \\ -s & c \end{pmatrix} 
\end{equation}
where we indicate with $c \equiv \cos(\theta)$ and $s \equiv \sin(\theta)$ the cosine and sine of the single $\theta$ parameter: the mixing angle.
We have,
\begin{equation}
|\Pnum\rangle =U_{\mu i} | {\nu}_i\rangle
\quad\text{and}\quad
|\APnum \rangle =U_{\mu i} | \bar{\nu}_i\rangle
\end{equation}
since the unitary matrix is real in the $2\times 2$ case: $U_{\ell i}^*=U_{\ell i}$. 
Now we evaluate the so called \textit{survival probability}, namely the probability to detect a $|\Pnum \rangle$ after the neutrino propagation in vacuum for a time $t$.  
The amplitude is,
\begin{equation}
\mathcal{U}_{\mu\to \mu}(t)\equiv
\langle \Pnum | \Pnum ,t \rangle = \sum_{j,k}
\langle \nu_k | U_{\mu k}
U^{*}_{\mu j}\, e^{-iE_jt} |\nu_j \rangle
= \sum_k |U_{\mu k}|^2\, e^{-i E_k t}
\end{equation}
Note in passing that the very same expression
is valid both for models $n$ mass states, if we simply allow the sum to 
run on all the possible values of $k$ and keep the values $E_k$ distinct; but now let us proceed with the simplest case where $n=2$.
We are just one step from calculating the survival probability,
\begin{equation}
P_{\pnum \rightarrow \pnum} \equiv |\ \mathcal{U}_{\mu\to \mu}(t)\ |^2 = \left| \; 
(1-|U_{\mu 3}|^2)+|U_{\mu 3}|^2\,e^{-i(E_3 - E_2)t}\; \right|^2
\label{2fla2mas}
\end{equation}
where we have used the unitarity condition $\sum_b |U_{ab}|^2 = 1$ valid $\forall a $. If we define $\varphi \equiv (E_3-E_2)t/2$ we can rewrite such a probability in a much more useful manner,
\begin{equation}
\begin{aligned} &P_{\apnum \rightarrow \apnum} =\\=& 
|1-|U_{\mu 3}|^2(1-e^{i2\varphi})|^2 =
|1-e^{i \varphi} |U_{\mu 3}|^2(e^{-i\varphi}-e^{i\varphi})|^2 =\\
=& |1+ 2 i e^{i \varphi} |U_{\mu 3}|^2\sin\varphi |^2 
=  |e^{-i \varphi}+ 2 i  |U_{\mu 3}|^2\sin\varphi |^2=\\
=&|\cos\varphi+i\sin\varphi (2 |U_{\mu 3}|^2 -1) |^2 =
1-4|U_{\mu 3}|^2(1-|U_{\mu 3}|^2)\sin^2(\varphi)
\end{aligned}
\end{equation}
The dependence of this expression on time (length) and mass spectrum is completely embedded in the $\varphi$ parameter and the energy has the standard expression $E_i = \sqrt{m^2_i + p^2}$ where $p$ is the momentum of the incoming neutrino.
It is quite useful to insert the numerical value of the physical constants in the formulae in order to make their meaning transparent.
In the ultra-relativistic limit the approximation 
the energy difference between two mass eigenstates of given momentum can be approximated to,
\begin{equation} 
E_3 - E_2 =\frac{E_3^2-E_2^2}{E_3+E_2}=
\frac{m_3^2-m_2^2}{E_3+E_2}
\simeq \frac{m^2_3-m^2_2}{2E} \equiv \frac{\Delta m^2}{2E} 
\end{equation}
where $E_2\approx E_3\approx p\approx E$.
If we plug this into the $\varphi$ definition we get,
\begin{equation} \sin^2(\varphi) = 
\sin^2\left(\frac{(E_3-E_2) t}{2 E}\right) 
\approx \sin^2\left(\frac{\Delta m^2  t}{4 E}\right)
\end{equation}
Since we are dealing with ultrarelativistic neutrinos (namely $p \gg m_i$), using the definition of natural units where $\hbar = c = 1$ we can rewrite the argument of the sine in previous expression, i.e., the phases, as,
\begin{equation} \frac{\Delta m^2 c^4 L}{4 E \; \hbar c} = \frac{\Delta m^2 c^4}{\mathrm{eV}^2} \frac{L}{\mathrm{km}} \frac{\mathrm{GeV}}{E} \left[\frac{\mathrm{eV}^2 \ \mathrm{km}}{4 \mathrm{GeV} (\hbar c)}\right]\end{equation}
and then recalling $\hbar c \sim \SI{197.396}{\mega\electronvolt\femto\meter}$ we get,
\begin{equation}  \frac{\mathrm{eV}^2\times \mathrm{km}}{4\,\mathrm{GeV} \times (\hbar c)}  \simeq  \frac{\mathrm{eV}^2\times 10^3\,\mathrm{m}}{4\times10^9\,\mathrm{eV} \times (197.396\times10^6\,\mathrm{eV} \times 10^{-15}\,\mathrm{m} )} \simeq 1.267 \label{notorio}
\end{equation}
the notorious number that appears in almost every paper on neutrino oscillation.

\subsubsection{A situation between 2 and 3 flavors}\label{3come2}
It might be interesting to stress analogies and differences with a slightly more complex model. Let us consider a world (a theory) not very different from the one we live in, where three mass eigenstates exist but two of them are degenerate. This translates in the following assumption,
\begin{equation}m_2 - m_1 = 0 \end{equation}
We are interested in computing the survival probability in such a system,
\begin{equation} \langle\nu_\ell|\nu_\ell(t)\rangle = e^{-iE_1t}\left[\left(|U_{\ell 1}|^2 + |U_{\ell 2}|^2\,e^{-i(E_2-E_1)t}\right)+|U_{\ell 3}|^2\,e^{-i(E_3-E_1)t} \right] \end{equation}
where we factored out the first phase contribution. Under the assumption of degeneracy of $2 \leftrightarrow 1$ states one can drop the phase arising from the second term of the sum. Using the unitarity of the mixing matrix, we have, 
\begin{equation} \label{p2f} P_{\ell \ell} = \left| \left(1-|U_{\ell 3}|^2\right)+|U_{\ell 3}|^2\,e^{-i(E_3-E_1)t} \right|^2 \end{equation}
that coincides with \eqref{2fla2mas}, obtained in the 2 flavor case. Note that the condition of degeneracy should not hold exactly; for our purpose it is sufficient that the milder condition $(E_2-E_1)t\ll 1$ holds true, which means that the distance between production and detection is small enough. In this sense, we can say that the two flavor formula can be used to describe (under suitable conditions) three flavor situations. It is worth noting that in this approximation, the oscillation probabilities of neutrinos and those of antineutrinos are the same, thus CP violation effects are not visible. This conclusion can be usefully presented also in another manner: the probability in \eqref{p2f} depends only on the three parameters $U_{\ell3}$, but one of them fixed by the unitarity condition $\sum_\ell |U_{\ell3}|^2=1$: thus, the oscillation probabilities depend only upon two physical mixing angles, as it is clear from the standard parameterization. This implies that the third mixing angle can be put to zero without changing the physics and the Jarlskog invariant $J_{CP}$ vanishes, as evident from \eqref{jarlsk}. 

\subsection{Oscillations in field theoretical formalism\label{sec:fiello}}
It is useful to develop in some 
detail the connection with  the 
field theoretical formalism. 
Let us consider the case when a neutrino is produced or detected by charged current weak interactions; in this case, 
the key quantity is just the matrix elements of the neutrino field between an initial (or final) neutrino (or antineutrino) and the vacuum.

Let us begin from the simple case of  a neutrino mass state. In this case, 
indicating explicitly the indices of mass $i,j$ and the indices of the 4-spinor $a,b$, 
the matrix element of interest is,
 \begin{equation} 
\langle 0 | (P_{\mathrm{L}})_{ab} \hat{\nu}_{a j}(x) | \nu_i,\vec{p}\, \rangle
\stackrel{\mathrm{UR}}{=}
\psi_a(\vec{x},\vec{p},-)\times \delta_{ij} e^{-i E_i t}
\label{pl1}
\end{equation}
where as usual we consider the ultrarelativistic limit, so that the one-to-one connection between helicity and chirality holds true; note that the presence of the chirality projector  leads to a negative helicity $-1$ for neutrinos.
The universal function $\psi(\vec{x},\vec{p},-)$ is given in 
\eqref{univrs}. 
 The matrix element shown in \eqref{pl1} 
is diagonal in the mass indices, as it is {\em a priori} evident from the fact that  the states with given mass and momentum are  stationary states; furthermore, it 
has the  
characteristic space-time dependence   of 
de Broglie wave, namely, $\sim e^{ i (\vec{p}\cdot\vec{x}-E_i t)}$.

Now, we compare this result with the one that we obtain when we have a state with given flavor and a field with another flavor. Using 
\eqref{mixfi} and \eqref{pera2}, we find,
 \begin{equation} 
 \begin{aligned}
 	\langle 0 | (P_{\mathrm{L}})_{ab} \hat{\nu}_{a \ell'}(x) | \nu_\ell,\vec{p}\, \rangle
&\stackrel{\mathrm{UR}}{=}
U_{\ell' j} U_{\ell i}^*
\langle 0 | (P_{\mathrm{L}})_{ab}\,\hat{\nu}_{b j }(x) | \nu_i,\vec{p}\,\rangle 
\\
&\stackrel{\mathrm{UR}}{=} \psi_a(\vec{x},\vec{p},-)\times U_{\ell' i} U_{\ell i}^* \,e^{-i E_i t}\\
&\stackrel{\hphantom{\mathrm{UR}}}{\equiv}
 \psi_a(\vec{x},\vec{p},-)\times 
\mathcal{U}_{\ell \ell'}(t) 
 \end{aligned}
\label{pl2}
\end{equation}
 In this matrix element, 
 the spinorial function $ \psi_a$
 is multiplied by transition amplitude
 $\mathcal{U}_{\ell \ell'}(t)$, that we have obtained previously.
 Comparing  \eqref{pl2} and \eqref{pl1}, 
 we see that the,  for flavor states, the 
 transition amplitude
  $\mathcal{U}_{\ell \ell'}(t)$ 
 plays the same role as  
 the ordinary phase factor $e^{-iE_i t}$ 
 for mass states.
 
 Of course, the above field theoretical  formalism leads 
 to the same result of the previous section. Moreover, it suggests the fact that the amplitude is a unitary matrix, as we can easily check by a direct computation. Using the definition of the amplitude in \eqref{bum2}, we have,
  \begin{equation} 
  \begin{aligned}
  	{\left[\mathcal{U}\,\mathcal{U}^\dagger\right]}_{\ell_1\ell_2}
  	& = 
 \sum_{\ell} \mathcal{U}_{\ell_1 \ell}\,\mathcal{U}_{\ell_2 \ell}^*
 =
 \sum_{\ell,j,k}   \left[ U_{\ell_1 j}^* \; e^{-iE_j t}\; U_{\ell j} \right]
  \left[ U_{\ell_2 k} \; e^{iE_k t}\; U_{\ell k}^* \right]=\\
  &= \sum_{j,k}   U_{\ell_1 j}^*\; e^{-iE_j t} \ 
  U_{\ell_2 k} \; e^{iE_k t} \ \delta_{jk}=
   \sum_{j}   U_{\ell_1 j}^*\, U_{\ell_2 j}  = \delta_{\ell_1 \ell_2}
  \end{aligned}
  \label{evoluto}
 \end{equation}
where we have repeatedly used the unitarity of the mixing matrix.\footnote{From a mathematical point of view, we can reach the same conclusion even more simply, 
noticing that  $\mathcal{U}(t)$ is a product of three unitary matrices: $\mathcal{U}(t)=U^*\mbox{diag}[ \exp(-i E_i t)] U^{\mathrm{t}}$.}
From a point of view of the physical interpretation, it is interesting to consider the diagonal term, 
 \begin{equation}
1=[\mathcal{U}\ \mathcal{U}^\dagger ]_{\ell \ell}=
\sum_{\ell'}  | \mathcal{U}_{\ell \ell'}|^2=\sum_{\ell'}  P_{\nu_\ell\to \nu_{\ell'}}
  \end{equation}
namely, the condition that the oscillation 
probabilities summed over all possible final states,  give just 1 --- that is, the neutrino does not disappear, it changes only flavor.

\subsection{The vacuum hamiltonians\label{sec:hhh}}
The neutrinos with given mass propagate according with the standard
free hamiltonian. In quantum field theory and with the notations of \eqref{definella} this hamiltonian for the $i$-th mass state is just,
\begin{equation}
\begin{aligned}
{\mathbf{H}}_{i}^{\mathrm{free}} &=
\int \mathrm{d}^3 x :\!\! \mathcal{H}_{i}^{\mathrm{free}}\!:\\
&=
\sum_{\vec{p}} \sqrt{p^2 + m_i^2}\sum_{\lambda=\pm 1}
[a_i^\dagger (\vec{p},\lambda)\,a_i  (\vec{p},\lambda) + 
b_i^\dagger (\vec{p},\lambda)\,b_i  (\vec{p},\lambda) ]
\label{calenda}
\end{aligned}
\end{equation}
two remarks are in order:
1.~we use the boldface to stress the operatorial character;
2.~the term with the $b$-oscillators should be omitted for Majorana neutrinos.
Therefore, we obtain 
the obvious solutions for the evolution of the energy  eigenstates $|\nu_i,t\rangle$ in the Schr\"{o}dinger representation, those 
used e.g., in \eqref{soluzella}.

The corresponding solutions for the flavor states, given again in \eqref{soluzella}, require to assume 
\eqref{pera1}, \eqref{pera2}: these equations, as discussed in \sectionname~\ref{sec:pera0}, are valid in the ultrarelativistic limit.  Under these conditions, we can usefully introduce a matrix on flavor space that describes the evolution of the state. This is obtained 
as follows, 
\begin{equation} 
i \frac{\mathrm{d}}{\mathrm{d}t} |\nu_\ell, t\rangle = U_{\ell j}^*
\; i \frac{\mathrm{d}}{\mathrm{d} t}  |\nu_j, t\rangle = 
U_{\ell j}^*
E_j  |\nu_j, t\rangle = 
U_{\ell j}^*
E_j  U_{\ell' j}  |\nu_{\ell'}, t\rangle
\end{equation}
which leads to,
\begin{equation}
i \frac{\mathrm{d}}{\mathrm{d}t} |\nu_\ell, t\rangle
	\equiv \left[ H_0 \right ]_{\ell \ell'} 
 |\nu_{\ell'}, t\rangle
\end{equation}
where of course the repeated indices are summed. In matricial notation, we have,
\begin{equation} 
 H_0= U^*\ \mathrm{diag}(E)\ U^{\mathrm{t}}\quad\text{where}\quad
 \mathrm{diag}(E)_{ij}=
 \begin{cases}
 	E_i & \text{if $i = j$}\\
	0 & \text{if $i \neq j$}\\
 \end{cases}
\end{equation}
With this matrix, the results of 
\eqref{bum1} or \eqref{bum2} can be presented\footnote{The calculation of $e^{-i H_0 t}$ can be performed by the Taylor expansion, $e^{-i H_0 t}=1+(-i H_0 t)+1/2  (-i H_0 t)^2+\dots$, noticing that $H_0^n= U^*\ \mathrm{diag}(E)^n\ U^{\mathrm{t}}$.} 
as,
\begin{equation} 
 |\nu_\ell, t\rangle =  \left[ e^{-i H_0 t} \right ]_{\ell \ell'}   |\nu_{\ell'}, 0\rangle \quad\text{or}\quad \mathcal{U}(t)=e^{-i H_0 t}
\end{equation}
respectively. 
Therefore, we see that the transition amplitude
can be regarded as an evolutor, in the quantum mechanical sense; compare also with \eqref{evoluto}.
For the antineutrinos, we need to replace $U_{\ell i}$ with $U_{\ell i}^*$ but otherwise the results are identical. 

In summary, we can describe the flavor transformation of neutrinos, caused by the effect of the relative phases of the neutrinos with given mass and by a non-trivial mixing matrix, 
by  introducing  the 
``effective hamiltonians'' valid in the ultrarelativistic limit, 
\begin{equation} 
H_{0,\nu}\equiv U^* \mathrm{diag}(E) U^{\mathrm{t}}\quad\text{and}\quad
H_{0,\bar{\nu}}\equiv U \mathrm{diag}(E) U^\dagger
\label{bibolott}
\end{equation}
for the propagation of free neutrinos and antineutrinos respectively, where we adopted the matricial notation and defined the diagonal matrix,
\begin{equation} 
\mathrm{diag}(E)_{ii} = E_i
\end{equation}
The ``effective hamiltonians'' $H_{0,\nu}$ and $H_{0,\bar\nu}$ are commonly called {\em vacuum hamiltonians} for two reasons:
	1.~to emphasize that the flavor transformation occur for neutrinos and antineutrinos that propagate in vacuum; and
	2.~to remark the difference with the (additional) matter hamiltonian term that will be introduced and discussed later on.
Note that, again in the ultrarelativistic approximation ($p \sim E$) one can write,
\begin{equation} E_i = \sqrt{m^2_i + p^2} \simeq p + \frac{1}{2} \frac{m^2_i}{E} \end{equation}
The first term gives rise to the diagonal matrix $p\times \mathbb{1}$ and this can be dropped, because it gives rise to an overall phase factor, common to all the flavor states and thus irrelevant for oscillations. Thus, one can equivalently write,
\begin{equation} \label{equib}
H_{0,\nu}= \frac{1}{2 E}U^* \mathrm{diag}(m^2) U^{\mathrm{t}}
\quad\text{or also}\quad
H_{0,\nu}= \frac{1}{2 E}U^* \mathrm{diag}(\Delta m^2) U^{\mathrm{t}}
\end{equation}
where e.g., $\mathrm{diag}(\Delta m^2)_{ii}=m_i^2-m_1^2$. 

Note that the ``effective hamiltonians''  depend explicitly upon the value of the momentum; thus, they should be thought more properly as sets of matrix elements of a true hamiltonian between states with given momentum --- i.e., plane waves. This is why the symbol $H_0$ is not in written in bold-face, as e.g., \eqref{calenda}.

\subsection{Oscillations and wave packets\label{sec:wp}}

The usage of the plane waves to derive oscillation formulae sometimes generates
confusion. In fact, it is not evident how to define the transit time or the distance between
production and detection with plane waves, that are not localized. We discuss here how to improve the description.

\subsubsection{Scalar wave packet}
Let us consider a packet of scalar waves  in one dimension,
propagating along the $z$ axis,
\begin{equation}
f(z,t)=\sum_{q} e^{i (q z -E_q t)} f(q)
\end{equation}
we recall that we assume that the particle is in a box so that its momenta are quantized. 
Suppose that this packet has typical momentum $p$ (namely, $|f(q)^2|$ is localized around the point  
$q=p$, where it attains its maximum) and then let us expand the energy around this point,
\begin{equation}
E_{p+q}\approx E_p + q v_p\quad\text{where}\quad
v_{p}=\frac{\mathrm{d} E_p}{\mathrm{d} p}
\end{equation}
We can change variable $q=p+k$ and rewrite,
\begin{equation}
f(z,t)\approx   e^{i (p z -E_p t)}  \sum_{k} e^{i k (z -v_p t)}   f(p+k)\equiv e^{i (p z -E_p t)}  F(z -v_p t)
\label{bagg}
\end{equation}
where we introduced the auxiliary function, $F(z)\equiv  \sum_{k} e^{i k z }   f(p+k)$. 
The two factors of \eqref{bagg} are amenable to the following interpretation: 
the first one describes a (de Broglie) plane wave, the second  one describes the position of the particle in the space. 
Two remarks are in order:
\begin{enumerate}[itemsep=-0.7ex,partopsep=1ex,parsep=1ex]
\item The function $f(z,t)$ is approximately an eigenstate of the the momentum $-i \:\mathrm{d}/\mathrm{d} z$ if 
$p\times | F(z-v_p t)|\gg \left| \mathrm{d} F(z-v_p t)/\mathrm{d} z \right| $, i.e., if 
the function $F(z-v_p t)$ does not vary much when it is non-zero. 
The same condition implies that $f(z,t)$ is also an approximate 
eigenstate of the energy, assuming the dispersion relation dictated by relativity.\footnote{First we multiply both sides by $v_p=v=p/E$ and note that the l.h.\ side $p^2/E\, | F(z-v t)|$ is smaller than $E | F(z-v t)|$. In the r.h.\ side, instead, we have $v\times  \left| \mathrm{d} F(z-v t)/\mathrm{d} z \right| =
 \left| \mathrm{d} F(z-v t)/\mathrm{d} t \right|$. Thus we obtain  $E | F(z-v t)|\gg \left| \mathrm{d} F(z-v t)/\mathrm{d} t \right|$, that implies that $f(z,t)$ is an approximate 
 eigenvalue of the energy $i\:\mathrm{d} /\mathrm{d} t$.}
 \item The modulus of the function is, $|f(z,t)|^2\approx |F(z -v_p t)|^2$. 
 Note that   in the ultrarelativistic regime in which we are particularly interested, we have $E\approx p c$ in good approximation and the velocity is just $v_p\approx c$. The wave packet becomes in this limit non-dispersive --- i.e.,  
 it propagates maintaining its shape.
 \end{enumerate}
 
 \subsubsection{Neutrino wave packet} Let us consider the state of a neutrino with flavor $\ell$ propagating 
 with momentum along the $z$ axis. 
 Summing over the component of the momentum along $z$, call it $q$, we have,
 \begin{equation}
 |\nu_\ell,g\rangle=\sum_{q} g(q)|\nu_\ell, q\rangle
 \end{equation}
 where $g(q)$ is a adimensional factor such that
 $\sum_q | g(q)|^2=1$. 
 Now we calculate the transition from this state to the vacuum, 
 following the  same manipulations described in
 \sectionname~\ref{sec:fiello}.
 We obtain, 
 \begin{equation}
 \begin{aligned}
 \langle 0| P_{\mathrm{L}} \hat{\nu}_{\ell'}(x)  |\nu_\ell,g\rangle &
 \stackrel{\mathrm{UR}}{=} U_{\ell' j}  U_{\ell j}^* \sum_q g(q) \psi(z,  q,-) e^{-i E_j t}\\
 & = U_{\ell' j}  U_{\ell j}^* \sum_q g(q) \frac{e^{i (q z- E_j t)}}{\sqrt{V}}\   u_-
 \label{piepolo}
 \end{aligned}
 \end{equation}
where $E_j=\sqrt{q^2+m_j^2}$. 
The ultrarelativistic 4-spinor of negative helicity\footnote{In the ``standard'' representation of the Dirac matrices, 
 this is just,
 \begin{equation}
 u_-=\frac{1}{\sqrt{2}}
 \left(
 \begin{array}{c}
 0\\
 +1\\
 0\\
 -1
 \end{array}
 \right)
 \end{equation} }
 $u_-$ factors out; so it does the dependence upon the helicity. Thus,  
the considerations concerning the scalar wave packet 
can be applied quite directly. Suppose that the typical momentum of the wave packet 
is $p$. Then, define the auxiliary wave function concerning a massive neutrino, 
\begin{equation}
\mathcal{G}(z)\equiv \sum_q \frac{e^{i q z}}{\sqrt{V}}\,g(p+q) \,u_-
\end{equation}
Finally, as in \eqref{bagg}, we conclude,
\begin{equation}
 \langle 0| P_{\mathrm{L}} \hat{\nu}_{\ell'}(x)  |\nu_\ell,g\rangle  \stackrel{\mathrm{UR}}{=} 
  U_{\ell' j} \  e^{i (p z-E_j t)} \  \mathcal{G}(z-v_j t)\  U_{\ell j}^*  
  \label{maina}
 \end{equation}
This is our main formula, where the formalism previously described 
has been enhanced to include a description of the  
wave packet. In fact, when we set $g(p+q)=\delta_{q,0} $,  we get 
$ \mathcal{G}\to \psi$, namely we have plane waves. If, furthermore, we replace $v_j\to c$, we end with \eqref{pl2}. Note that the r.h.\ side of \eqref{maina}
can be equated to the $\ell'$ component of the wave-function,  
\begin{equation}
[\Psi^{\ell}(t) ]_{\ell'}=U_{\ell' j}   e^{i (pz-E_j t)}  \mathcal{G}(z-v_j t) U_{\ell j}^*\label{parello}
 \end{equation}
 that describes  a neutrino $\ell$ at $t=0$ that evolves  in the course of the time.
Let us proceed with a detailed discussion of \eqref{maina}.
 
\begin{figure}[t]
\centering
   \includegraphics[width=0.7\textwidth,bb=0 0 319 214]{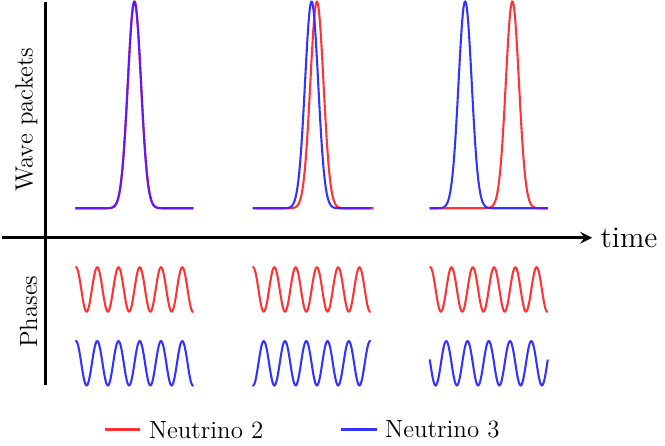}
   \caption{
   Physics of oscillations illustrated with the 
   \Pnum--\Pnut system. We assume the existence of 
   two mass components  $\nu_2$ and $\nu_3$, indicated 
    in red and blue respectively.  
   The the arrow indicates the direction of the time.
We display  the wave packets --- plots over the arrow --- and the phases  of the two 
   mass components --- plots below the arrow --- at three representative  times. 
   1.~Initially (left), the two wavepackets are overlapped; 
   until the two mass components are in phase, 
   the state is $|\Pnum\rangle$. 2.~Then (middle), 
    the two wavepackets are still overlapped; 
   when the two 
   mass components get counter-phased, the state becomes  $|\Pnut\rangle$. 3.~After a long time  (right), when the mass components are fully separated, we have a 50/50 flavor mixture.
   \label{an:fig.ill}}
\end{figure}
 
\subsubsection{Discussion of the formula}
Neutrino masses enter \eqref{maina} or equivalently 
\eqref{parello}
 through the energy 
$E_j=\sqrt{p^2+m^2_j}$ and the velocity
$v_j=p/E_j$. 
Thus, they produce two kinds of effects: the relative phases get different after a time $(E_1-E_2) t_{\mathrm{osc}} \sim 1$ and
the wave packets separate after a time $(v_1 -v_2) t_{\mathrm{dec}}\sim \sigma_z$, where $\sigma_z$ is the
size of the wave packet. 
These are two different effects: interferences connected to the phase change, namely 
neutrino oscillations and separation of the wave packets, namely decoherence. 
It is important to stress that oscillations occur when the components of the wave packet associated to the different mass states overlap. 
Using $\sigma_E\approx c\,\sigma_p\sim c/\sigma_z$,
we have,
\begin{equation}
t_{\mathrm{osc}}\sim t_{\mathrm{dec}} \times \frac{\sigma_E}{E}
\end{equation}
Therefore, if the energy uncertainty $\sigma_E$ is much less than the value of the energy $E\approx p c$,
neutrino oscillations will happen earlier and decoherence will happen later. 

Equivalently, we can carry on the discussion by introducing and comparing  three lengths: propagation length (i.e., distance between production and detection) $L=c\,t$; oscillation length $L_{\mathrm{osc}}=c\, t_{\mathrm{osc}}$; 
coherence length  $L_{\mathrm{dec}}=c\, t_{\mathrm{dec}}$.
Similarly, it is useful to note 
the two phenomena mentioned above corresponds to 
different velocities among the individual components: oscillations 
are caused by the difference  among 
phase velocities $E_j/p$; decoherence is caused by 
the difference among group velocities $\mathrm{d} E_j/\mathrm{d} p=p/E_j$ instead.

{\subsubsection{Summary and illustration} At this point we have a general picture of 
what happens to ultrarelativistic neutrinos during their 
propagation:
\begin{enumerate}[itemsep=-0.7ex,partopsep=1ex,parsep=1ex]
	\item The light neutrinos are produced according to
	the mixing matrix, since we assume that their kinetic  energy is much larger than their mass.
	\item Initially, we can assume
in excellent approximation that the  components with different masses, and thus also their wave packets, travel
with velocity $c$.  
Thus it is appropriate to use $\mathcal{G}(z-ct)$ for all mass components in \eqref{maina} or \eqref{parello}.
Then, we rewrite \eqref{parello} as follows,
\begin{equation}
[\Psi^\ell]_{\ell'}= \mathcal{U}_{\ell \ell'}(t) \  \mathcal{G}(z-ct)
\end{equation}
and we see that the last, flavor independent factor, that describes the shape of the wave-function, is multiplied by 
the usual transition amplitude 
\begin{equation}
\mathcal{U}_{\ell \ell'}(t) =U_{\ell j}^*\,  e^{i (p z - E_j t)}\, U_{\ell' \! j} 
\end{equation}
The key feature of this phase is that the  
differences of phases $(E_i-E_j)t$ deviate from zero and oscillations in proper sense occur.
\item[$2'.$] Subsequently,  the differences of phases of the components 
$(E_i-E_j)t$ appearing in $\mathcal{U}_{\ell \ell'}$ 
become very large. Thus, the phase factors oscillate very rapidly when we vary, even slightly, the distance or the energy at the source or at the detection.  For this reason, the effect of the oscillatory terms 
is not anymore measurable in practice.
\item Eventually, the wave packets separate and the description   
becomes even simpler conceptually.
When the different components do not overlap anymore, oscillations are completely lost (i.e., not only for practical purposes). Interestingly,  neutrinos with different masses could be in principle  detected separately in this stage even if, in practice, this is extremely difficult. If they are not measured, the probability of transmutation that follow from   step $2'.$ and step 3.\ are the same.
\end{enumerate}
The steps 1., 2.\ and 3.\ are illustrated graphically in 
\figurename~\ref{an:fig.ill}.

\section{Applications and examples\label{sec:ae}}

  \begin{figure}[t]
  \centerline{
  \includegraphics[width=0.8\linewidth,bb=0 0 800 535]{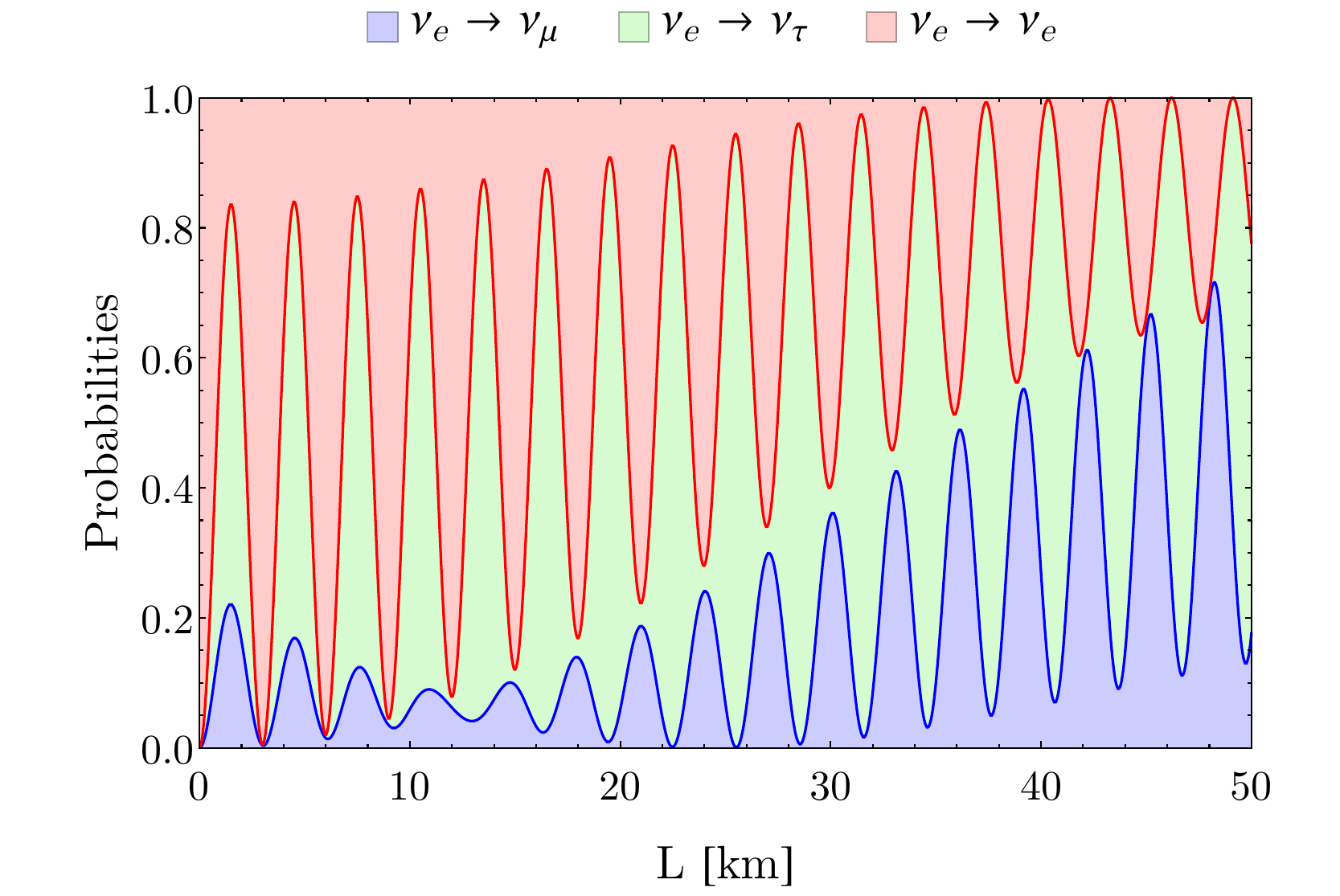}} 
   \caption{In this plot we consider 
   an electron neutrino of \SI{3}{\mega\electronvolt} that propagates for several kilometers. For each fixed distance, 
   the probability that the electron neutrino   
   becomes a muon (or tau) neutrino 
   in the course of the propagation  
   is given by the blue (or green) region. The probability that it remains an electron neutrinos is given by the red region. The plot emphasizes that the sum of the three probabilities is one, i.e., there is no loss of probability. 
      \label{an:fig.bl}}
\end{figure}

We will proceed to illustrate the formulae given above by considering 
several specific applications and particular cases. However, we would like to begin  simply by justifying better the name of  ``oscillations'' given to this phenomenon, especially in the modern literature. 
We have seen the appearance of a several 
oscillatory functions, however this point is illustrated much better by a plot. In \figurename~\ref{an:fig.bl}, we show the  probabilities that an electron neutrino remains such or changes in the course of its propagation, with the parameters given in \tablename~\ref{tab:lisi16}.

\subsection{Why two flavor formulae are so useful \label{ss:2f}}
Two flavor oscillation formulae in vacuum allow us to discuss the main facts concerning 
the observed neutrino oscillations. In order to see how, recall that:
	1.~electron antineutrinos are produced in nuclear fission  
reactors and then detected; similarly
2.~muon neutrinos are produced by charged pion decays (either naturally or artificially) and then detected.

 In both cases, we are interested in the probability of survival. Let us consider,
\begin{equation}
|\APnue\rangle =U_{\ee i} | \bar{\nu}_i\rangle
\quad\text{and}\quad
|\Pnum \rangle = U_{\mu i}^* | {\nu}_i\rangle
\end{equation}
When the distances are not too large, as quantified below, only  the third neutrino, the one that has the larger mass difference with the other two states, causes oscillations. 
The other two neutrinos mass states have effectively the same mass. 
In other words, the discussion of Section~\ref{3come2} applies to the real situations, due to the value of the oscillation parameters given in \tablename~\ref{tab:lisi16}.
  Let us recall the amplitudes of oscillations of interest, namely, 
\begin{equation}
\langle \APnue| \APnue, t\rangle =  \sum_{i=1}^3 |U_{\ee i}|^2 e^{-i E_i t}
\approx e^{-i E_1 t}\left[1- |U_{\ee 3}|^2  +  |U_{\ee 3}|^2 e^{i (E_3 -E_1) t}\right]
\end{equation}
where we used the above assumption, along with unitarity. The 
two flavor formulae follow,
\begin{eqnarray}
P_{\apnue \to \apnue }&=&1-4 |U_{\ee 3}^2| \left( 1 -|U_{\ee 3}^2|
\right) \sin^2\varphi\\
P_{\pnum \to \pnum }&=&1-4 |U_{\mu 3}^2| \left( 1 -|U_{\mu 3}^2| \right) 
\sin^2\varphi
\end{eqnarray}
and adopting the standard parameterization
(\eqref{stdd} and \tablename~\ref{tab:lisi16})
we can replace,
\begin{equation}
|U_{\ee 3}|=\sin\theta_{13}
\quad\text{and}\quad|U_{\mu 3}|=\cos\theta_{13} \sin\theta_{23}
\end{equation}
Using these simple formulae,
Daya Bay  and Reno collaborations (reactor experiments)
measured  $\theta_{13}\sim 9^\circ$; 
Super-Kamiokande and MACRO (atmospheric neutrino 
experiments\footnote{Neutrinos are produced at a height of 
some $H=15-\SI{20}{\kilo\meter}$ height in the atmosphere. Those coming from the horizontal direction, travel $L\approx (2 H R_\oplus)^{1/2}=450-\SI{500}{\kilo\meter}$ while the vertical ones travel much less.}) and subsequently K2K and MINOS (accelerator experiments) 
measured $\theta_{23}\sim 45^\circ$.
Due to the fact that $\Delta m^2_{31}\approx 
\SI{2.5e-3}{\electronvolt\squared}$,
oscillations manifested in these experiments, and indeed the phases in the conditions of these experiments are large enough,
\begin{equation} \small
\varphi=1.267\; \frac{\Delta m^2_{31} L}{E_\nu} \approx 1
\quad\text{if}\quad
\begin{cases}
	L = \SI{1}{\kilo\meter},\ E_\nu = \SI{3}{\mega\electronvolt}
	& \small \text{[reactor $\nu$]}\\
	L = \SI{700}{\kilo\meter},\ E_\nu = \SI{2}{\giga\electronvolt}
	& \small \text{[long baseline]}\\
\end{cases}
\end{equation}
The average value of the electron survival probability for solar neutrinos,   
that follows from the above formula including the effect of $\theta_{13}$ is, 
is $\langle P_{\pnue\to \pnue}\rangle =
\langle P_{\apnue\to \apnue}\rangle = 
1-(\sin^2 2\theta_{13})/2
=0.95$: this is too small to account for solar neutrino disappearance. In fact, it is the effect of the splitting between the other two mass states that explains solar neutrinos!
This hypothesis 
was tested by KamLAND reactor experiments at distances of the order of $L\sim \SI{100}{\kilo\meter}$,  
that determined precisely $\Delta m^2_{21}\approx
\SI{7.4e-5}{\electronvolt\squared}$ using the formula,\footnote{This equation follows by averaging to zero the fast oscillation due to $\Delta m^2_{31}$, namely, adopting the following approximation, 
$ P_{\apnue \to \apnue } = 
| \; 
|U_{\ee 1}^2|  + |U_{\ee 2}^2| e^{-i2\; \varphi_{21}}  + |U_{\ee 3}^2| e^{-i2\;  \varphi_{31}}
\, |^2\approx
| \; 
|U_{\ee 1}^2|  + |U_{\ee 2}^2| e^{-i2\; \varphi_{21}}\, |^2  + |U_{\ee 3}^4|
$.}
\begin{equation}
P_{\apnue \to \apnue }=\sin^4\theta_{13} + \cos^4\theta_{13} ( 1  - \sin^2 2\theta_{12} \sin^2\phi) 
\end{equation}
where this time $\phi=1.267\,\Delta m^2_{21}\, L/E_\nu$. 
The value of the last angle,
that describes the composition of $\nu_e$ in terms of the last two mass states, is $\theta_{12}\approx 33^\circ$.
The ambiguity $\theta_{12} 
\leftrightarrow 90^\circ-\theta_{12}$
due to the occurrence of 
$\sin^22\theta_{12}$ 
in the formula is resolved by high energy solar neutrinos, that are affected by 
matter (MSW) effect and will be discussed in the next section.

The above considerations show that, despite their simplicity, 
the two-flavor formulae (or their direct extensions) allow us to understand a great deal of facts concerning the evidences of neutrinos oscillations. See 
\tablename~\ref{tab:lisi16} for an updated summary of the values of the oscillation  parameters.

\subsection{A special case: maximal mixing\label{pipin}}

\subsubsection{Definitions}
The cases of maximal mixing are of special interest. 
We begin the discussion with the two flavor case.
One assumes that a pair of neutrinos, say $\mu$ and $\tau$, is described by,
\begin{equation}
|\Pnum\rangle= \frac{|\nu_1\rangle+|\nu_2\rangle}{\sqrt{2}}
\quad\text{and}\quad
|\Pnut \rangle= \frac{-|\nu_1\rangle+|\nu_2\rangle}{\sqrt{2}}
\end{equation}
namely,
\begin{equation}
|\nu_\ell\rangle=U_{\ell i}^* |\nu_i\rangle
\quad\text{with}\quad 
U=
\frac{1}{\sqrt{2}}
\left(
\begin{array}{cc}
+1 & +1 \\
-1 & +1
\end{array}
\right)
\end{equation}
so that the mixing elements satisfy $|U_{\ell i}|^2=1/2$. This simple case is a good approximation of some situations of physical interest.
The maximal mixing in the $3\times 3$ case is a bit more intricate. It is  based on the assumption that the mixing matrix is,
\begin{equation}
U=
\frac{1}{\sqrt{3}}
\left(
\begin{array}{ccc}
1 & 1 & 1 \\
1 & \omega & \omega^2 \\
1 & \omega^2 & \omega \\
\end{array}
\right)
\quad\text{where}\quad\omega=e^{i\; 2 \; \pi/3}
\label{sopra}
\end{equation}
so that $\omega^3=1$ and all matrix elements satisfy $|U_{\ell i}|^2=1/3$. This case is to be thought as a toy model since it does not correspond to any situation of physical reality but it is anyway useful as test bed, in particular to better understand CP violating phenomena.
In fact, we find 
$J_{\mathrm{CP}}=
\mbox{Im}[ U_{\ee 1} U_{\mu 1}^* U_{\ee 2}^* U_{\mu 2}]=1/(6 \sqrt{3})$, namely, the maximal amount of CP violation, as discussed after 
\eqref{jarlsk}. In fact, it can be checked by direct calculation that,
by adopting a suitable redefinition of phases as in \eqref{ridif}, 
the mixing matrix given in \eqref{sopra} 
is {\em equivalent} to the standard form 
given in \eqref{stdd} with the values 
$\theta_{12}=\theta_{23}=\pi/4$, $\sin\theta_{13}=1/\sqrt{3}$ and $\delta=+ \pi/2$.

\subsubsection{Two flavor case}
Using the oscillation probability of \eqref{Pll'} and the definition of maximal mixing, it is easy to show that,
\begin{equation}
P_{\mu \tau}=P_{ \tau\mu}
=\sin^2 \phi=
1-P_{\mu\mu}=1-P_{\tau\tau}
\quad\text{where}\quad\phi= \frac{\Delta m^2 L}{4 E}
\end{equation}
Thus:
\begin{enumerate}[itemsep=-0.7ex,partopsep=1ex,parsep=1ex]
\item the probabilities of appearance and disappearance are all connected;
\item the averaged values are $\langle P_{\ell \ell'} \rangle  =1/2$ for $\ell,\ell'=\mu,\tau$;
\item it is possible that the original type of neutrino fully disappears, when $\phi=\pi/2$.
\end{enumerate}

\subsubsection{Three flavor case} 
We recall that (by definition) 
all matrix elements satisfy $|U_{\ell i}|^2=1/3$; moreover, it is interesting to impose also the additional condition,\begin{equation}\phi=\phi_{21}=\phi_{32}=\phi_{31}/2\end{equation} 
or equivalently
$\Delta m^2=m_2^2-m_1^2= m_3^2-m_2^2=(m_3^2-m_1^2)/2$. This is a rather particular case that is not 
realized in nature but has some didactic interest for what concerns CP (or better T) violation.
We have,
\begin{equation}
P_{\ell\ell}=\frac{(1+2 \cos\phi)^2}{9}\quad\text{and}\quad
P_{\ell\ell'}=\frac{1-P_{\ell\ell}}{2} + \frac{2\; \xi}{3 \sqrt{3}}\sin\phi \,
(1-\cos\phi)
\end{equation}
where $\ell\neq \ell'$. The last term describes CP-violation and the sign $\xi$ is,
\begin{equation}
\xi=-1\quad\text{for}\quad\ell\ell'=\ee\mu,\mu\tau,\tau
e\quad\text{and}\quad\xi=+1
\quad\text{for}\quad \ell\ell'=\mu \ee,\tau\mu,\ee \tau 
\end{equation}
It is noticeable that when $\phi=2\pi/3$ we have,
\begin{equation}
\left(
\begin{array}{ccc}
P_{\ee\ee } & P_{\ee \mu } & P_{\ee \tau } \\
P_{\mu \ee } & P_{\mu \mu } & P_{\mu \tau } \\
P_{\tau \ee } & P_{\tau \mu } & P_{\tau \tau } 
\end{array}
\right)=
\left(
\begin{array}{ccc}
0& 0& 1 \\
1 & 0 & 0 \\
0 & 1 & 0 
\end{array}
\right)
\end{equation}
therefore, the chosen combination of mixing angles and of phases of oscillations produces the maximum possible deviation 
from the case when the oscillation probability is symmetric, $P_{\ell \ell'}=P_{\ell' \ell}$.
In other words, this is the maximum amount of T violation, that is, the   maximum difference  between the cases when $\nu_\ell\to \nu_{\ell'}$ and $\nu_{\ell'}\to \nu_{\ell}$ with $\ell\neq \ell'$.

\subsection{A 2 flavor case of oscillations with neutrino interactions}
When we have muon and tau neutrinos of high energies
we should  consider the fact  that they oscillate, approximatively  with maximal mixing $\theta_{\mbox{\tiny max}}=45^\circ$, and, also, they interact (differently) with matter. 
(This is relevant, for example, for the neutrinos produced in the center of the Sun due to the annihilation or decay of hypothetical dark matter particles \cite{aless}.)
After charged current interactions, neutrinos convert into the corresponding charged leptons that are eventually slowed down and absorbed. This can be formally described including an imaginary term in the hamiltonian, that  accounts for their disappearance at the original energy.\footnote{Considering again the analogy with optics, this term is, in principle, similar to the complex term that is introduced in the refractive index definition, $n\equiv n_1+in_2$, when a light signal propagates in a dispersive medium; $n_2$ accounts for the possible absorption of the signal. There is, also, a stringent analogy with the kaon system, where the complex term in the hamiltonian allows the disappearance of the particle, e.g.\ 
for a decay process --- see \eqref{bambulo}. Note that we consider  charged current interactions; the inclusion of neutral current interactions is more complicated, because a term that describes the appearance of a neutrino in the final state should be introduced.}
  Thus,
\begin{equation}
H=U(\theta_{\mathrm{max}}) 
\left(
\begin{array}{cc}
-\kappa & 0 \\
0 & \kappa 
\end{array}
\right) U^{-1}(\theta_{\mathrm{max}})-  \frac{i}{2 } 
\left(
\begin{array}{cc}
\Gamma_\mu & 0 \\
0 & \Gamma_\tau
\end{array}
\right)
\end{equation}
with,
\begin{equation}
\kappa=\frac{\Delta m^2}{4 E}
\end{equation}
and where $\Gamma_{\mu,\tau}$ are the interaction rates with matter.
Setting,
\begin{equation}
\Gamma=\frac{\Gamma_\mu+\Gamma_\tau}{2}\qquad
\gamma=\frac{\Gamma_\mu-\Gamma_\tau}{4}
\end{equation}
we get the convenient form of the hamiltonian,
\begin{equation}
H=
\left(
\begin{array}{cc}
-i\gamma & \kappa \\
\kappa & +i \gamma
\end{array}
\right)-  \frac{i \Gamma}{2 } 
\left(
\begin{array}{cc}
1 & 0 \\
0 & 1
\end{array}
\right)
\end{equation}
In the case of constant density, we can easily calculate the evolutor $\mathcal{U}=\exp(-i H r)$, where $r$ is the distance
travelled by the neutrino, proceeding as follows. The second part,  proportional to the unit matrix, is trivially exponentiated. For the rest, let us begin by  assuming the parameter $\gamma$ as imaginary.
The resulting matrix is hermitian and can be easily diagonalized.  Its eigenvalues $\lambda_{\pm}$ and mixing angles are,
\begin{equation}\lambda_{\pm}=\pm \alpha\ , \ 
\cos2\theta=\frac{i\gamma}{\alpha}, \ 
\sin2\theta=\frac{\kappa}{\alpha}
\quad\text{where}\quad\alpha=\sqrt{\kappa^2-\gamma^2}
\end{equation}
At this point, it is easy to exponentiate the full hamiltonian. The result is,
\begin{equation}
\mathcal{U}=
e^{-\Gamma r/2} 
\left(
\begin{array}{cc}
\cos\varphi -\frac{\gamma}{\alpha} \sin\varphi & -i \frac{\kappa}{\alpha} \sin\varphi  \\
-i \frac{\kappa}{\alpha} \sin\varphi  & \cos\varphi +\frac{\gamma}{\alpha} \sin\varphi  
\end{array}
\right)
\quad\text{with}\quad\varphi=\alpha r
\end{equation}
Finally, we use analytical continuation to come back to the 
case when $\gamma$ is real, i.e., the case in which we are actually interested. We find the explicit formula for $\Pnum\to \Pnut$ conversion,
\begin{equation}
P_{\pnum\to \pnut}=
\frac{e^{-r\, \Gamma }}{1-\varepsilon^2}\ \sin^2\!\left(\displaystyle\frac{\Delta m^2\; r}{4\;  E_\nu} \sqrt{1-\varepsilon^2} \right)
\quad\text{with}\quad\varepsilon=\frac{\Gamma_\mu-\Gamma_\tau}{\Delta m^2/E_\nu}
\end{equation}
It is possible to find in Ref.~\cite{aless} the numerical values of the interaction rates $\Gamma$ and more discussion of this and similar situations.

\subsection{Electron survival probability in three flavors}
The electron survival probability in the full three flavor regime contains an imprint of the type of mass hierarchy. 
The goal of the future experiment JUNO is just 
to probe this effect, see Ref.~\cite{capo2} for a very detailed discussion.
Here we derive the oscillation probability and demonstrate this effect. The formal manipulations are slightly more demanding than the previous ones and for this reason we use a brief and special notation.
The oscillation amplitude for $\APnue\to \APnue$ in 3 flavors is,
\begin{equation}
\begin{aligned}
\mathcal{U}_3 &= e^{-iE_3t}\left[C_{12} C_{13} \,e^{i f_{31}} + S_{12} C_{13}\, e^{i f_{32}}   + S_{13}\right] \\ &=e^{-iE_3t}\left[C_{13}\; \mathcal{U}_2 + S_{13} \right]
\end{aligned}
\end{equation}
where the corresponding expression in 2 flavor (or, if $\theta_{13}=0$) is,
\begin{equation}
\mathcal{U}_2=C_{12}\, e^{i f_{31}} + S_{12} \, e^{i f_{32}}  
\end{equation}
and where we simplified the notation by using the standard 
parameterization and setting,
\begin{equation}
C_{ij}=\cos^2\theta_{ij}, \ 
S_{ij}=\sin^2\theta_{ij}\quad\text{and}\quad f_{ij}=( E_i-E_j ) t\approx \frac{m_i^2-m_j^2}{2 E} L
\end{equation}

In this manner we get,
\begin{equation}
\begin{aligned}
P_3&= C_{13}^2 |\mathcal{U}_2|^2 + S_{13}^2 + 2\, C_{13}  S_{13} \,\mbox{Re}[ \mathcal{U}_2]\\
&= C_{13}^2 P_2 + S_{13}^2 + 2\, C_{13}  S_{13}\,{\mbox{Re}[
\mathcal{U}_2]}
\end{aligned}
\end{equation}
where we set by definition,
\begin{equation}
P_2=|\mathcal{U}_2|^2\quad\text{and}\quad P_3=|\mathcal{U}_3|^2
\end{equation}
The piece in which we are mostly interested is,
\begin{equation}
\mbox{Re}[ \mathcal{U}_2]=C_{12} \cos( f_{31})  + S_{12} \cos( f_{32})  
\end{equation}
The two large phases of oscillations, that is  $ f_{31}$ and 
 $f_{32}$, differ but only by a small amount. This is given by
the phase of the `solar' (or better KamLAND's) neutrino oscillations, 
\begin{equation}
f_{31}-f_{32}=f_{21}\equiv f 
\end{equation}
It is convenient to introduce a 
special choice of the {\em common part} of the 
large phase factor $\Phi$, such that 
the slow and the fast oscillations get neatly separated.
This choice is \cite{nuno}, 
\begin{equation}
 f_{31}=  \alpha\,\Phi + S_{12}\, f\quad\text{and}\quad
 f_{32}=  \alpha\,\Phi -  C_{12}\, f
\end{equation}
where $\alpha=\pm 1$ corresponds to normal/inverse hierarchy, respectively. 
In fact, we have, 
\begin{equation}
	\begin{aligned}
	\mbox{Re}\left[ \mathcal{U}_2\right]  =& C_{12} \cos\left( \Phi + \alpha S_{12} \, f\right)  + S_{12} \cos\left(  \Phi -  \alpha C_{12} \, f\right)=\\
	 =& \cos  \Phi\ \left[\ C_{12}  \cos( S_{12} \, f )+S_{12}  \cos( C_{12} \, f) \right]+\\&-\alpha
\sin \Phi\;\left[\ C_{12}  \sin\left(  S_{12} \, f \right)- S_{12}  \sin\left( C_{12} \, f \right)\right]
\end{aligned}
\end{equation}
This expression is consistent with the following definition of $\varphi$,
\begin{equation}
\begin{array}{l}
 \cos\varphi\ |\mathcal{U}_2| \equiv   C_{12}  \cos\left( S_{12} \, f \right)+S_{12}  \cos\left( C_{12} \, f \right) \\
 \sin\varphi\ |\mathcal{U}_2|  \equiv C_{12}  \sin\left( S_{12} \, f \right)- S_{12}  \sin\left( C_{12} \, f \right)
\end{array}
\end{equation}
This special definition allows us to recast  the term of interest in a 
very neat manner,
\begin{equation}
{\mbox{Re}\left[ \mathcal{U}_2\right]} = {|\mathcal{U}_2|}\times \cos\left(  \Phi+\alpha \varphi  \right)
\end{equation}
In this manner, we have separated the effects of the 
fast and the slow oscillations in the terms ${\mbox{Re}[ \mathcal{U}_2]}$: 
the phase $\Phi$ is large and produces the fast oscillations, whereas 
the phase $\varphi$ is small, and varies  only with  
the solar phase $f$, just as the term 
$|\mathcal{U}_2|=\sqrt{P_2}$.

Summarizing, our expression for the electron antineutrino survival probability is,\footnote{In order to compare with  
Ref.~\cite{capo}, the following replacements are needed 
$\Phi\to 2 \Delta_{\ee\ee}$, $f\to 2 \delta$, $C_{ij}\to c_{ij}^2$, $S_{ij}\to s_{ij}^2$.}
\begin{equation}
\begin{aligned}
P_3 &= C_{13}^2 |\mathcal{U}_2|^2 + S_{13}^2 + 2\,C_{13}  S_{13} \,\mbox{Re}[ \mathcal{U}_2] \\
& = C_{13}^2 P_2 + S_{13}^2 + 2\, C_{13}  S_{13} \sqrt{P_2}   
\cos( \Phi+\alpha \varphi  ) 
\end{aligned}
\end{equation}
The last terms shows that in the region of 
$\Phi$-driven  fast  oscillations, 
 the prediction of the 
 two mass hierarchies are slightly `out of tune' between them; 
 the detection of this effect requires a very good energy resolution.
 Moreover, the overall factor $S_{13}=\sin^2\theta_{13}$ suppresses 
 this effect, which contributes to 
 make the task of JUNO demanding.

 \subsection{The limit of fully averaged oscillations\label{sec:ave}}
 
 An important limiting case is when we have vacuum oscillations and 
 the phases of oscillation are very large. In this case, the oscillatory (cosinusoidal) terms get averaged to zero in any practical situation: e.g., when the production or the detection points are not perfectly known, when there is a distribution of initial energies and this is not perfectly measured in the detector, {\em etc.} The case is known as {\em averaged oscillations} and sometimes also as Gribov-Pontecorvo regime \cite{gp}. Strictly speaking, one cannot speak of `oscillations' in proper sense. However, this case is of wide physical interest: for instance, it is important for low energy solar neutrinos (when `matter effects' are negligible) or also for high energy neutrino from cosmic sources, as those that have been plausibly measured with IceCube. 
 
 The formulae for averaged oscillations are very simple and were given in Ref.~\cite{bilenky},
 \begin{equation}
 P_{\nu_\ell\to \nu_{\ell'}}= P_{\bar{\nu}_\ell\to {\bar \nu}_{\ell'}} =
 \sum_{i=1}^{3} | U_{\ell i} |^2 
  | U_{\ell' i} |^2 
  \end{equation} 
Of course the averaged probabilities are subject to unitarity constraints as any other set of oscillation probabilities, e.g.: 
$\sum_{\ell'=\ee,\mu,\tau}  P_{\nu_\ell\to \nu_{\ell'}}=1$, or just as the 
elements of the mixing matrix  squared $| U_{\ell i} |^2$.  
However, 
owing to the symmetry  in the exchange $\ell\leftrightarrow \ell'$, 
the averaged probabilities 
depend only upon three parameters, and not four as 
 $| U_{\ell i} |^2$. The 
 probabilities can be arranged 
 in a symmetric matrix,
\begin{equation} \label{puppalo}
\mathcal{P}=\left( \begin{array}{ccc}
\frac{1}{3}+2 P_0 & \frac{1}{3}-P_0+P_1 & \frac{1}{3}-P_0-P_1 \\
 & \frac{1}{3}+\frac{P_0}{2}-P_1+P_2 & \frac{1}{3}+\frac{P_0}{2}-P_2 \\
 &  & \frac{1}{3}+\frac{P_0}{2}+P_1+P_2 \end{array} \right) 
 \end{equation} 
 where we introduced 3 new parameters, that  
 can be expressed in terms of the usual ones as follows,
\begin{eqnarray}
	P_0 &=& \frac{1}{2} \left\{ (1-\epsilon)^2 \left[ 1- \frac{\sin^2 2\theta_{12}}{2} \right] + \epsilon^2-\frac{1}{3} \right\}\\
	P_1 &=& \frac{1-\epsilon}{2} \left\{ \gamma \cos 2\theta_{12}+\beta \frac{1-3 \epsilon}{2} \right\}\\
	P_2 &=&\frac{1}{2} \left\{ \gamma^2+\frac{3}{4} \beta^2 (1-\epsilon)^2 \right\} 
\end{eqnarray}
and the small quantities $\alpha,\beta,\gamma,\epsilon$ (only the last one is known to date) are,
\begin{equation}
	\begin{cases}
		\epsilon=\sin^2 \theta_{13}\\
		\alpha=\sin \theta_{13} \cos \delta \sin 2 \theta_{12} \sin 2 \theta_{23}\\
		\beta=\cos 2 \theta_{23}\\ \displaystyle
		\gamma=\alpha-\frac{\beta}{2} \cos 2\theta_{12} (1+\epsilon)
	\end{cases}
\end{equation}
It is interesting that the averaged 
oscillation probabilities depend upon the CP-violating phase and also upon the departure of $\theta_{23}$ from $45^\circ$, through $\cos2\theta_{23}$. However, 
 there are only 3 free parameters, thus one combination of the usual parameters is degenerate. 
 More precisely
$P_0$ is well known; $P_1$ contains most uncertainties; 
$P_2$ is positive and very small. 
The above results are taken from Ref.~\cite{dgdg}, where further discussion, references, applications and numerical results can be found.

	\chapter{Neutrino oscillations in matter\label{ch:mo1}}
\minitoc

In this section, we introduce and discuss 
the matter term of Wolfenstein. Then we present a 
few analytical solutions of the equation of propagation, including the one discussed by Mikheyev and Smirnov. The aim is to  illustrate the characteristic features of oscillations due to the matter term.  Finally, we consider some physical  applications.

\section{The matter (or Wolfenstein) term}
\subsection{Origin}

Vacuum oscillations are caused by the relative {\em phases} that the components of a neutrino with given flavor obtain in the course of the time $t$. This is a typical undulatory or quantum phenomenon, as it is evident from the  connection between energies $E_i$ of each component 
and the frequencies of oscillation $\nu_i$  of each component: $E_i=h\, \nu_i$. Indeed, as we have shown in \sectionname~ \ref{sec:hhh},  
vacuum oscillations can be described by an effective hamiltonian.
In this section we discuss a new point:  electron neutrinos that propagate in a material medium receive a special, additional phase of scattering, as first remarked by Wolfenstein. We can think to this new phase as an additional term in the hamiltonian of propagation, or also as a flavor-specific refraction term.

 Consider the scattering amplitude of 
 neutrinos onto electrons at rest 
 and more specifically consider the case when the neutrino momentum is unchanged (=forward scattering),  
\begin{equation}
\nu_\ell(\vec{p}\,)\ \Pelectron \to \nu_\ell(\vec{p}\,)\ \Pelectron
\end{equation}
This amplitude of scattering 
is maximum when $\ell=\Pe$, due to charged current interactions.
Thus, electronic neutrinos receive a special phase, that contributes to the effective hamiltonian and eventually also to oscillations.  
Also neutral currents give an additional phase due to the forward scattering of neutrinos with any particle  $X$ of the medium
 (excepting neutrinos themselves) 
$\nu_\ell(\vec{p}\,) X \to \nu_\ell(\vec{p}\,) X$:  
however this contribution is the same for all types of neutrinos and therefore it contributes only to an irrelevant overall phase factor for all neutrinos.\footnote{Two remarks are in order:
1.~If we have also new (sterile) neutrinos that do not interact, these  
phases become relevant. 
2.~If the density of neutrinos is very high, the scattering of neutrinos onto neutrinos can cause additional effects; however their complete description still eludes our understanding.} 

The most important case when this new phase is  relevant
is when the matter, where neutrinos propagate,  contains electrons at rest or almost at rest. We will focus the subsequent discussion on this case. 
However, in principle, one can have also effects due to polarization, to collective motions, etc.

\subsection{Formal derivation of the matter term}

We are interested to describe 
the contribution of Fermi interactions to the forward scattering of 
an electron neutrino. The idea is to 
evaluate the amplitude of forward scattering in terms of the matrix element of the hamiltonian.
Thus, we need a formal  description of the hamiltonian, of the 
state of the neutrino and of the state of the medium.  
We show how to derive, under reasonable assumptions, the standard result (namely, the Wolfenstein or matter term).
 
\subsubsection{Fermi hamiltonian for neutrino-electron interaction}
The hamiltonian density of the Fermi interactions contains the following interaction between
electrons and electron neutrinos,
\begin{equation}
\mathcal{H}_{\nu_{\ee\ee}}^{\mathrm{cc}}=+\frac{G_\mathrm{F}}{\sqrt{2}}\ 
\bar{\nu}_{\ee}\gamma_a (1-\gamma_5) \ee \ \bar{\ee} \gamma^a (1- \gamma_5) \nu_{\ee}
 \end{equation}
where $G_{\mathrm{F}}> 0$ whose sign is predicted by the standard electroweak model. This implies the following hamiltonian,
written in terms of relativistic quantum fields,
\begin{equation}
\mathbf{H}_{\nu_{\ee\ee}}^{\mathrm{cc}}=\sqrt{2}\, G_{\mathrm{F}} \int \mathrm{d}^3 x\ 
\bar{\nu}_{\mathrm{eL}}\gamma_a \nu_{\mathrm{eL}} \ \bar{\ee} \gamma^a (1- \gamma_5) \ee
 \end{equation}
 where:
\begin{enumerate}[itemsep=-0.7ex,partopsep=1ex,parsep=1ex]
 \item we have integrated over the space coordinates to pass from the hamiltonian density to the hamiltonian operator;
 \item we have introduced the chiral projector $P_{\mathrm{L}}=(1-\gamma_5)/2$ and defined $\nu_{\ee\mathrm{L}}=P_{\mathrm{L}}\nu_{\ee}$;
 \item we have applied the Fierz transformation described in 
 \sectionname~\ref{frz}.
\end{enumerate}
\paragraph{Matrix element for forward scattering}
Let us consider the previous hamiltonian at 
a given time, say, $t=0$. 
 Then, we consider its
matrix element  between the {\em same} initial and final state $|s\rangle$, namely,
 \begin{equation}
 \langle s| 
\mathbf{H}_{\nu_{\ee\ee}}^{\mathrm{cc}}(t=0) | s \rangle
\quad\text{with}\quad
| s \rangle=|\nu_{\ee},{\ee}_1,{\ee}_2,{\ee}_3\dots\rangle
 \end{equation}
The state $| s \rangle$  describes a single electron neutrino, localized around $\vec{x}=\vec{x}_{0}$, and 
a large number of electrons almost at rest, localized around the positions $\vec{x}=\vec{x}_{i}$ with $i=1,2,3,\dots $ 
(we have in mind a neutrino propagating in 
the matter e.g.\ of  the Sun or of the Earth). 
The only operator containing  the electron fields 
that has large matrix elements is the time component
of the vectorial part $\bar{\ee} \gamma_0 \ee={\ee}^\dagger \ee$,
 that is 
the electron density operator.\footnote{The other operators are 
${\ee}^\dagger \vec{\alpha}\, \hat{\ee}$, that is proportional to the current and thus is suppressed with $\beta=v/c$,  
${\ee}^\dagger \vec{\alpha}\gamma_5 {\ee}$ that is proportional to the polarization, and that we assume to be small, and finally 
${\ee}^\dagger\gamma_5 \ee$ that is even smaller.}   
 Thus, the relevant part of the matrix element is, 
\begin{equation}
 \langle s| 
\mathbf{H}_{\nu_{\ee\ee}}^{\mathrm{cc}}(t=0)
| s \rangle= \sqrt{2}\, G_{\mathrm{F}} \int \mathrm{d}^3 x\ 
 \langle s| \ 
\mathbf{n}_{\ee}(\vec{x})\ \mathbf{n}_{{\nu}_{\ee\mathrm{L}}}(\vec{x})
\ | s \rangle \label{cacabum}
 \end{equation}
 where we indicate 
the density operators of the electron and of the neutrino with 
$\mathbf{n}_{\ee}(\vec{x})={\ee}^\dagger \! (0,\vec{x}){\ee}(0,\vec{x}) $ and 
$\mathbf{n}_{{\nu}_{\ee\mathrm{L}}}(\vec{x})={{\nu}_{\ee\mathrm{L}}}^\dagger \! (0,\vec{x}) {\nu}_{\ee\mathrm{L}} (0,\vec{x}) $. 
We factorize the state $|s\rangle$ and 
consider separately the matrix elements of the two operators.   We have,
\begin{eqnarray}
	\langle {\nu}_{\ee\mathrm{L}} | \mathbf{n}_{{\nu}_{\ee\mathrm{L}}}(\vec{x})  | {\nu}_{\ee\mathrm{L}} \rangle &=& |\Psi_{{\nu}_{\ee\mathrm{L}}}(\vec{x},\vec{x}_{0})|^2\\
	\langle {\ee}_1, {\ee}_2,{\ee}_3\dots| \mathbf{n}_{\ee}(\vec{x})  | {\ee}_1,{\ee}_2, {\ee}_3\dots \rangle &=&  \sum_{i} |\Psi_e(\vec{x},\vec{x}_{i})|^2
\end{eqnarray}
where we have introduced the wave-functions of the neutrino and of the electrons, 
emphasizing the position of each particle. We supposed for simplicity that the electrons are distributed similarly, but around different positions.

In the case of interest, the individual electrons are distributed in small cells, e.g., atomic sizes of $\sigma_{\ee} \sim \SI{1e-8}{\centi\meter}$, whereas the 
neutrino is distributed in a much larger scale $\sigma_\nu$. Another important quantity is the {\em macroscopic}  electron density
that we denote by $n_{\ee}(\vec{x})$. We suppose that the scale  of variation of the macroscopic electron density, $\Delta r_e$  is much larger  than the size of the neutrino wave-function. For the Sun, e.g., this is $\Delta r_{\ee}\sim R_\odot/10$ and also in the Earth it is typically km size or more.
Summarizing, we suppose that,
\begin{equation}
\Delta r_{\ee}\gg \sigma_\nu \gg \sigma_{\ee}
\end{equation} 

In the above assumptions
we can evaluate the integral in \eqref{cacabum}. 
We divide the region where the neutrino wave-function is non-zero in many cubes 
$a=1,2,3\dots$. Each cube has size  $\Delta x$ such that $\sigma_\nu\gg \Delta x \gg \sigma_{\ee}$ and 
 volume $\Delta x^3$. In each cube there are a lot of electrons, while 
the neutrino density is almost constant with value $ |\Psi_{{\nu}_{\ee\mathrm{L}}}(\vec{x}_a,\vec{x}_{0})|^2$.
If we integrate the electron density 
in the cube with label $a$, we have $\int_a \mathrm{d}^3 x\sum_i |\Psi_{\ee}(\vec{x},\vec{x}_{i})|^2=N_{\ee}(\vec{x}_a)\gg 1$,  namely the number of electrons inside the cube. This is almost the same for all $a$, because we suppose that $\Delta r_{\ee}\gg \Delta x$.
We can then approximate: $N_{\ee}(\vec{x}_a)\approx n_{\ee}(\vec{x}_{0}) \Delta x^3$,
where we indicate the position $\vec{x}_{0}$ of the neutrino since we are interested  in the value of the electronic density where the neutrino is.  
Thus,  the integral is,
\begin{equation}
	I(\vec{x}_0) = \int \mathrm{d}^3 x \sum_{i} |\Psi_{\ee}(\vec{x},\vec{x}_{i})|^2\times|\Psi_{{\nu}_{\ee\mathrm{L}}}(\vec{x},\vec{x}_{0})|^2 
\end{equation}
which can be solved as,
\begin{equation}
\begin{aligned}
	I(\vec{x}_0) &\approx \sum_a   N_{\ee}(\vec{x}_a)
|\Psi_{{\nu}_{\ee\mathrm{L}}}(\vec{x}_a,\vec{x}_{0})|^2\approx\\
&\approx n_{\ee}(\vec{x}_{0}) 
  \sum_a   
\Delta x^3\ |\Psi_{{\nu}_{\ee\mathrm{L}}}(\vec{x}_a,\vec{x}_{0})|^2
\approx
  \\ &\approx  n_{\ee}(\vec{x}_{0})\times \int \mathrm{d}^3 x\  |\Psi_{{\nu}_{\ee\mathrm{L}}}(\vec{x},\vec{x}_{0})|^2
  =n_{\ee}(\vec{x}_{0})
\end{aligned}
\end{equation}
Therefore, the matrix element of interest, for a single neutrino 
located around $\vec{x}=\vec{x}_0$, is just,
\begin{equation}
\langle \mathbf{H}_{\nu_{\ee\ee}}^{\mathrm{cc}} \rangle 
=\sqrt{2} G_{\mathrm{F}}\ n_{\ee}(\vec{x}_{0})
 \end{equation}
 Following the above derivation, it is easy to see that an electron antineutrino has the same matrix element but with the opposite sign, due to the fact that 
 the operator $\mathbf{n}_{\nu_{\ee\mathrm{L}}}$ counts the net number of electron neutrinos: therefore, an antineutrino counts $-1$.
 
\subsection{Hamiltonians of propagation including the matter term}

 Thus, we should add to the $3\times 3$ vacuum hamiltonians, that describe free neutrino/antineutrino propagation, also the matter hamitonian
 that acts on flavor states, 
 \begin{equation}\label{mss}
H_{\mathrm{MSW}} =\pm  \sqrt{2}\, G_{\mathrm{F}} n_{\ee}(\vec{x})  \left(\begin{array}{ccc}
1 & 0 & 0 \\
0 & 0 & 0 \\
0 & 0 & 0 
\end{array}
 \right)
 \quad
 \text{where}
 \quad
 \begin{cases}
 	+ & \text{for $\nu$}\\
 	- & \text{for $\bar{\nu}$}
 \end{cases}
 \end{equation}
 where $\vec{x}$ is the position of the neutrino and 
 $n_{\ee}(\vec{x}) $ the number density of electrons. 
 To be sure, note that this term has the right dimension, namely ``energy'', since 
 the Fermi constant has $1/\text{energy}^2$ and the electron density is $1/\text{volume}=\text{energy}^3$
 in natural units. 
 Note that we consider a {\em linear} contribution in the Fermi constant (whereas cross sections or decay widths depend upon $G_{\mathrm{F}}^2$). This fact shows that we are considering amplitudes and/or wave-functions rather than probabilities and witnesses the quantum nature of the phenomenon.  A last remark is in order: 
 this new term can be regarded as a contribution to the 
 refraction of the electronic neutrino wave.
 
Summarizing, the flavor of 
ultrarelativistic neutrinos  of given energy $E$ that propagate 
in a medium with electronic density  $n_{\ee}(x)$,
changes according to  
neutrino and antineutrino hamiltonians that have one part due to vacuum oscillations and another one due to the Wolfenstein (or matter) term.  For neutrinos, this is given by, 
\begin{equation}\label{por}
H_\nu=U^*\,\mathrm{diag}(k)\,
U^t + V(x)\,
\mathrm{diag}(1,0,0)
\end{equation}
where the symbols $U^*$ and $U^t$ indicate the complex conjugate and the transpose of the mixing matrix, respectively, and where,
\begin{eqnarray}
	k_i&=&m_i^2/(2 E)\\
	V(x)&\equiv& \sqrt{2} G_{\mathrm{F}} n_{\ee}(x)
\end{eqnarray}
For antineutrinos we need  to replace,
\begin{equation}
U\to U^* \quad\text{and}\quad V\to -V
\end{equation}
We have adopted a shorthand $k$ for the quantity that describes the effect that are purely due to neutrino masses (i.e., of vacuum oscillations). 

Before discussing some specific example and some solutions,
we note that these hamiltonians contain the energy $E=p c$ (and/or the momentum) of the neutrino and its position $x$, too. This description might seem to be in contradiction with the Heisenberg indetermination principle, i.e., with some basic facts of  quantum mechanics. However, as we have discussed above, neither of these observable is exactly measured. The idea is that, in actual circumstances, neutrinos travel in wave packets, which are supposed to be non dispersive. 

\subsection{Remarks\label{sez:rmk}}

\subsubsection{Numerical value}
The numerical value of the matter term and of the vacuum term, given in the same units of inverse meter, is,
\begin{eqnarray}
	V&=&\sqrt{2}\, G_{\mathrm{F}} n_{\ee}= \frac{3.868\times 10^{-7}}{\mathrm{m}}  \times  \frac{n_{\ee}}{\mathrm{mol}/\mathrm{cm}^3}\\
k&=&\frac{\Delta m^2}{2 E}= \frac{2.533}{\mathrm{m}}
\times \frac{\Delta m^2}{\mathrm{eV}^2} \times \frac{\mathrm{MeV}}{E}
\end{eqnarray}

Let us remark that, when the matter term is small in comparison to the other one, we reduce to the case of vacuum oscillations. In order to clarify when this happens, let us examine the ratio of these two terms. We find,
\begin{equation}
\frac{V}{k}\equiv 
\frac{\sqrt{2}\; G_{\mathrm{F}} n_{\ee} }{\Delta m^2/(2 E)}\approx 
\left( \frac{\rho\ Y_{\ee}}{100\,\mathrm{mol}/\mathrm{cm}^3} \right)
\left( \frac{8\times 10^{-5}\,\mathrm{eV}^2}{\Delta m^2} \right)
\left( \frac{E}{5\,\mathrm{MeV}} \right)
\label{100}
\end{equation}
where $\rho$ is the mass density and $Y_{\ee}$ is the fraction of electrons and/or of protons (since ordinary matter is neutral). The 
values of the parameters are those needed to interpret solar neutrino data: 1.~the central density of the Sun is expected to be about $150\,\mathrm{g}/\mathrm{cm}^3$, with some fraction of Helium;
2.~the numerical value of $\Delta m^2$ 
is close to the one given in \tablename~\ref{tab:lisi16}; 3.~the energy thresholds of Super-Kamiokande and SNO are close to \SI{5}{\mega\electronvolt}.  The  comparison in \eqref{100}
shows that at low neutrino energies $E\ll \SI{5}{\mega\electronvolt}$, the matter (Wolfenstein) term has a negligible effect whereas at high energy it is relevant.

\subsubsection{A comment on CP and T violation}
Let us compare the dependence of certain 
probabilities of conversion 
in  various cases,
\begin{equation}
\begin{aligned}
P_{\pnum\to \pnue}= & F(U, V;E_\nu,L)\\
P_{\apnum\to \apnue}= & F(U^*, -V;E_\nu,L)\\
P_{\pnue\to \pnum}= & F(U^*, V;E_\nu,L)\\
P_{\apnue\to \apnum}= & F(U, -V;E_\nu,L)
\end{aligned}
\end{equation}
where $F$ is a certain function. When we exchange the initial and final state, we consider T-reversal which implies $U\neq U^*$. 
Furthermore, we see that there are also 
C and CP violating effects when we exchange particles and antiparticles, that should be  attributed to the exchange $V\to -V$.  
In fact, 
ordinary matter contains electrons, and therefore, it induces C and CP violating effects  in the neutrino propagation that are of  `environmental nature', 
on top of the CP (or T) violating effects of fundamental nature, due to the fact that the leptonic mixing matrix is complex. 

\begin{figure}
\begin{center}
\includegraphics[width=9cm,bb=0 0 800 542]{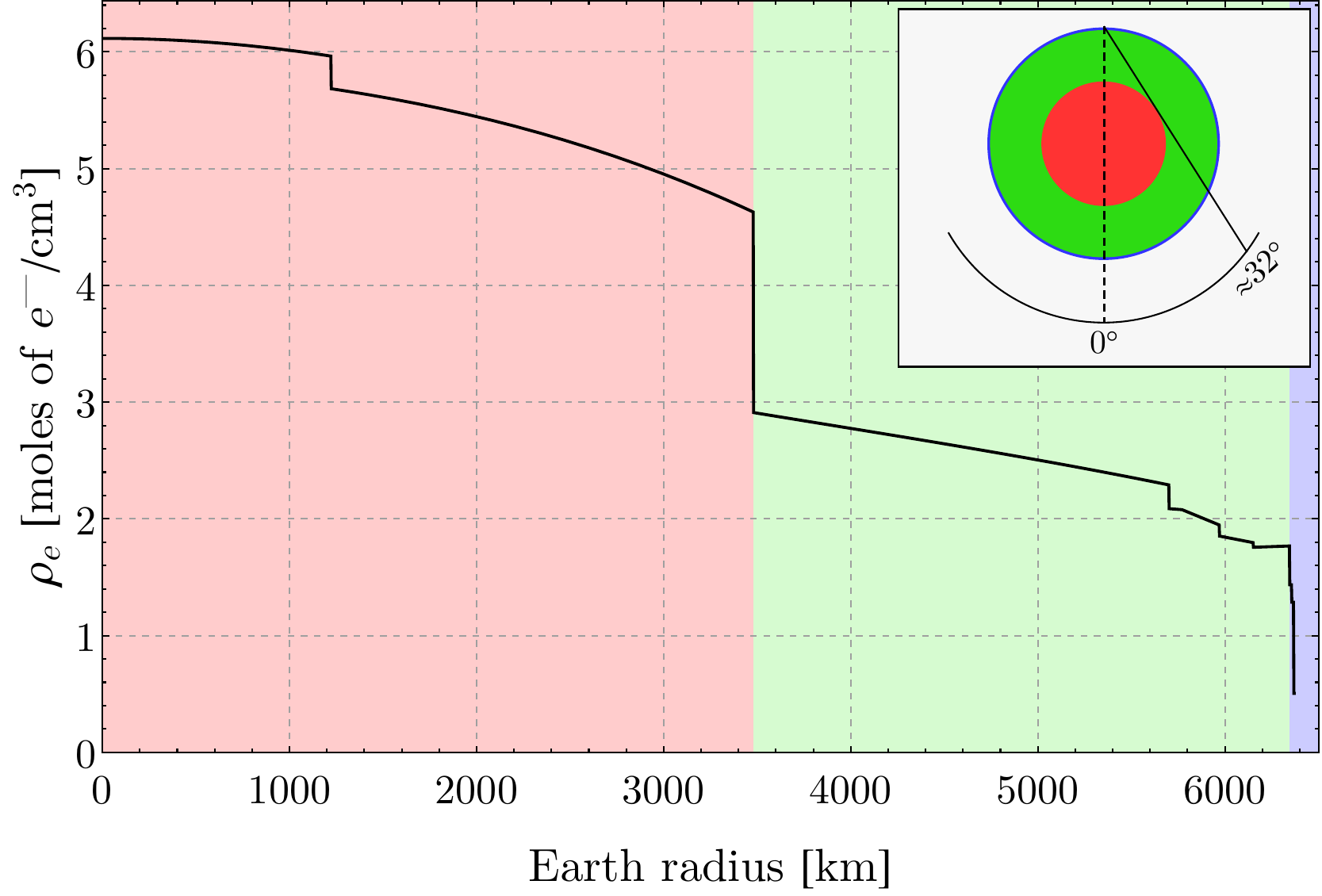}
\caption{Density profile of the Earth according to the
PREM model \protect\cite{prem}.
Different colors correspond to different Earth layers: 1.~inner ad outer core in red; 2.~mantle in green; and 3.~crust in blue.
Visualizing the shell structure in scale, we notice that a neutrino
can cross the core only if its angle is smaller than
$\approx 32^\circ$.
\label{figprem}}
\end{center}
\end{figure}


\subsubsection{Amplitude of three flavor neutrino oscillations}

In order to obtain a general expression for the oscillation amplitude in matter in the three flavor case, we can consider the mixing matrix $U= 
R_{23}\ \Delta \ R_{13}\ \Delta^*\ R_{12}$, given in \eqref{mixingmatrix}, such that the propagation hamiltonian in \eqref{por} is,
\begin{equation}\label{acci}
H_\nu=R_{23} \ \Delta \ \tilde{H} \ \Delta^* \ R_{23}^{\mathrm{t}}
\end{equation}
with,
\begin{equation}\label{Htilde}
 	\tilde{H}=R_{13}\ R_{12}\ \mathrm{diag}(k) \ R_{12}^{\mathrm{t}}\ R_{13}^{\mathrm{t}} + V(x)\mathrm{diag(1,0,0)}
 \end{equation}
 where $\tilde{H}$ is independent\footnote{The definition of $\tilde{H}$ in \eqref{Htilde} obtains from $\left[\Delta,\ R_{12}\right]=0$ and $\left[V,\ R_{23}\right]=\left[V,\ \Delta \right]=0$.} on $\theta_{23}$ and $\delta$.
The formal expression of the amplitude of the three flavor neutrino oscillations when
 the matter effect is included, 
 can be hence written as\footnote{
If we divide the interval $[0,t]$ in $n$ equal intervals of size $\mathrm{d}t=t/n$, and indicate with
$t_j=j \mathrm{d}t$ the discrete values of the time, where $j$ is an integer with $n\ge j>0$, we can define the time-ordered product as,
$$
\mbox{Texp}\left[ -i \int_0^t \mathrm{d}\tau H(\tau) \right] = \lim_{n \to \infty }
e^{-i\, \mathrm{d}t\, H(t_n) }
e^{-i\, \mathrm{d}t\, H(t_{n-1}) } \dots
e^{-i\, \mathrm{d}t\, H(t_1) }
$$}, 
 
\begin{equation}
		\mathcal{U}=\mathrm{Texp}\left[-i\int \mathrm{d} t\; {H}_\nu(t)\right] \equiv R_{23} \Delta \left(
									\begin{array}{ccc}
									u_{11} & u_{12} & u_{13} \\
									u_{21} & u_{22} & u_{23} \\
									u_{31} & u_{32} & u_{33} 
									\end{array}
											\right) \Delta^* R_{23}^{\mathrm{t}} \label{cagon}
\end{equation}
where $u_{ij}$ depend upon $\Delta m^2_{23}$, $\theta_{13}$,
 $\Delta m^2_{12}$, $\theta_{12}$ and on the type of mass hierarchy (normal/inverted), that can be easily fixed changing a discrete parameter $h=\pm 1$.
 When we solve the evolution equations with a certain assumption on the matter density, we calculate the 
complex numbers $u_{ij}$ and therefore 
the amplitudes and the probabilities of oscillations.
For example, at the web address,
\begin{quote}
{\small \url{http://pcbat1.mi.infn.it/~battist/cgi-bin/oscil/index.r}}
\end{quote}
the interface of the `Neutrino Oscillations  Simulator' allows the user to plot and to download the probabilities 
for neutrinos or antineutrinos that propagate from two points on the surface of the  Earth. The solutions are 
obtained numerically, as described in Ref.~\cite{gb}, see in particular Eqs.~(1), (3), (5), (6)  there and using the Preliminary reference Earth model (PREM) of the electronic Earth density illustrated in \figurename~\ref{figprem}. 

\section{Analytical solutions of the equations of propagation}

For a generic profile of electronic density, it is always possible to resort to numerical methods to solve the equations of propagations. However, the solutions have an oscillatory character and  this makes the problem difficult; we need to use a sufficiently small step, that can be eventually time consuming. 

However, there are various cases when the propagation equations have exact or approximate analytical solutions, that are useful to describe the physics. We will discuss two of them  in the following.  
The first is simply the case when we have a constant density; this can be applied in sequence to layers of constant density, finding eventually the overall amplitude. 
The second case instead is when the density varies slowly, which is called the adiabatic case; this admits a simple solution, which moreover is physically important.

\subsection{Constant matter density  and `resonance'}\label{constdensity}
The case when we have two flavors 
($2\times 2$ matrices) and matter with constant density is particularly instructive. We have,
\begin{equation}
U(\theta)=\left(
\begin{array}{cc}
c & s \\
-s & c
\end{array}
\right)\quad\text{and}\quad
k\equiv{k_2-k_1}=\frac{\Delta m^2}{2 E}
\end{equation}
where we can assume 
$\Delta m^2>0$ and $0<\theta<\pi/2$.

Note that, for the description of oscillations,  it is harmless  
if the hamiltonian of propagation is modified by adding a real constant: in fact, this constant yields the same phase for all flavor states, that is irrelevant for the flavor transformations. Thus, let us choose the constant in such a manner that our hamiltonian is traceless, i.e., let us subtract the term $-\mathrm{tr}(H)/2\cdot\mathbb{1}$, getting\footnote{This procedure relies on the same principle that allows us to change the overall factor  
into the vacuum hamiltonian --- see e.g., \eqref{equib}.},
\begin{equation}\label{constH}
H= \frac{k}{2}\, U\, \mathrm{diag}(-1,1)\ U^{\mathrm{t}} + \frac{V}{2}\, \mathrm{diag}(1,-1)=
\frac{1}{2} 
\left(
\begin{array}{cc}
V-k c_2 & k s_2 \\
k s_2 & k c_2-V
\end{array}
\right)
\end{equation}
where we use 
the obvious shorthands $c_2=\cos 2\theta$ and $s_2=\sin 2\theta$.

At this point, we can simply {\em redefine} the angles and the parameters $k$, in such a manner that this hamiltonian coincides formally with a vacuum hamiltonian. This is obtained by setting, 
\begin{equation}\label{eqmat}
\begin{cases} 
k c_2-V\equiv k_m c_{2m} \\
k s_2\equiv k_m s_{2m} 
\end{cases}\quad
\Longleftrightarrow\quad
\begin{cases}\displaystyle
	t_{2m}=\frac{s_2}{c_2-V/k }\\
	k_m =  k \sqrt{   s_2^2 + (c_2 -V/k)^2 }
\end{cases}
\end{equation}
In this manner, we can just reinterpret 
the solutions, that have been  
already obtained for the case of vacuum oscillations. 
We need only to examine the so called 
mixing angle in matter  $\theta_m$ and the parameter $k_m\equiv 
 \Delta m^2_m/(2 E)$.
 
  Let us discuss  the most interesting feature. 
  Suppose that  $V>0$ (as it is usually the case) and 
  $k c_2$  has  the same sign (namely, suppose $\theta<\pi/4$). 
By varying the energy of the neutrino $E$,  we vary $k$ from $0$ to $\infty$. Thus, 
it is always possible to realize the case  
$c_{2m}=0$, namely, the case when $\theta_{m}=\pi/4$ (maximal mixing). This case is called {\em resonance} and the corresponding energy is called the resonance energy. The resonance point corresponds to the minimum splitting between the energy levels of the mass eigenstates. The possibile transition between the two energy levels (let us call them $\lambda_{\pm}$) is usually defined {\em level crossing}. From \eqref{eqmat} is evident that the minimum distance between $\lambda_{\pm}=\pm k_m/2$ occurs for $V=kc_2$.
Note that:
1.~this is possible 
 whichever the value of the mixing angle in vacuum;
2.~at the resonance energy, also $k_m$ and thus $\Delta m^2_m$ 
 get the minimum value.
 
Summarizing, we will have three limiting situations for neutrinos,
\begin{equation}
\begin{cases}
	V\ll k  &\Rightarrow \theta_m\approx  \theta \\
V= k c_2 &\Rightarrow \theta_m\approx  45^\circ \\
V\gg k  &\Rightarrow \theta_m\approx  90^\circ \\
\end{cases}
\label{icasi}
\end{equation}
while for the antineutrinos instead, the sign of $V$ will be opposite and the mixing angle in matter will be smaller than the one in vacuum. 
We leave to the interested Reader a complete  study of   $\Delta m^2_m$, $\theta_m$  and $P_{\ee \mu}(k_m, \theta_m)$, 
as the parameter $V$ is varied.

\subsection{Adiabatic propagation / MSW effect\label{solr}}
Generally, as the neutrino proceeds along its path, the electron density varies. The case when the variation is slow (in a sense to be made precise later) is particularly important and in this case it is possible to 
obtain a solution of the equations of propagation. This is called the adiabatic solution. Let us proceed step by step:
\begin{enumerate}[itemsep=-0.7ex,partopsep=1ex,parsep=1ex]
	\item First of all we rewrite the hamiltonian, by introducing the local eigenvalues. We rewrite,
	\begin{equation}
i \partial_t \nu = [\, U_m\, \mathrm{diag}(k_m)\, U_m^\dagger \, ] \ \nu
\end{equation}
where $\nu$ is the vector of the flavor states, $\nu=(\Pnue,\Pnum, \Pnut)^{\mathrm{t}}$, 
while $U_m$ is the mixing matrix that depends upon the position and/or the time along the path.
\item Now we introduce the local mass eigenstates $n=(n_1,n_2,n_3)^{\mathrm{t}}$,
\begin{equation}
\nu\equiv U_m n\;\Rightarrow \;i \partial_t n = [ \mathrm{diag}(k_m) - i
(U_m^\dagger \partial_t U_m)] n
\end{equation}
If the second term was not there, it would be easy to solve the problem.
\item In order to understand the physics, we focus from here on the
$2\times 2$ case, when all matrices are real  
and say $\nu=(\Pnue,\Pnum)^{\mathrm{t}}$ and  
$n=(n_1,n_2)^{\mathrm{t}}$. After a few direct calculations, we find that  the propagation equation is\footnote{We have added the term  $-\mathrm{tr}[\mathrm{diag}(k_m)]/2\cdot\mathbb{1}$, that modifies only the overall phase, irrelevant for neutrino oscillations.}, 
\begin{equation}\label{mattevol}
i \partial_t 
\left(
\begin{array}{c}
n_1 \\
n_2
\end{array}
\right)
= 
\left(
\begin{array}{cc}
-k_m/2 &\quad -i \partial_t \theta_m \\
i \partial_t\theta_m &\quad k_m/2
\end{array}
\right)
\left(
\begin{array}{c}
n_1 \\
n_2
\end{array}
\right)
\end{equation}
where the effect of the change in the electron density is contained 
in  $\partial_t \theta_m$.
\item Now we introduce the definition,
\begin{equation}\label{adpar}
\quad\text{Adiabaticity parameter:}\quad \gamma=\left| \frac{2 \dot \theta_m}{k_m}\right|
\end{equation}
and we call `adiabatic condition' the one when 
$\gamma$ is always small, so that the new term is negligible and we can 
easily read the solution. In order to have a more precise idea, we note that the case when $k_m$ is minimum, i.e., at the `resonance', 
the adiabaticity parameter is,
\begin{equation}
V(x_{\mathrm{res}})= k c_2 \;\Rightarrow \;
\gamma(x_{\mathrm{res}})= 
\left| \frac{\dot{V}}{ V} \ \frac{c_2}{k\; s_2^2}\right|
\label{ader}
\end{equation}
The case when the condition of adiabaticity is violated significantly is a bit more complex. However, it can be still treated analytically\footnote{This was obtained in Ref.~\cite{parke} and it is 
reviewed, e.g., in Ref.~\cite{r-k}. 
We will not need its explicity expression but 
will discuss it very briefly later on in \sectionname~\ref{sccsn}.},
and the effects are described by a negative exponential in the parameter  $1/\gamma$.
\item Suppose that $\gamma$ is small. The local mass eigenstates
$n_1$ and $n_2$
will not change and evolve only getting a phase,
\begin{equation}
n_1(t) = e^{i \xi}\ n_1(0)\quad
n_2(t) = e^{-i \xi}\ n_2(0)
\quad\text{where}\quad \xi=\int \frac{k_m}{2}\mathrm{d} t 
\label{nada}
\end{equation}
This is the formal statement that describes the case of `adiabatic propagation' in our two flavor system.
At this point it is easy to find the propagation amplitude, e.g.,
$\mathcal{U}_{\ee\to \ee}= c c_m e^{i \xi} +s s_m e^{-i
  \xi}$ (In order to derive this equation, we begin from 
$\nu=U_m n$, we find that
$\Pnue=c_m n_1 + s_m n_2$ and then we calculate 
$\Pnue(t)$ using $n_1(t)$ ed $n_2(t)$, thanks to the linearity of the 
equation; finally, we use the definition, 
$\mathcal{U}_{\ee\to \ee}=\langle \Pnue, \Pnue(t)
\rangle$, or in other words we take the scalar product).
\item Since the phases $\xi$ are large, we can average to zero the 
interference terms in the formulae of the probability of electron survival, getting,
\begin{equation}
P_{\ee\ee}=c^2 c_m^2 + s^2 s_m^2 
\label{adeb}
\end{equation}
this is the formula that is usually employed and in which we are interested. 
We recall that the case when the interference terms drop out is a typical classical average. In the three flavor case the obvious generalization is  $P_{\ee\ee}=|U_{\ee i}^2| |U^m_{\ee i}|^2$ and will be discussed in the last part of this section. 
\item For practical purposes, it is useful to note that when $V\gg k$, the electron neutrino is $\Pnue\approx n_2$, since, as just explained, it propagates in matter with the combination $\Pnue(t)=c_mn_1(t)+s_mn_2(t)$ and for $V \gg k$, $\theta_m\approx 90^\circ$, as discussed in \sectionname~ \ref{constdensity}. Since the central density in the Sun is about \SI{150}{\gram\per\cubic\centi\meter}, this is a good approximation for the propagation of solar neutrinos above about \SI{5}{\mega\electronvolt}. 
\end{enumerate}

Let us note that the above description is strictly analogous to the situation of a two-level hamiltonian that depends upon one parameter. This analogy can help to understand the propagation in the adiabatic regime, as discussed below.
 
 \subsubsection{Time dependent hamiltonian in adiabatic approximation} 
Let us consider a time dependent hamiltonian and introduce the eigenstates and eigenvalues
\begin{equation}
H(t) |n(t)\rangle=E_n(t)  |n(t)\rangle
\end{equation}
and suppose the hamiltonian is not degenerate. 

Let us derive w.r.t.\ time,
\begin{equation}
\dot{H} |n\rangle + H \rndot
=\dot{E_n} |n \rangle + E_n \rndot
\end{equation}
then we multiply by $\langle m|$ with $n\neq m$, 
\begin{equation}
\begin{aligned}
\mathrm{l.h.s.} &=
\langle m|\dot{H} |n\rangle + \langle m| H \rndot = \langle m|\dot{H} |n\rangle + E_m \langle m \rndot \\
\mathrm{r.h.s.} &=\dot{E_n} \langle m| n \rangle + E_n \langle m \rndot=E_n \langle m \rndot
\end{aligned}
\end{equation}
so we find,
\begin{equation}
\langle m\rndot=\frac{\langle m|\dot{H} |n\rangle }{E_n-E_m}\quad \text{if}\quad n\neq m
\end{equation}
Thus the eigenstate acquires overlap with other eigenstates if the derivative of $H$ is large in comparison to the energy differences.

Instead in the case $n=m$ we have,
\begin{equation}
\langle n| n\rangle=1\:\Rightarrow\: \lndot n\rangle+ \langle n \rndot =0
\:\Rightarrow\: \langle n \rndot =- \langle n \rndot^*
\end{equation}
in other words, this number is imaginary. We define,
\begin{equation}
\eta_n =i \langle n \rndot \in \mathbb{R}
\end{equation}

Now we can write the generic states as,
\begin{equation}
|\psi\rangle=\sum_{n} e^{-i S_n}c_n |n\rangle\quad\text{where}\quad
S_n(t)=\int_0^t d\tau\ E_n(\tau)
\end{equation}
Its time evolution is dictated by the Schr\"{o}dinger equation
$i\dot{|\psi\rangle}=H |\psi\rangle$, that implies,
\begin{equation}
\begin{aligned}
&\sum_{n}  \dot{S_n} e^{-i S_n}c_n |n\rangle     + 
  i\sum_{n}  \left( 
e^{-i S_n}\dot{c}_n |n\rangle   + 
e^{-i S_n}c_n \rndot 
 \right)\\
 =&\sum_{n} e^{-i S_n}c_n  H\, |n\rangle 
\end{aligned}
\end{equation}
the first and last term cancel out. Multiplying by $\langle m| $, we have
\begin{equation}
\dot{c}_m= - \sum_n e^{i (S_m- S_n) }c_n \langle m \rndot = -c_m \langle m\rxdot{m}
-    \sum_{n\neq m} e^{i (S_m- S_n) }c_n \langle m \rndot 
\end{equation}
Using the above result,
\begin{equation}
\dot{c}_m= -c_m \langle m\rxdot{m}  -    \sum_{n\neq m} c_n \ \frac{\langle m|\dot{H} |n\rangle }{E_n-E_m} e^{i (S_m- S_n) } 
\end{equation}

If the second term is negligible, namely in the adiabatic case, we have,
\begin{equation}
\displaystyle 
c_m(t)=c_m(0) \exp\left[  i \int_0^t \mathrm{d}\tau \ \eta_m(\tau) \right]
\end{equation}
namely, $|c_m|^2$ does not change with time.

From the time evolution of the mass eigenstates in \eqref{mattevol}, we have,
\begin{equation}\label{n2ev}
	i\partial_x n_2=\frac{k_m}{2}\left(i\frac{2\dot{\theta}_m}{k_m}n_1+n_2 \right)=\frac{k_m}{2}\left(i \gamma n_1+n_2 \right)
\end{equation}
where $\gamma$ is the adiabatic parameter in \eqref{adpar}. Since in the adiabatic approximation the limit $\gamma \ll 1$ holds, the mass eigenstate time evolution is,
\begin{equation}
\begin{cases}
	i\;\partial_x n_1 = -k_m\, n_1/2 \\
	i\;\partial_x n_2 = k_m \, n_2/2
\end{cases}
\end{equation}
i.e.\ the time evolutions (time and spatial evolutions are equivalent in the relativistic regime) of the two eigenstates are independent one from another. As a consequence, in the adiabatic approximation, if $\Pnue \approx n_2$ in the production point (e.g. for solar neutrinos), $\Pnue$ remains in that eigenstate during all the adiabatic evolution.

To summarize, the case of adiabatic propagation applies if the probability of transition between the two eigenstates is always negligible, even when the two levels get as close as possible without crossing each other  (a sub-case called ``avoided crossing'' in quantum mechanics).

\subsubsection{Effect of the Earth on neutrino propagation}

As we have just seen,  in the case of solar $\Pnue$ we expect to receive on Earth just the mass eigenstate $\nu_2$ 
above a certain energy
(adiabatic conversion). However, it is possible that neutrinos
reach the detector passing through the Earth, i.e.\ in night time.
  Therefore, it is interesting to calculate the probability of conversions  from $\nu_2$ to $\Pnue$, in two flavors approximation.
 This has a very simple and useful expression in the case of constant matter density, that we discuss here. 
 
 In the two flavor basis, for the time independent hamiltonian in \eqref{constH}, the evolutor that describes the effect of the propagation in the Earth is,
 \begin{equation}\label{prop}
 \begin{aligned}
 \mathcal{U}&=U_m\,\mathrm{diag}[e^{i\varphi_m}, e^{-i\varphi_m} ]\ U_m^{\mathrm{t}}\\&=
 \left(
 \begin{array}{cc}
 c_m + i s_m \cos 2\theta_m & -i s_m \sin 2\theta_m \\
 -i s_m \sin 2\theta_m & c_m - i s_m \cos 2\theta_m 
 \end{array}
 \right)
 \end{aligned}
 \end{equation}
 where $\varphi_m=\Delta m^2_m L/(4 E)$, $c_m=\cos\varphi_m$, $s_m=\sin\varphi_m$. Consider the case of solar neutrinos; 
the initial state is $\nu_2$, that in flavor space is represented by $(\sin\theta,\cos\theta)$ and the final state is 
$\Pnue$ that in flavor space is represented by $(1,0)$. The relevant matrix element of the evolutor is,
\begin{equation}
\begin{aligned}
 \mathcal{U}_{2\to \ee}&=(c_m+i s_m \cos 2\theta_m ) \sin\theta- i s_m \sin 2\theta_m \cos\theta\\&=
 c_m \sin\theta-is_m \sin(2 \theta_m-\theta)
 \end{aligned}
 \end{equation}
 The last factor can be rewritten using the formulae for the 
 mixing angles in matter,
 \begin{equation}
 \sin(2 \theta_m-\theta)=\sin\theta\ \frac{1+\varepsilon}{\sqrt{(1+\varepsilon)^2-4 \varepsilon \cos^2\theta }}
 \end{equation}
  where as usual $\varepsilon=V/k=2 E V /\Delta m^2$. 
With a bit of algebra, we obtain the desired result, 
\begin{equation}\label{eq:184}
P_{\nu_2\to \pnue}=
| \mathcal{U}_{2\to \ee}|^2= 
\sin^2\theta\left[1+ {4 \varepsilon\ \cos^2\theta} \ 
{\sin^2\left(\frac{\Delta m^2 L}{4 E}  \xi\right)}\; /\; {\xi^2}    \right]
\end{equation}
with,
\begin{equation}
{\xi}=\sqrt{ (1+\varepsilon )^2 - 4 \varepsilon \cos^2\theta}
\end{equation}
where of course we have in mind the so-called ``solar neutrino mass splitting'' (the lowest one) and $\theta$ is in good approximation just
$\theta_{12}\approx 33^\circ$.

\section{Applications and examples}
\subsection{High energy atmospheric neutrinos}

In this section, taken from Ref.~\cite{carol}, we 
examine the oscillations of high energy 
atmospheric neutrinos. This is interesting, since 
one could discriminate the neutrino mass hierarchy by 
observing the details of the matter effect.

We begin by using some simplifying assumptions, 
postponing derivations and refinements. 
Let us consider  oscillations with a single scale, 
and let us consider oscillations in constant matter density.
In fact, the second hypothesis  is rather inaccurate in the conditions in which we are interested for the experiments, and a more accurate evaluation of the oscillation probabilities  
for the actual calculations is required. However, a qualitative discussion based on simple-minded analytical results is sufficient to  correctly understand the relevant features. 

\begin{figure}[t]
\centering
\subfigure[Normal hierarchy.]
{\includegraphics[width=0.48\textwidth,
bb=0 0 800 539]{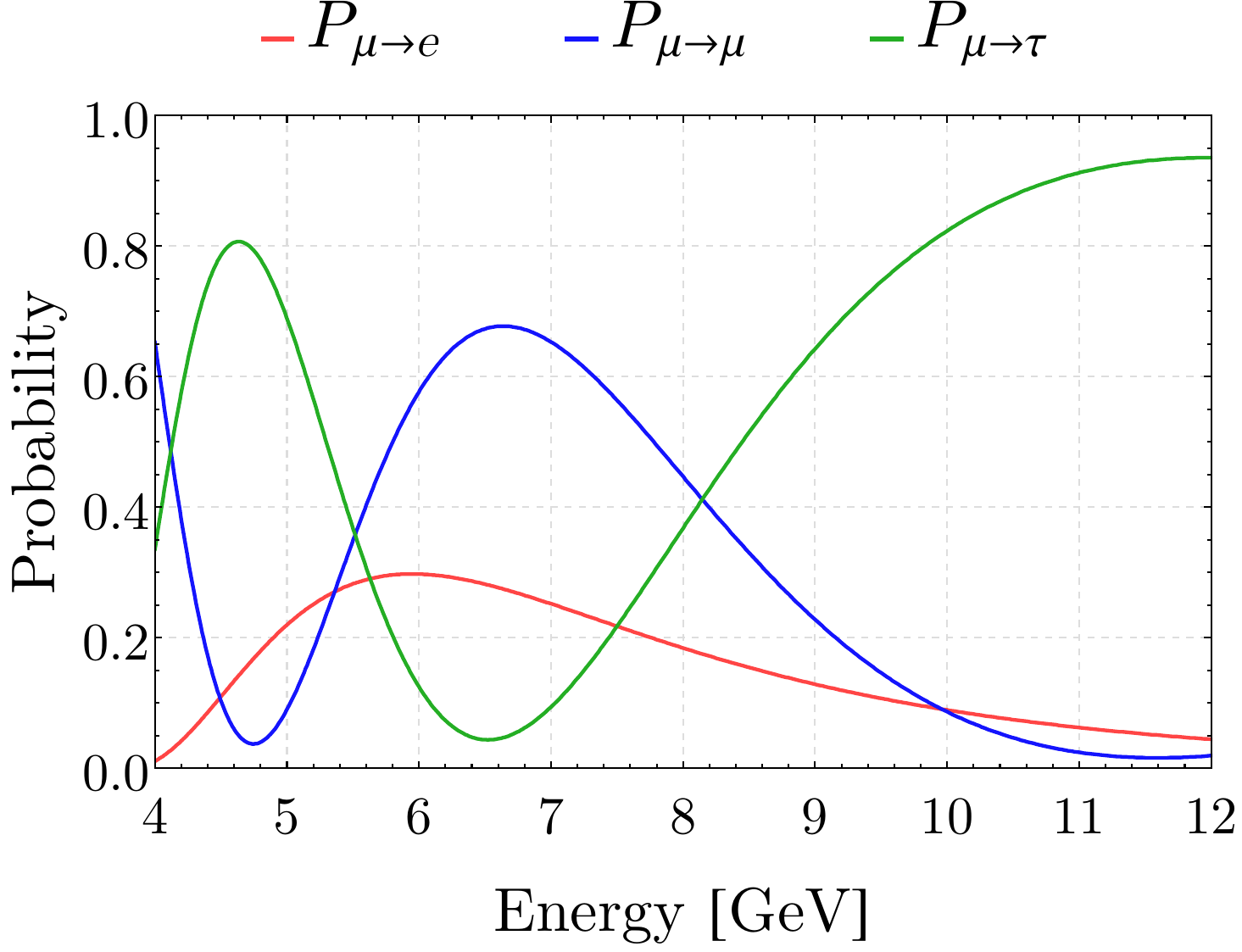}
\label{fig:7kNH}}
\subfigure[Inverted hierarchy.]{
\includegraphics[width=0.48\textwidth,
bb=0 0 800 539]{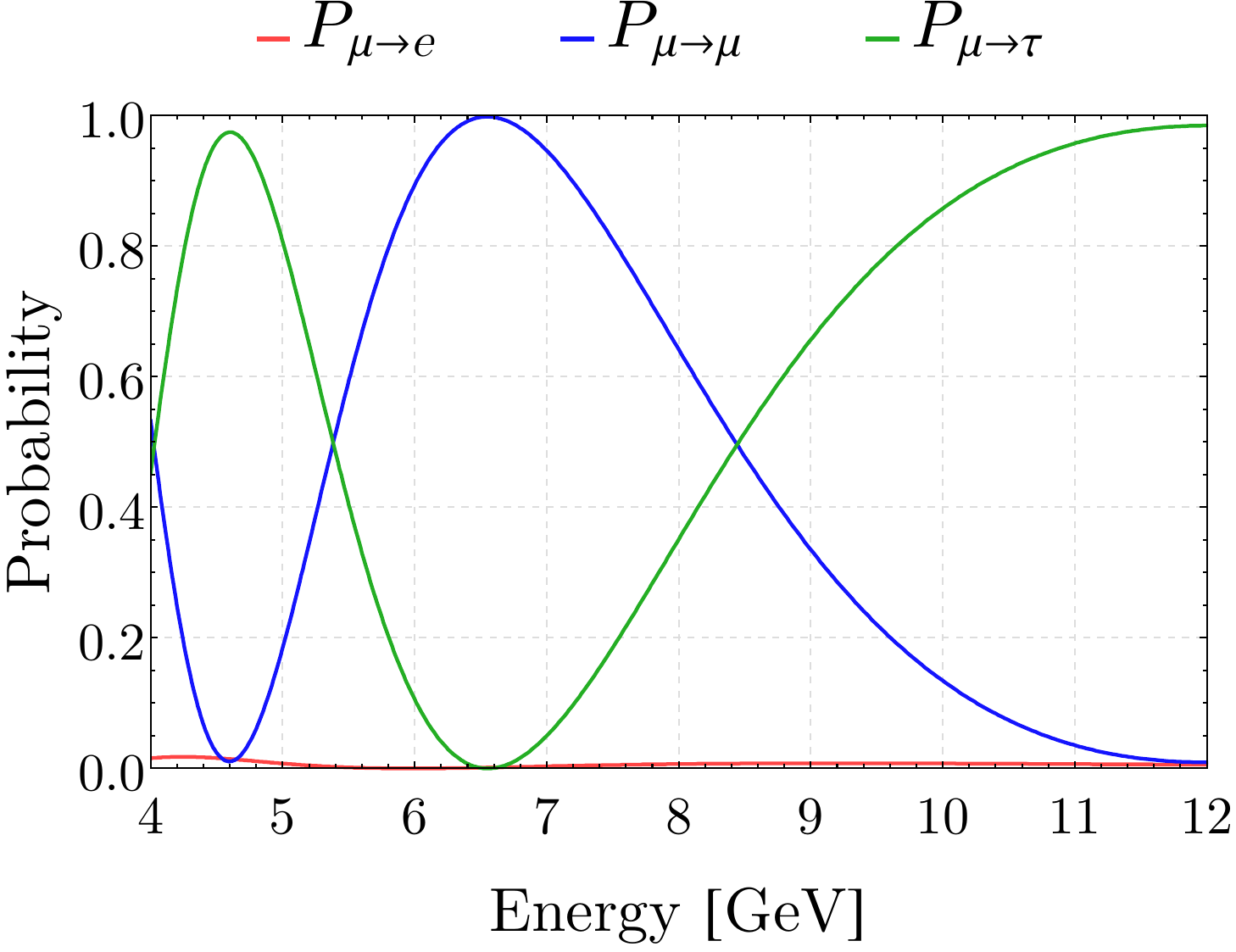}
\label{fig:7kIH}}
\caption{Oscillation probabilities in the Earth for a path of  
$L=7000 \,\mathrm{km}$: the blue line is 
$P_{\pnum\to\pnum}$; the red line  $P_{\pnum\to 
\pnue}$; the green one $P_{\pnum\to\pnut}$. The two panels
refer to normal and inverted hierarchy respectively.
The increase of $P_{\pnum\to \pnue}$ for normal hierarchy, 
due to matter effect, causes the decrease of
$P_{\pnum\to\pnum}$.}
\label{7k}
\end{figure}

Under the above assumptions,
the probability that a muon neutrino converts into an electron neutrino  
is simply,  
\begin{equation}
P_{\mu \mathrm{e}}=\sin^2\theta_{23}\;\sin^22\widetilde{\theta_{13}}\;\sin^2\widetilde{\varphi} 
\quad\text{with}\quad\widetilde{\varphi} =\frac{\widetilde{\Delta m^2} L}{4 E}  \label{aq}
\end{equation}
Most of the results in which we are interested
follow from the above simple formula.  
In \eqref{aq}, we introduced the usual matter-modified mixing angle and squared-mass difference,
\begin{equation} \label{peta}
\begin{cases}
	\sin2\widetilde{\theta_{13}}=\sin2{\theta_{13}}\ / \Delta\\
\cos2\widetilde{\theta_{13}}=(\cos2{\theta_{13}} -\varepsilon) / \Delta\\\widetilde{\Delta m^2}={\Delta m^2}\times \Delta 
\end{cases}
\end{equation}
where,
\begin{equation}
\Delta=\pm \sqrt{ (\cos2\theta_{13}-\varepsilon)^2 + \sin^22\theta_{13}} 
\end{equation}
The sign of $\Delta$ is  
matter of convention;
the ratio between matter and vacuum term is,
\begin{equation}\label{msw}
\varepsilon\equiv \pm \frac{\sqrt{2} G_{\mathrm{F}} n_{\ee} }{\Delta m^2/(2 E)}
\approx \pm \frac{\rho}{\SI{5.5}{\gram\per\cubic\centi\meter}}\times
\frac{Y_{\ee}}{1/2}\times 
\frac{\SI{2.4e-3}{\electronvolt\squared}}{\Delta m^2}\times
\frac{E}{\SI{5.5}{\giga\electronvolt}}
\end{equation}
where $G_\mathrm{F}$ is the Fermi coupling 
and we identify $\Delta m^2$ with $\Delta m^2_{23}$.
Now, instead, the sign is important: it is plus for normal hierarchy and minus for inverted hierarchy.
Considering the average matter density of the Earth
$\rho=\SI{5.5}{\gram\per\cubic\centi\meter}$ and $Y_{\ee}=1/2$, we get 
$n_{\ee}= 1.7\times 10^{24}\Pelectron\mathrm{cm}^{-3}$ for  the electronic density.
Thus, the 
characteristic length of MSW theory is,
\begin{equation}
L_* \equiv \frac{1}{\sqrt{2} G_{\mathrm{F}} n_{\ee}}\sim \SI{1000}{\kilo\meter}
\end{equation}
We see that, for normal hierarchy, the maximum of $P_{\mu \ee}$ 
 is obtained when:
1.~$\Delta$ is as small as possible, in order to maximize 
 $\sin2\widetilde{\theta_{13}}$; moreover
2.~the phase of propagation is $\widetilde{\varphi} \sim \pi/2$.
These conditions are met when the neutrino energy and the propagation distance are,
\begin{equation}
E_{\mathrm{max}}=\frac{\Delta m^2 L_*}{2}\cos2\theta_{13}\quad\text{and}\quad
L_{\mathrm{max}}=\frac{\pi L_*}{\tan2\theta_{13}}
\end{equation}
or, numerically,
\begin{equation}
\label{masul}
	E_{\mathrm{max}}\sim \SI{5.5}{\giga\electronvolt}
	\quad\text{and}\quad
	L_{\mathrm{max}}\sim \SI{9000}{\kilo\meter}
\end{equation}
In the case of inverted hierarchy, the matter effect depresses $P_{\mu \ee}$, that becomes negligible.

As we have told, one would like to 
observe the matter effect.
From the point of view of the experiments, one direct approach is to 
check the simple prediction concerning 
$P_{\mu \ee}$. 
However, there are some practical advantages 
to detect muons rather
than electrons. Then, let us consider the  survival probability $P_{\mu \mu}$, focussing again
on the normal hierarchy case.  
Let us discuss the case when the first 
local maximum of $P_{\mu \mu}$ is as small as possible.
This results  from $P_{\mu \tau}$ and from $P_{\mu \ee}$, since 
$P_{\mu \mu}=1-P_{\mu \tau}- P_{\mu \mathrm{e}}$. 
We are interested in the case when the minimum of $P_{\mu \tau}$ happens in the vicinity of the energy identified in \eqref{masul}. When the phase of oscillation of $P_{\mu\tau}$ is close to the vacuum phase, the condition $\Delta m^2 L/(2 E_{\mathrm{max}})=2\pi$ gives  $L\sim \SI{6000}{\kilo\meter}$. This suggests that the distance that amplifies the matter effect on $P_{\mu\mu}$ is between 6000 and \SI{9000}{\kilo\meter}, as confirmed by more complete analyses.
These points are illustrated in \figurename~\ref{7k}, that was obtained using the web interface mentioned in \sectionname~\ref{sez:rmk}.

\subsubsection{Description of atmospheric neutrino oscillations 
at various levels of accuracy}  
The general case requires a numerical solution based on \eqref{cagon}. However, 
when the ``solar'' $\Delta m^2_{12}$ is set to zero --- i.e., when its effects
 are negligible --- the only non-zero out-of-diagonal elements $u_{ij}$ in \eqref{cagon}
 are $u_{13}$ and $u_{31}$.
The CP violating phase $\delta$ drops out from the 
probabilities $P_{\ell\ell'}=|\mathcal{U}_{\ell'\ell}|^2$, that moreover  
becomes symmetric, $P_{\ell\ell'}=P_{\ell'\ell}$ for each $\ell,\ell'=e,\mu,\tau$.
Therefore, in this approximation 
we have 3 independent probabilities and all the other ones are fixed. 
We can choose, e.g., 
\begin{equation} \label{3e}
\begin{cases}
P_{\ee\mu}=\sin^2\theta_{23} |u_{13}|^2\\ 
P_{\ee\tau}=\cos^2\theta_{23} |u_{13}|^2\\ 
P_{\mu\tau}=\sin^2 \theta_{23} \cos^2 \theta_{23} |u_{33}-u_{22}|^2
\end{cases}
\end{equation}
so that, e.g., $P_{\ee\ee}=1- P_{\mu \ee}-P_{\tau \ee}=|u_{11}|^2$. From these formulae we obtain,
 \begin{equation} \label{3eb}
P_{\mu \ee}=\sin^2\theta_{23}\ (1-P_{\ee\ee})\quad\text{and}\quad
P_{\mu\tau}=\frac{1}{4}\sin^2 2 \theta_{23} \left|1-\sqrt{P_{\ee\ee}}\ e^{i\hat\varphi}\right|^2
\end{equation}
where $\hat\varphi$ is a (rapidly varying) phase factor. 
Two important remarks are in order: 
\begin{enumerate}[itemsep=-0.7ex,partopsep=1ex,parsep=1ex]
\item The last equation shows that  $P_{\mu \ee}$ is large in the region 
where $P_{\ee\ee}$ is small, and that $P_{\mu\tau}$ remains close to zero 
in the first non-trivial minimum near $\hat\varphi=2\pi$, even when $P_{\ee\ee}\approx 0.3-0.4$ due to matter effect. 
\item The sign of 
$\Delta m^2$ controls the sign of the vacuum hamiltonian; thus, switching between the two 
mass hierarchies or switching between neutrinos and antineutrinos has the same effect; e.g.,
$P_{\ee\mu}(\text{IH})=P_{\bar{\ee}\bar{\mu}}(\text{NH})$.
\end{enumerate}
The first remark is consistent with our numerical findings, that $P_{\mu \ee}$ is amplified and 
 $P_{\mu\tau}$ does not deviate strongly from its behavior in vacuum 
 in the conditions that are relevant for our discussion.

Proceeding further with the approximations, and considering at this point 
the case of constant matter density, we obtain simple and closed expressions. For the case of  
 normal mass hierarchy, they read,
 \begin{equation} \label{3e2}
 \begin{cases}
 	u_{13}=u_{31}=-i \sin\widetilde{\varphi}  \sin2 \widetilde{\theta_{13}}\\
 	u_{11}=\cos\widetilde{\varphi} + i \sin\widetilde{\varphi} \cos2 \widetilde{\theta_{13}}=u_{33}^*\\
 	u_{22}=\cos\widetilde{\varphi}' + i \sin\widetilde{\varphi}\,' 
 \end{cases}
 \end{equation}
 where,
\begin{equation}\widetilde{\varphi}\,'=\frac{{\Delta m^2}  L}{4 E} (1+\varepsilon)\end{equation}
From  \eqref{3e} and \eqref{3e2}, we recover the expression of \eqref{aq}, used in the above discussion.
In the approximation of constant matter density, the phase $\hat\varphi$ entering 
the expression of the probability $P_{\mu\tau}$ 
is given by $\sqrt{P_{\ee\ee}}\cos\hat\varphi\equiv \cos\widetilde\varphi \cos\widetilde\varphi\,'-
 \sin\widetilde\varphi \sin\widetilde\varphi\,'\cos 2\widetilde{\theta_{13} }$. This is close to 
 the vacuum phase when 
 $\varepsilon$ is large or small in comparison to 1: in fact,  we have
 $\cos 2\widetilde{\theta_{13} }\sim \pm 1$ and  
 $\widetilde\varphi\sim \pm \Delta m^2 L/(4 E) (1-\varepsilon)$ from 
 \eqref{peta}, 
 so that $\cos\hat\varphi\sim 
 \cos[\Delta m^2 L/(2 E)]$.

\subsection{Solar neutrinos}
Let us discuss how the adiabatic solution applies to solar neutrinos. 
First of all, we note that the expected matter density profile is approximately exponential, $n_{\ee}\propto \exp(-10 x/R_\odot )$, where  $R_\odot=\SI{7e8}{\meter}$
is the solar radius.  Considering the mass difference squared, 
$\Delta m^2_{12}=\SI{8e-5}{\electronvolt\squared}$ and $\sin^2 2\theta_{12}\sim 0.8$, 
\eqref{ader} reads,
\begin{equation}
\gamma\approx  \left| \frac{10\ c_2}{k R_\odot\ s_2^2}\right|\ll 1 
\end{equation}
for any possible relevant energy  $E=0.1-\SI{20}{\mega\electronvolt}$. Thus, the local mass eigenstates
$n_1$ and $n_2$ remain always such, just as described in \eqref{nada}; therefore, the result of \eqref{adeb} follows. 

At this point we can interpret solar neutrinos.
Let us begin by noting that we expect, {\em a priori}, that
matter effects are important only for solar neutrinos of high energy, 
see \eqref{100}. Given the density of electrons in the center of the Sun, \SI{100}{\mole\per\cubic\centi\meter}, the mixing angle in matter is about $90^\circ$, see \eqref{icasi}. Thus, since $s_m\approx 1$ and  $c_m\approx 0$, the adiabatic formula for  
$P_{\ee\ee}$ simply becomes,
\begin{equation}
P_{\ee\ee}=\sin^2\theta_{12}
\end{equation}
and there is no problem to explain the solar neutrino observations that indicate, at high energies, that the suppression factor of electron neutrinos is about  1/3. At this point, we know that 
the right solution is the one with $\theta_{12}\sim 30^\circ$ 
(since if we choose the other one that is allowed by vacuum neutrino oscillations, with $\theta_{12}\sim 60^\circ$, we would get 2/3 instead). Evidently, at low energies we recover 
the vacuum oscillation formula, 
\begin{equation}
P_{\ee\ee}=1- \frac{\sin^22 \theta_{12}}{2}
\end{equation}
In fact, 
making reference to the cases of \eqref{icasi}, the limit 
$\theta\sim\theta_m$ applies. 

\subsection{Supernova neutrinos\label{sccsn}}

Neutrinos from the gravitational collapse have energies 
from a few MeV to almost \SI{100}{\mega\electronvolt}, propagate at distances of the order of kpc and are produced in nuclear density regions, billion times larger than solar ones. They certainly undergo some oscillation effects. However, it is not clear, to date, to what extent
it is possible to separate the effects of the oscillations from the astrophysical uncertainties since the latter ones are very large.
We will limit ourselves   
to show here (without deriving them) a few basic  formulas
in view of the fact that this type of 
neutrinos do exist and are detectable.

First of all, we calculate the neutrino masses from 
the hamiltonian of propagation given in \eqref{por} as a function of the electron density. These masses depend upon the neutrino masses in vacuum, upon the energy of the neutrino, upon the electron density and upon the mixing angles $\theta_{12}$ and $\theta_{13}$ (see discussion in \eqref{cagon}). 
The result for the energy of a typical supernova neutrino is given in \figurename~\ref{lc} assuming the case of normal mass hierarchy. We see that if the density of electrons is large enough, we will  have two ``level crossing'' or ``resonances''. While in the case of solar neutrinos we had only one (of course, the one occurring for lowest densities), in the case of supernova we will have both of them, since in the center of the star (where the neutrinos are produced) 
we reach nuclear densities. 

Then, we consider a generalization of the 
2 flavors discussion of neutrino oscillations, 
given in \sectionname~\ref{solr}.
For generic values of the mixing angle $\theta$ and of the 
corresponding $\Delta m^2$, the  
matrix containing the probabilities of survival or of conversion, averaged over the large phases, is,
\begin{equation}
P=
\left(
\begin{array}{cc}
c^2 & s^2 \\
s^2 & c^2 
\end{array}
\right)
\left(
\begin{array}{cc}
1-P_{\mathrm{f}} & P_{\mathrm{f}} \\
P_{\mathrm{f}} & 1-P_{\mathrm{f}}
\end{array}
\right)
\left(
\begin{array}{cc}
c_m^2 & s_m^2 \\
s_m^2 & c_m^2 
\end{array}
\right)
\end{equation}
The probability $P_{\mathrm{f}}$  is called {\em flip probability} and describes the possibility that there is a deviation from the relations given in  \eqref{nada}.  It is easy to verify that, in the limit  $P_{\mathrm{f}}\to 0$, the above formula yields $P_{\ee\ee}$ as in \eqref{adeb}.   
In fact, as we have checked in the case of solar neutrinos, the adiabaticity parameter $\gamma$ is very small, and this boils down in the conclusion that $P_{\mathrm{f}} \to 0$.

\begin{figure}[t]
\begin{center}
\centerline{\includegraphics[width=0.8\textwidth,bb=0 0 800 533]{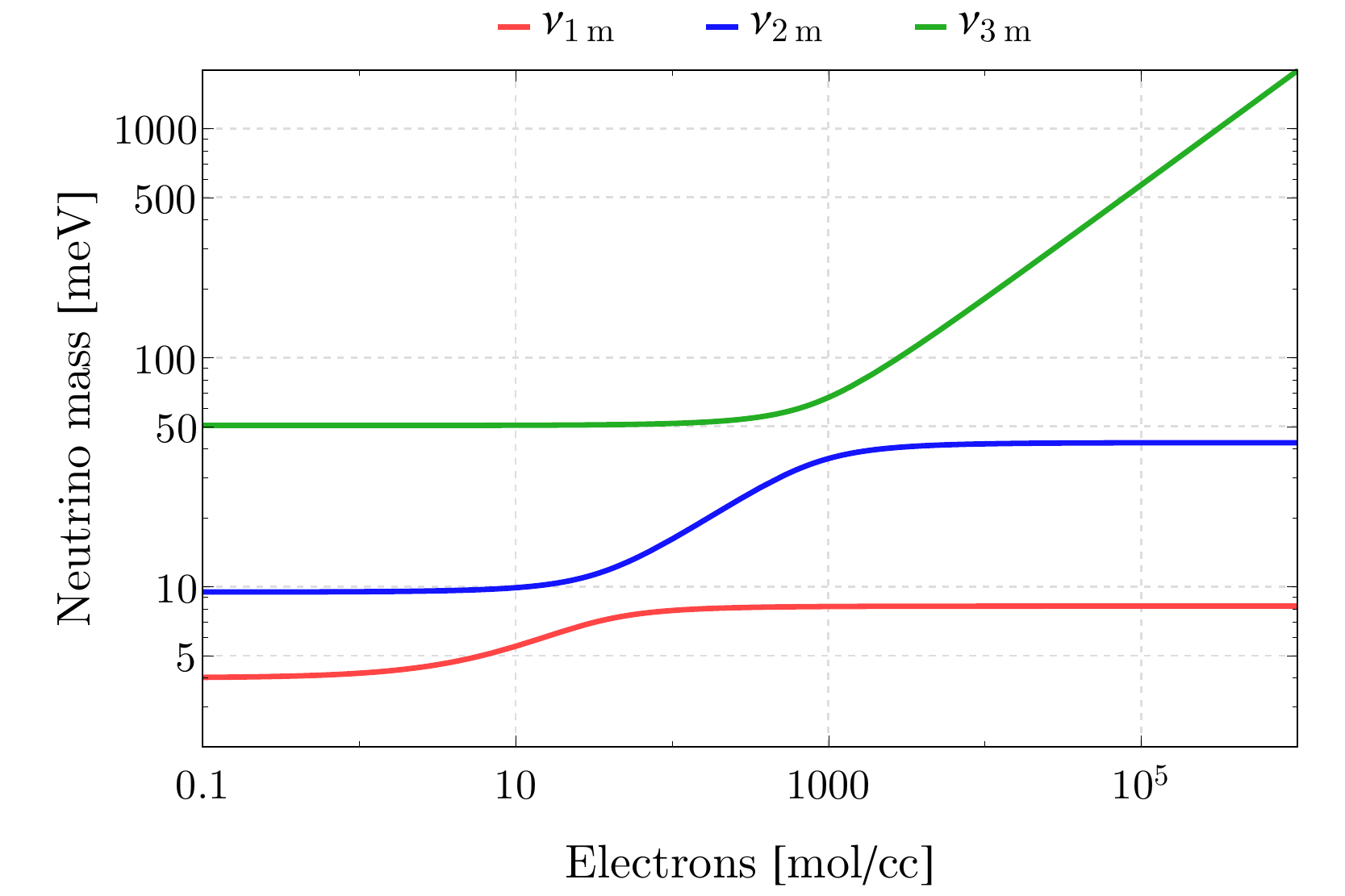}}
\caption{Neutrino masses calculated from the hamiltonian of propagation given in \eqref{por} and given as a function of the electron density for  a typical energy of supernova neutrinos,  $E_\nu=\SI{20}{\mega\electronvolt}$. We 
assumed normal mass hierarchy and fixed 
the lightest neutrino mass to be $m_1=\SI{4}{\milli\electronvolt}$ just for illustration purposes. }
\label{lc}
\end{center}
\end{figure}

The generalization to the 3 flavor case is quite direct, but requires to know
the type of neutrino mass hierarchy. In fact,
we know from the solar neutrino 
oscillations in matter 
that \Pnue  is mostly in $ \nu_1$, 
which corresponds to the choice of the angle $ \theta_{12} $
discussed in the previous paragraph. Conversely, 
we still do not know for sure whether 
the neutrino $ \nu_3$, that is responsible for 
atmospheric oscillations, 
is the heaviest (= direct or normal mass hierarchy) or the lightest (= inverse mass hierarchy), see \figurename~\ref{an:fig.hie}. 
Assuming the first case, that is slightly favored at present 
and it is also the most plausible theoretically, one finds \cite{dighe},  
\begin{equation}
\begin{array}{rl}
P=&\left(
\begin{array}{ccc}
|U_{\ee 1}^2| & |U_{\ee 2}^2| & |U_{\ee 3}^2| \\
|U_{\mu 1}^2| & |U_{\mu 2}^2| & |U_{\mu 3}^2| \\
|U_{\tau 1}^2|  & |U_{\tau 3}^2| & |U_{\tau 3}^2| 
\end{array}
\right)\cdot 
\left(
\begin{array}{ccc}
1-P_{\mathrm{S}} & P_{\mathrm{S}} & 0 \\
P_{\mathrm{S}} & 1-P_{\mathrm{S}} & 0 \\
0 & 0 & 1
\end{array}
\right)\cdot  \\[4ex]
& \left(
\begin{array}{ccc}
1 & 0 & 0 \\
0 & 1-P_{\mathrm{A}} & P_{\mathrm{A}}  \\
0 & P_{\mathrm{A}} & 1-P_{\mathrm{A}}
\end{array}
\right)\cdot  
\left(
\begin{array}{ccc}
|U_{\ee 1,m}^2| & |U_{\mu 1,m}^2| &    |U_{\tau 1,m}^2| \\   
 |U_{\ee 2,m}^2|  & |U_{\mu 2,m}^2| &   |U_{\tau 3,m}^2|   \\
|U_{\ee 3,m}^2|  & |U_{\mu 3,m}^2|  & |U_{\tau 3,m}^2| 
\end{array}
\right)
\end{array}
\end{equation}
We see that now there are 2 flip probabilities. The first one is $P_{\mathrm{A}}$ connected to the ``atmospheric'' parameters (the delta mass squared  $\Delta m^2_{23}$ and $\theta_{13}$) the second is $P_{\mathrm{S}}$ 
connected to the ``solar'' parameters
(the delta mass squared  $\Delta m^2_{12}$ and $\theta_{12}$).
From this expression, we derive the electron survival probability,
\begin{equation}
P_{\ee\ee}= (1-P_{\mathrm{A}}) |U_{\ee 3}^2| + P_{\mathrm{A}} \left[\ (1-P_{\mathrm{S}}) |U_{\ee 2}^2| + P_{\mathrm{S}} |U_{\ee 1}^2|\ \right]
\end{equation}
As in the case for solar neutrinos, we expect that  $P_{\mathrm{S}}\sim 0$ also in the case of the supernova. The same is expected to be true for the other flip probability, $P_{\mathrm{A}}\sim 0$.\footnote{This is expected to be true in all reasonable cases; the only possible (unlikely) exception is if the electron density changes very sharply.} 
Therefore, we conclude,
\begin{equation}
P_{\ee \ee}\sim |U_{\ee 3}^2|\sim 0
\end{equation}
This implies that the electron neutrinos that we detect at Earth were produced at the source as muon or tau neutrinos. 


It is possible to recap this conclusion 
with the help of \figurename~\ref{lc}.
Let us focus on the continuous (green) line,
that shows 
the evolution of the 
heaviest mass eigenstate  $\nu_{3 m}$ 
from the vacuum (lower electronic densities)
to the innermost part of the star
(highest electronic densities).
We note that after a certain density ($n_{\ee}\sim \SI{1000}{\mole\per\cubic\centi\meter}$) 
the mass of $\nu_{3 m}$ 
increases linearly with the electron density (see \figurename~\ref{lc}). 
In fact, the 
matter term in \eqref{por} is the largest part of 
the hamiltonian of propagation 
in the center of the star, it is linear in $n_{\ee}$ 
and evidently 
the heaviest eigenstate of the matter term is just the
electron neutrino. 
This means that an electron neutrino,  produced in the center of the star, is  just the heaviest mass eigenstate  of the full hamiltonian 
in very good approximation,
$|\nu_{\ee}\rangle \approx |\nu_{3m}\rangle$.
The adiabatic evolution implies   
that a neutrino follows the evolution of the mass eigenstate; in our case, the electron neutrino produced in the center of the star 
remains always the heaviest mass eigenstate, 
$|\nu_{\ee}\rangle \approx |\nu_{3m}(t)\rangle$ for all times (up to an irrelevant overall phase)  
and exits from 
the star as $|\nu_3\rangle$. Thus, its survival probability is $|\langle \nu_{\ee}|\nu_3\rangle|^2=|U_{\ee 3}^2|$.

\subsection{Earth matter effect at solar and supernova energies}

At solar neutrino energies, the probability $P_{\nu_2\to \nu_{\ee}}$
given in \eqref{eq:184} admits the simple numerical approximation
\begin{equation}
P_{\nu_2\to \nu_{\ee}}\approx
\displaystyle 0.3 + 1.7\%\ \kappa\
\sin^2\left(\frac{\Delta m^2 L}{4 E}  \right)
\end{equation}
with,
\begin{equation}
\kappa=\left( 
\frac{\SI{7.37e-5}{\electronvolt\squared}}{\Delta m^2}
 \right)\left( 
 \frac{n_{\ee}}{\SI{2}{\mole\per\cubic\centi\meter}}
 \right)\left( 
 \frac{E}{\SI{5}{\mega\electronvolt}} 
 \right)
\end{equation}
Therefore, we see that the detectors receive (slightly) more electron neutrinos  in night than in day, 
which is the reason why this is called  {\em day-night} effect or {\em regeneration} (adopting neutral kaon's terminology). This  feature remains true on average when we consider the full description of the Earth density.

 Something similar is expected to happen in the case of supernova neutrinos. In this case the typical energy is $20-\SI{30}{\mega\electronvolt}$, larger than the one of solar neutrinos; thus, if neutrinos pass through the Earth, the effect is larger.
In the approximation described above (\sectionname~\ref{sccsn}),  when we start from  \Pnue  we  expect to receive on Earth $\nu_3$ (or $\nu_2$) for normal (or inverse) mass hierarchy; instead, 
 \APnue becomes $\bar\nu_1$ (or $\bar\nu_3$).  If we receive 
 $\nu_3$ or $\bar\nu_3$,  the relevant vacuum term (related to the larger value of $\Delta m^2$)
 is too large at these energies and no effect is expected. In the other cases, we can use almost immediately the above formula:
in fact,  we have the relation
 $P_{\nu_1\to \pnue}=1-P_{\nu_2\to \pnue}$, 
and we can obtain    $P_{\bar\nu_1\to \apnue}$ from 
$P_{\nu_1\to \pnue}$ by flipping the sign of the matter term:
therefore, it is sufficient to calculate just one of these probabilities and we can 
immediately use the above result. 
 Note the curious fact that 
 $P_{\bar\nu_1\to \apnue}$  can be readily obtained from 
 $P_{\nu_2\to \pnue}$ given above simply exchanging $\theta\to \pi/2-\theta$. 
This allows us to describe \APnue in normal hierarchy and to understand that 
 \APnue are expected to be ``regenerated'' if they reach the detector after passing through the Earth.

 For an exhaustive discussion in connection with solar neutrinos   see Ref. \cite{lisos}; 
 for a study concerning supernova neutrinos see 
Ref.~\cite{marco1}.  

	\addtocontents{toc}{\protect
\setcounter{tocdepth}{0}}
\chapter{Summary and discussion\label{ch:mo2}}
\minitoc

Neutrino oscillations are an important topic of high energy particle physics. With these notes, we have provided the Reader with 
an introduction to the main points and results. 
The usefulness of these results 
has been illustrated by examining specific cases, applications, solutions of the equations of propagation and by discussing various formal developments. Moreover, we have thoroughly 
discussed the conceptual context and collected in the appendices various considerations concerning the description of relativistic fermions, that complement a more-or-less standard knowledge on Dirac matrices from a course on 
high energy/particle physics. 
In this brief summary, we list  the most important basic equations concerning neutrino oscillations and related facts; we overview the status of the field and we remark once again the introductory character of these notes, in the hope of encouraging the Reader to join the research on neutrino physics.

\section{Survival kit}

For the purpose of recollecting the main facts and of allowing one to test the overall understanding of the expounded material, we offer our selection of the basic formulae discussed in the text, that  a Reader is supposed to master at this point,
\begin{enumerate}[itemsep=-0.7ex,partopsep=1ex,parsep=1ex]
\item Definition of the mixing matrix,
\eqref{mixfi} and its role in the expressions
of neutrino and antineutrino states,  \eqref{pera2}.
\item Standard expression of the mixing matrix, 
\eqref{stdd}, values of the parameters,
Tab.~\ref{tab:lisi16}, explanation of the part  that 
matters for oscillations, \eqref{ridif}.  
\item Amplitude of transition for vacuum oscillations, \eqref{bum2}.
\item General expression of the probabilities of transition,
\eqref{cit}.
\item Numerical factor, 1.267, \eqref{notorio}.
\item Matter term, \eqref{mss} and 
its numerical value, \eqref{100} and
\eqref{msw}.
\item Vacuum hamiltonians for the evolution of the states, \eqref{bibolott} and also \eqref{equib}. 
Matter hamiltonians,  \eqref{por}.
\end{enumerate}

\subsection{Neutrino oscillations today}
A detailed knowledge of this phenomenon is crucial to  correctly interpret several results concerning neutrino experiments, either in laboratory or in the context of neutrino astronomy. 
In similar cases, it is necessary to master the formulae expounded here (or possibly their extensions), to dispose of a detailed model  of the source and of the detector, to adopt suitable statistical procedures for data analysis.  
Of course, these formulae are necessary to investigate the value of the neutrino parameters: the hottest topics are the discrimination of neutrino mass hierarchy, a precise measurement of CP violation, the assessment of the deviation of $\theta_{23}$ from $45^\circ$.
Moreover, they are needed to describe the impact of flavor transformations on the expectations, to search for {\em other} neutrinos besides the usual ones\footnote{These are often called `sterile' neutrinos. 
We  have various hints to date, but no  significant evidences or compelling  theoretical indications for them \cite{sterili1,sterili2}.}, to search for new `matter effects', to hunt for (speculative) effects of 
neutrino magnetic moments, etc.
The axiomatic aspects of neutrino oscillations   
are still considered matter of debate and a detailed description of the non relativistic limit on neutrino oscillations is regarded with interest, even if the contribution of these discussions 
to the field has not a crucial importance. 
More  in general, the resemblance of oscillations with other phenomena concerning the light suggests the  interest in investigating the possibilities of analogous phenomena, such as neutrino reflection. The understanding of the flavor transformation of neutrinos in supernovae, where also the neutrino density itself plays a role (and the magnetic field is very high), is a very active area of research. 

\subsection{What else?}
There are several topics that we have not treated. E.g., there are several useful analytical works to make more efficient the calculation of the oscillation effects; there are a lot of solutions that have been worked out and discussed in neutrino oscillations;  furthermore, 
there are a lot of useful points of view and formalisms, not only concerning the advanced topics, but also the basic ones. In this connection, an important and general question is: where to get more information on neutrino oscillations? In view of the introductive character of these notes, we list here only a few  books \cite{bk,fy,gk,mp} and review papers \cite{r-p,r-k,r-e1,r-e2,strumiavissani,r-bs}. 
We offer them to the Reader since we have found them useful, 
but being deeply aware of the incompleteness of this list and  of 
 the subjective character of this choice.
To conclude, we note that neutrino oscillations are, to date, the only successful probe we have of neutrino masses. However, they do not allow us to probe the lightest neutrino mass, or to investigate the nature of the mass (either Dirac or Majorana). Likewise, the evidence of neutrino mass raises the question of their meaning in extensions of the standard model of the electroweak particles; etc.
	
	\appendix
	\begin{appendices}
\chapter{Special results on Dirac matrices\label{ch:app}}
\minitoc

\section[All Dirac matrices are equivalent (Pauli theorem)]{All Dirac matrices are equivalent\label{pat}}
\subsection{Gamma matrices in physics}
In physics, the Dirac $\gamma_\mu$ matrices are labeled by  $\mu=0,1,2,3$: 
$\gamma_0$ is hermitian while the other three are 
antihermitian. They obey the $\gamma_\mu \gamma_\nu + \gamma_\nu \gamma_\mu = 2 \eta_{\mu \nu} 
\mathbb{1}$, where 
$\eta=\mathrm{diag}(1,-1,-1,-1)$. Thus, 
a one-to-one connection with a set of hermitian $
\gamma_a$ matrices (with $a = 1,2,3,4$) is
possible, by setting 
$\gamma_4\equiv\gamma_0$ and multiplying
$\gamma_{1,2,3}$ by a factor $i$ (or $-i$).
They obey the same algebra where $\eta_{\mu\nu}$ is
replaced by the Kronecker $\delta_{ab}$.

\begin{theorem}
Any set of 4 hermitian matrices, that satisfy the algebra,
\begin{equation}
\gamma_a \gamma_b + \gamma_b \gamma_a =2 \delta_{ab} \mathbb{1}
\quad\text{where}\quad a,b=1,2,3,4
\end{equation}
is unitary equivalent, in the sense that being $ \gamma_a $ and $ \gamma'_a $ two different sets, it exists a unitary matrix $ U $ with,
\begin{equation} \gamma_a' = U \gamma_a U^\dagger \end{equation} 
\end{theorem}
\begin{proof}
There are various proofs\footnote{E.g.,  
use a result of finite group theory: {\em given a 
group of order $N_{\mathrm{G}}$, the dimensions of all its irreducible representations $d_i$ satisfy
$\sum_{i\ge 1} d^2_i=N_{\mathrm{G}}$.}
Consider the  group of order $N_{\mathrm{G}}=32$ with elements
$1,\gamma_1,\gamma_2...,\gamma_1\gamma_2,$
$ ....,-1,-\gamma_1...., -\gamma_1\gamma_2\gamma_3\gamma_4$.  
Note that we have one 
4-dimensional representation and also the trivial one.  
Thus, $32=4^2+ 1+\sum_{i\ge 3} d_i^2$, and 
there is no room for other representations with $d_i\ge 4.$} and
we choose an inelegant one, that however has the merit of being very direct and of using
elementary and important considerations.
Before starting, let us stress that, if $\gamma_a$ are a set of such matrices, any unitary matrix $U$ defines another valid set of matrices $\gamma' \equiv U \gamma U^\dagger$. Then it is enough to prove that a such set of matrices exists and that can be written in a universal form by a choice of the basis, i.e., by a unitary transformation.

\begin{enumerate}[itemsep=-0.7ex,partopsep=1ex,parsep=1ex]
\item The first step is to look for a set of $\gamma_a^\dagger = \gamma_a$ matrices and note that $\gamma_a^2=\gamma_a^\dagger \gamma_a = \mathbb{1}$ implies that the matrices we are looking for are also unitary. 
Moreover, their eigenvalues are $\pm 1$. Now, 
\begin{equation}
\gamma_1=-\gamma_2 \gamma_1\gamma_2 \;\Rightarrow\; \mathrm{tr}\left(\gamma_1\right)=-\mathrm{tr}\left(\gamma_1 \gamma_2^2\right)
\;\Rightarrow\; \mathrm{tr}\left(\gamma_1\right) = 0
\end{equation}
thus, we see that the dimension of the space should be even\footnote{Please recall that $\mathrm{tr}(abc) = \mathrm{tr}(bca)$ and $h^\dagger=h \Rightarrow \mathrm{tr}(h) = \sum_i \lambda_i$ where $\lambda_i$ are the eigenvalues of $h$.}. Now, it is impossible to realize this algebra in $d=2$ dimensions: 3 matrices could be the Pauli ones, but the fourth one cannot be realized.
So, we need at least $d=4$ dimensions. 
\item We assume $d=4$ and choose the first matrix to be diagonal,
\begin{equation}
\gamma_1=\mathrm{diag}(1,1,-1,-1)
\end{equation}
By imposing the anti-commutation condition, the other three hermitian matrices are of the form,
\begin{equation}
\gamma_i=\left( \begin{array}{cc} 0 & a_i \\ a_i^\dagger & 0 \end{array} \right)\quad i=2,3,4
\end{equation}
and from $\gamma_i^2=1$ we find that the $2\times 2$ matrices $a_i$ are unitary. Of course these are 3 different $a_i$; 
in the rest of the proof, we will fix the matrices $a_i$ for 
$\gamma_{2,3,4}$.  
\item We are free to change basis, without changing the form of $\gamma_1$, using the subset of transformations,
\begin{equation}
\left(\begin{array}{cc} u & 0 \\  0 & v \end{array} \right) 
\end{equation}
where $u$ and $v$  are unitary and $2\times 2$.  
Under this assumption the $a_i$  matrices transform as, 
\begin{equation} a_i\to u\, a_i\, v^\dagger\end{equation}
If we choose $u=v$, we can make one $a_i$ diagonal, say, $a_2=\mathrm{diag}(e^{i\alpha}, e^{i\beta})$. 
Using at this point  
$u=\mathbb{1}$ and $v=\mathrm{diag}(e^{i\alpha}, e^{i\beta})=a_2$ we obtain,
\begin{equation}\gamma_2=\left( \begin{array}{cc} 0 & \mathbb{1} \\  \mathbb{1} & 0 \end{array} \right) \end{equation}
Thus, we have set the matrix $\gamma_{2}$ in a standard form.
Note that with this position we retain the invariance
of $\gamma_1$ and $\gamma_2$ under any matrices $u=v$, that we use from here on. 
\item When we impose the anti-commutation conditions of $\gamma_{3,4}$ with $\gamma_2$, we find, 
\begin{equation}
a_i+a_i^\dagger=0\quad i=3,4
\end{equation}
Therefore, we have that the last two $a_i$ matrices are antihermitian and 
satisfy,
\begin{equation}
a_i=i \ \vec{ m }_i \cdot \vec{\sigma} = i\,\sum^3_{j=1} (\vec{m}_i)_j \,\sigma_j \quad\text{where}\quad\vec{m}_i\in \mathbb{R}^3
\quad\text{and}\quad
\vec{m}_i^2=1
\end{equation}
where  $\sigma_i$ are the Pauli matrices and the last condition
follows from unitarity.   
A change of basis,
\begin{equation}
u=e^{i \frac{\vec{\alpha}\cdot \vec{\sigma}}{2} }=\cos(\alpha/2) + i \,\vec{n}\cdot\vec{\sigma}\;\sin(\alpha/2)
\end{equation}
where,
\begin{equation}
\alpha=\left| \vec{\alpha} \right|
\quad\text{and}\quad
\vec{n}=\vec{\alpha}/\alpha
\end{equation}
works as a rotation acting on $\vec{m}$ and 
allows us to choose $a_3=i\sigma_1$ that defines $\gamma_3$. At this point, we 
have still a freedom to change the basis without changing $\gamma_{1,2,3}$, by choosing 
$u=\exp[ i \,{\beta} {\sigma}_1 / {2}] $.
In this manner, we can set $a_4=i\sigma_2$, fixing the form of $\gamma_4$.  
\end{enumerate}
\end{proof}

Summarizing, we have shown that 
given a set of hermitian matrices $\gamma$,  
that obey the anti-commutative algebra, 
it is always possible to choose a suitable basis (i.e., to perform a unitary transformation) so that they are cast in a {\em universal}  form: 
$U\gamma\, U^\dagger=\gamma_{\mathrm{univ}}$. 
This implies that any such set of matrices is 
equivalent to each other.

\subsection{Majorana representation}
When we consider 
the above universal set of matrices and when we 
reinsert the factors $i$, we are free to set
$\gamma_0'=\gamma_3$, 
$\gamma_1'=i\gamma_1$, 
$\gamma_2'=i\gamma_2$, 
$\gamma_3'=i\gamma_4$, which are all imaginary $(\gamma')^*=-\gamma'$. This set of matrices was introduced by 
Majorana and simplifies several equations. Note that the hermiticity property is generic and invariant under a unitary change of basis, while the feature $(\gamma')^*=-\gamma'$ is a characteristic of Majorana basis and rather specific.



\section{Charge conjugation matrix\label{app:ccm}}

An application of the Pauli theorem of great practical importance is the following one. 
Consider a given set of Dirac matrices $\gamma_\mu$. Since the opposite of their transposes,
$-\gamma_\mu^{\mathrm{t}}$, satisfies the same algebra, we know that there is a unitary matrix $C$ such that,
\begin{equation}
 \gamma_\mu^{\mathrm{t}}  = -C^{\dagger}\gamma_\mu C
\end{equation}
From this property, we can deduce several useful forms, as,
\begin{equation}\label{conj}
\gamma C^{\mathrm{t}}=  - C^{\mathrm{t}} \gamma^{\mathrm{t}} 
\quad\text{or}\quad     
   \gamma^{\mathrm{t}}  C^{\dagger}= -  C^{\dagger} \gamma 
\end{equation}
The unitary matrix $C$ is called {\em charge conjugation matrix}, since it transforms particle into antiparticle and viceversa, without affecting the polarization, according to,
\begin{equation} \label{piriw}
\psi^{\mathrm{c}}=C(\bar{\psi})^{\mathrm{t}}
\end{equation}
(see also the discussion after \eqref{abaro1} and \eqref{abaro2}).
Here we remark an important relation that involves $C$ and the chiral projectors, that allows us to handle the weak leptonic charged currents. This is,
\begin{equation}
P_L\psi^{\mathrm{c}}=C(\overline{\psi_{\mathrm{R}}})^{\mathrm{t}}
\end{equation}
where $P_{\mathrm{L}} = (1-\gamma_5)/2$ and $\gamma^5=i\gamma^0\gamma^1\gamma^2\gamma^3$.
The proof of this equality\footnote{The complete proof is,
$
	P_{\mathrm{L}}C(\bar{\psi})^{\mathrm{t}} =CP_{\mathrm{L}}^{\mathrm{t}}(\bar{\psi})^{\mathrm{t}}=C(\bar{\psi}P_{\mathrm{L}})^{\mathrm{t}}=C(\psi^\dagger\gamma_0P_{\mathrm{L}})^{\mathrm{t}}=C(\psi^\dagger P_{\mathrm{R}}\gamma_0)^{\mathrm{t}}=C((P_{\mathrm{R}}^\dagger\psi)^\dagger\gamma^0)^{\mathrm{t}}=C((P_{\mathrm{R}}\psi)^\dagger\gamma^0)^{\mathrm{t}}=C(\psi_{\mathrm{R}}^\dagger\gamma^0)^{\mathrm{t}}=C(\overline{\psi_{\mathrm{R}}})^{\mathrm{t}}
	$.
} follows when one observes that:
\begin{enumerate}[itemsep=-0.7ex,partopsep=1ex,parsep=1ex]
	\item
	$P_{\mathrm{L}}C=CP_{\mathrm{L}}^{\mathrm{t}}$ that is due to the definition of $\gamma^5$;
	\item
	$\gamma^0P_{\mathrm{L}}=P_{\mathrm{R}}\gamma^0$, since $\gamma_\mu\gamma_5=-\gamma_5\gamma_\mu$; and
	\item
	$\gamma_5$ is hermitian, that implies $P_{\mathrm{L}}^\dagger=P_{\mathrm{L}}$.
\end{enumerate}
 Note that this relation, in a more readable form, is,
\begin{equation}
(\psi^{\mathrm{c}})_{\mathrm{L}}=(\psi_{\mathrm{R}})^{\mathrm{c}}
\end{equation}

\subsection[Properties of the $C$-matrix]{Properties of the \boldmath$C$-matrix} 
By using the identities in 
\eqref{conj}, one can show that
the matrix $C$ is antisymmetric. This requires a bit of effort:
\begin{enumerate}[itemsep=-0.7ex,partopsep=1ex,parsep=1ex]
	\item The first step is to note that,
$
\gamma\, C^{\mathrm{t}} C^{\dagger} =- C^{\mathrm{t}}\,\gamma^{\mathrm{t}}\,  C^{\dagger} = C^{\mathrm{t}} C^{\dagger}\,\gamma$, namely 
$[  C^{\mathrm{t}}   C^{\dagger}, \;\gamma ] =0$. The $\gamma$-matrices generate an irreducible representation of the Lorentz algebra, the spinorial representation.
Thus, for the Schur lemma it follows 
that $ C^{\mathrm{t}}  C^{\dagger}$ is proportional to the identity,
$C^{\mathrm{t}}=\kappa C$.
\item The next step is to take 
again the transpose, getting,
$C=(C^{\mathrm{t}})^{\mathrm{t}}=\kappa^2 C$, thus $\kappa=\pm 1$.
\item The last step is the boring one. Consider the symmetry properties of 16 different  combinations, we have,
$C^{\mathrm{t}}=\kappa C$,  
$( \gamma_\mu C)^{\mathrm{t}}= -\kappa  \gamma_\mu C$,
$( [\gamma_\mu, \gamma_\nu]  C)^{\mathrm{t}}=- \kappa  [\gamma_\mu, \gamma_\nu] C$,
$( \gamma_5 \gamma_\mu  C)^{\mathrm{t}}=\kappa \gamma_5 \gamma_\mu C$,
$( \gamma_5 C)^{\mathrm{t}}= \kappa \gamma_5 C$: there are 10 matrices 
with symmetry property $-\kappa$ and 6 matrices with symmetry property $+\kappa$. 
Now, a generic complex $4\times 4$ matrix has 16 parameters. The 16 combinations introduced above are a complete basis for this space. Since,
as known, a complex $4\times4$ symmetric matrix has 10 parameters, while the antisymmetric one has only 6, we conclude that the case that occurs is $\kappa=-1$.
\end{enumerate}

 Summarizing,
\begin{equation}
\begin{cases}
	C,\, \gamma_5 C,\, \gamma_5 \gamma_\mu  C&
	\text{are antisymmetric matrices}\\
	 \gamma_\mu C,\,  [\gamma_\mu,\gamma_\nu] C
	& \text{are symmetric matrices}
\end{cases}
\end{equation}
Notice the interesting consequence that, given a single (anti-commuting) field $\chi$, we can form only a few non-zero covariant quantities,\footnote{In conventional notation~\eqref{piriw} the matrix $C$ accompanies Dirac conjugate spinors $\bar{\psi}$, and conversely, the matrix $C^\dagger$ appears with the spinors $\psi$. Of course, relations similar to those derived in this section apply to the $C^\dagger$ matrix.}
namely, 
$\bar{\chi}  C \bar{\chi}^{\mathrm{t}} $,  
$\bar{\chi}  \gamma_5 C \bar{\chi}^{\mathrm{t}} $, 
$\bar{\chi}  \gamma_5 \gamma_\mu C \bar{\chi}^{\mathrm{t}} $, 
since, if $M$ is a symmetric matrix, 
$\bar\chi^{\mathrm{t}}M\bar\chi\equiv
\bar\chi_aM_{ab}\bar\chi_b=
\bar\chi_a\bar\chi_bM_{ab}=
-\bar\chi_b\bar\chi_aM_{ab}=
-\bar\chi_bM_{ba}\bar\chi_a=
-\bar\chi^{\mathrm{t}}M\bar\chi=0$.

%
%

This implies, for instance, that it is not possible to have the ordinary couplings of a single field $\chi$ to the photon field 
$A^\mu$ and to the electromagnetic tensor $F^{\mu\nu}$. 
In other words, a self-conjugate field 
$\chi=\chi^c$ has neither charge nor  
magnetic moment, 
while we can have non-zero operators using two different
such fields $\bar\chi_1,\bar\chi_2$, since the symmetry properties would imply 
$\bar\chi^{\mathrm{t}}_1M\bar\chi_2\equiv
\bar\chi_{1,a}M_{ab}\bar\chi_{2,b}=
-\bar\chi_{2,b}M_{ba}\bar\chi_{1,a}=
-\bar\chi^{\mathrm{t}}_2M\bar\chi_1$, that in general is non vanishing.


\subsection{Discrete symmetries of the Dirac equation}
Let us consider the Dirac equation in an external electromagnetic field $A_\mu(x)$,
\begin{equation}
\left[ \hat{p} -e \hat{A}(x) -m \right] \psi(x)=0
\end{equation}
where of course 
$p^\mu=i \partial/\partial x_\mu=(i\partial_t,-i\vec{\nabla})$, $m$ is the mass (a real parameter), 
and the contractions with the gamma matrices are indicated 
by $\hat{A}=\gamma_\mu A^\mu$.
This equation is covariant under proper Lorentz transformation but also under various discrete symmetries. In other words, by transforming the space-time and the fields appropriately, as listed in \tablename~\ref{tab:appendi}, we map valid solutions into valid solutions. 

\begin{table}[htb]
\centering
	{\begin{tabular}{@{}cr@{\hskip3pt}lr@{\hskip3pt}lr@{\hskip3pt}l@{}}
	\toprule
		& \multicolumn{2}{c}{Coordinates} & \multicolumn{2}{c}{EM field}& \multicolumn{2}{c}{Fermionic field}\\  
		\midrule
		 P & $x_{\text{\tiny P}}=$&$(t, -\vec{x}$\,) & $A_{\text{\tiny P}}(x)=$&$( A^0 (x_{\text{\tiny P}}), -A^i(x_{\text{\tiny P}}))$ &$\psi_{\text{\tiny P}}(x)=$&$\gamma_0 \psi(x_{\text{\tiny P}})$\\
		 C & $x_{\text{\tiny C}}=$&$(t, \vec{x}$\,) & $A_{\text{\tiny C}}(x)=$&$(-A^0 (x_{\text{\tiny C}}), -A^i(x_{\text{\tiny C}}))$ & $\psi_{\text{\tiny C}}(x)=$&$C \gamma^t_0 \psi^*(x_{\text{\tiny C}})$\\
		T & $x_{\text{\tiny T}}=$&$(-t, \vec{x}$\,) & $A_{\text{\tiny T}}(x)=$&$( A^0 (x_{\text{\tiny T}}), -A^i(x_{\text{\tiny T}}))$&$\psi_{\text{\tiny T}}(x)=$&$\gamma_5 C \psi^*(x_{\text{\tiny T}})$\\
		 \bottomrule
	\end{tabular}}
	\caption{Discrete symmetries of the Dirac equation. P: Spatial parity. C: Charge conjugation. T: Time reversal.}
	\label{tab:appendi}
\end{table}

We note that:
\begin{enumerate}[itemsep=-0.7ex,partopsep=1ex,parsep=1ex]
	\item The transformation of the space-time coordinates justifies the names of $P$ and $T$ (whereas for $C$ this is trivial, namely no transformation is required);
\item the transformations of the electromagnetic field under $P$ and $T$  imply that the vector potential is a polar vector, and it transforms like the velocity thereby changing 
sign under $T$;
\item for $C$ symmetry, the change of the EM\ potential is equivalent to flip the sign of the charge $e\to -e$. 
\end{enumerate}

Let us discuss now the properties of the transformed fields. 
We begin from the {\bfseries spatial parity} $P$. Applying the Dirac differential operator, we have,
\begin{equation}
\left[ \hat{p} -e \hat{A}_{\mathrm{P}}(x) -m \right] \psi_{\mathrm{P}}(x)=\gamma_0 
\left[ \hat{p}_{\mathrm{P}} -e \hat{A}(x_{\mathrm{P}}) -m \right] \psi(x_{\mathrm{P}})
\end{equation}
where $p_{\mathrm{P}}=i \partial/\partial x_{\mathrm{P}}=
(i\partial_t, i\vec{\nabla})$. 
 Changing the names of the variables and more 
precisely replacing $-\vec{x} \to \vec{x}$, we recognize that 
the differential operator in the r.h.s. is just the Dirac one.
Therefore, if the field $\psi(x)$ satisfies Dirac equation 
also the $P$-transformed field does it.

Now we repeat the same steps for the {\bfseries charge conjugation} $C$. 
Using the well-known property 
$\gamma_\mu \gamma_0=\gamma_0 \gamma_\mu^\dagger$, we have,
\begin{equation}
\begin{aligned} \label{abaro1}
\left[ \hat{p} -e \hat{A}_{\mathrm{C}}(x) -m \right] \psi_{\mathrm{C}}(x) &=
- \left[ \hat{p} + e \hat{A}(x) -m \right]  \gamma_0 C \psi^*(x)
\\ &=
- \gamma_0 \left[ \hat{p}^\dagger + e \hat{A}^\dagger(x) -m \right]   C \psi^*(x)
\end{aligned}
\end{equation}
where $\hat{p}^\dagger=p^\mu \gamma_\mu^\dagger$ and similarly for $ \hat{A}^\dagger$.
We recall the defining property
$C \gamma^{\mathrm{t}} =C \gamma^{\mathrm{t}}C^{-1}C= -\gamma\:\! C$, that implies 
 $\gamma^{\mathrm{t}} C^\dagger = - C^\dagger \gamma $ and therefore 
 $C \gamma^* = C \gamma^*C^{-1}C=-\gamma^\dagger C$.
 We also notice that $p_a^*=(i\partial_a)^*=-p_a$, getting,
 \begin{equation}
 \begin{aligned}
 \left[ \hat{p} -e \hat{A}_{\mathrm{C}}(x) -m \right] \psi_{\mathrm{C}}(x)
&= - \gamma_0 \left[ \hat{p}^\dagger + e \hat{A}^\dagger(x) -m \right]  C \psi^*(x)\\
&=
 - \gamma_0  C \left( \left[ \hat{p} - e \hat{A}(x) -m \right]   \psi(x) \right)^*
 \end{aligned}
 \label{abaro2}
 \end{equation}
Again, we conclude that if the field $\psi(x)$ satisfies the Dirac equation 
also the transformed configuration does it. Note that the first equation can be interpreted as 
if the charge of the field $e$ is formally inverted into $-e$; this justifies the name of {\em charge conjugation}
given to the matrix $C$.

Finally, we perform the same steps for 
{\bfseries time reversal} $T$. Consider the above defined spinor,
\begin{equation}
\left[ \hat{p} -e \hat{A}_{\mathrm{T}}(x) -m \right] \psi_{\mathrm{T}}(x)=
\gamma_5 C \left[ \hat{p}^{\mathrm{t}} - e \hat{A}^{\mathrm{t}}_{\mathrm{T}} (x_{\mathrm{T}}) -m \right]  \psi^*(x_{\mathrm{T}})
\end{equation}
Using the fact that the matrix $\gamma_0$ is hermitian while 
$\gamma_i$ are antihermitian, we have $\hat{A}_{\mathrm{T}}^{\mathrm{t}}=\hat{A}^*$ and 
$\hat{p}^{\mathrm{t}}=\hat{p}^*_{\mathrm{T}}$  where 
$p_{\mathrm{T}}=(-i\partial_{t} , -i\vec{\nabla})
=(i\partial_{-t} , -i\vec{\nabla})$, thus,
\begin{equation}
\begin{aligned}
\left[ \hat{p} -e \hat{A}_{\mathrm{T}}(x) -m \right] \psi_{\mathrm{T}}(x)&=
\gamma_5 C \left[ \hat{p}^{\mathrm{t}} - e \hat{A}^{\mathrm{t}}_{\mathrm{T}} (x_{\mathrm{T}}) -m \right]  \psi^*(x)\\
&=\gamma_5 C \left( \left[\hat{p}_{\mathrm{T}} - e \hat{A} (x_{\mathrm{T}}) -m \right]  \psi (x_{\mathrm{T}})\right )^*
\end{aligned}
\end{equation} 
and once again we see that modulo the redefinition of the time coordinate, $-t\to t$, the fact that 
$\psi(x)$ is a solution of the Dirac equation implies that also its time reversed conjugate is a solution.

\section[Fierz identity for $(V-A)$ current-current operator]{Fierz identity for \boldmath $(V-A)$ current-current operator\label{frz}}
Let us extends the set of  the three hermitian 
Pauli matrices introducing also,
\begin{equation}\sigma^0_{ab}\equiv \delta_{ab}\end{equation}
The four matrices $\sigma^\alpha$, with $\alpha=0,1,2,3$, form a basis of the 
$2\times 2$ complex matrices. Now, consider the tensor $\sigma^0_{ab}\ \sigma^0_{cd}$; we would like to consider it as a set of matrices
with indices $c$ and $b$. Thus, each of them can be rewritten in terms of the basis of the $2\times 2$ matrices that 
we have introduced by using 4 complex coefficients, namely,
\begin{equation}
\sigma^0_{ab}\ \sigma^0_{cd}=\sum_{\alpha=0}^{3} \kappa_{ad}^\alpha\ \sigma_{cb}^\alpha
\end{equation} 
In order to determine the coefficients $\kappa$, we multiply by $\sigma^\beta_{bc}$ and sum over the repeated coefficients, finding,
\begin{equation}
\sigma^0_{ab}\ \sigma^\beta_{bc}\ \sigma^0_{cd}=\sigma^\beta_{ad} =
\sum_{\alpha=0}^{3} \kappa_{ad}^\alpha\ \mathrm{tr}[ \sigma^\alpha\ \sigma^\beta ]= 2  \kappa_{ad}^\beta 
\end{equation}
Therefore we can replace $k^\beta_{ad}$ with $\sigma^\beta_{ad}/2$ and thus we have proven the first identity,
\begin{equation}
\sigma^0_{ab}\ \sigma^0_{cd} = \frac{1}{2}\left[  \sigma^0_{ad}\ \sigma^0_{cb}  +  
 \sigma^i_{ad}\ \sigma^i_{cb}  
\right]
\end{equation}
where $i=1,2,3$ is summed over. Similarly, we can deduce an identity for the combination,
$\sigma^i_{ab}\ \sigma^i_{cd}$ and finally write the useful identity,
\begin{equation}
\eta_{\alpha\beta}\ \sigma^\alpha_{ab}\ \sigma^\beta_{cd}= 
- \eta_{\alpha\beta}\ \sigma^\alpha_{ad}\ \sigma^\beta_{cb}
\end{equation}
where $\alpha$ and $\beta=0,1,2,3$ are summed over and $\eta$ is the (flat) metric tensor,
\begin{equation}\eta=\mathrm{diag}(+1,-1,-1,-1)\end{equation}

Let us now consider 4 {\em anti-commuting} bi-spinors  
named $\psi_{\mathrm{L}},\chi_{\mathrm{L}},\lambda_{\mathrm{L}},\phi_{\mathrm{L}}$, namely quantized fermionic fields.  
The previous identity implies, 
\begin{equation}
\psi^\dagger_{\mathrm{L}} \bar{\sigma}^\alpha \chi_{\mathrm{L}} \ \lambda^\dagger_{\mathrm{L}} \bar{\sigma}_\alpha \phi_{\mathrm{L}} = 
+ \psi^\dagger_{\mathrm{L}} \bar{\sigma}^\alpha  \phi_{\mathrm{L}} \ \lambda^\dagger_{\mathrm{L}} \bar{\sigma}_\alpha\chi_{\mathrm{L}} 
\end{equation}
where we have introduced the matrices,
\begin{equation}
\bar{\sigma}^\alpha=(\sigma^0,-\sigma^i)
\end{equation}
The proof is based on the previous identity and on the
anti-commutativity  that flips the minus sign into a plus.

The last remark is that in the chiral (or Weyl) representation of the Dirac matrices, where $\gamma_5=\mathrm{diag}(-1,-1,+1,+1)$ and $\gamma_0=\mathrm{antidiag}(1,1,1,1)$,
\begin{equation}
\psi= 
\left(
\begin{array}{c}
\psi_{\mathrm{L}}\\
\psi_{\mathrm{R}}
\end{array}
\right)\quad\text{and also}\quad
\gamma^{\alpha} P_{\mathrm{L}}=
\left(
\begin{array}{cc}
0 & 0 \\
\bar{\sigma}^\alpha & 0
\end{array}
\right)
\end{equation}
and similarly for any other four-spinor.
Therefore the previous identity concerning bi-spinors can be translated into an identity concerning (anti-commuting) 
four-spinors,
\begin{equation}
\bar{\psi}  \gamma^\alpha P_{\mathrm{L}} \chi\;\,\bar{\lambda}\gamma_\alpha P_{\mathrm{L}} \phi = 
\bar{\psi}  \gamma^\alpha P_{\mathrm{L}} \phi\;\,\bar{\lambda}\gamma_\alpha P_{\mathrm{L}} \chi 
\label{eq:fiertzidentity}
\end{equation}
namely,  in this combination of $V-A$ currents, we can simply exchange the 2nd and the 4th four-spinors. Owing to the fact that this is an identity between scalars, we conclude that this is a general fact, valid in {\em any} representation of Dirac matrices.  

This Fierz identity \eqref{eq:fiertzidentity} is useful in various situations, in particular:
\begin{enumerate}[itemsep=-0.7ex,partopsep=1ex,parsep=1ex]
\item In the study of neutrinos, it can be used for the derivation of the matter terms of MSW theory.
\item In the theory of weak interactions, it allows us to cast in the same form the two contributions from CC and NC.
\end{enumerate}
\end{appendices}
	
	\addcontentsline{toc}{chapter}{Bibliography}
	\printbibliography
	
\end{document}